\documentclass[10pt,journal,letterpaper]{IEEEtran} 

\usepackage{multirow}
\usepackage[textwidth=1cm,backgroundcolor=yellow,textsize=tiny]{todonotes}
\newcommand{\itamar}[1] {\todo[color=blue!20,inline]{Itamar: #1}}
\newcommand{\paolo}[1]{\todo[color=green!20,inline]{Paolo: #1}}

\usepackage{amsmath,amsfonts, amsthm, bm}

\usepackage{comment, tagging}

\usetag{long} 

\usepackage{colortbl}
\usepackage{stackengine} 

\usepackage{array}
\newcolumntype{P}[1]{>{\centering\arraybackslash}p{#1}}
\newcolumntype{M}[1]{>{\centering\arraybackslash}m{#1}}
\usepackage{setspace}

\usepackage[compress]{cite}
\usepackage{tabto}
\usepackage[noend]{algpseudocode}
\usepackage{algorithm}
\usepackage{varwidth} 
\usepackage{xspace, paralist}
\PassOptionsToPackage{hyphens}{url}
\usepackage{float}
\usepackage{graphicx, pstricks, pst-node,pst-plot}
\usepackage{subfloat, tabu, color}
\usepackage{subfiles}
\usepackage{subcaption}

\algnewcommand{\AND}{\textbf{and}\xspace}
\algnewcommand\algorithmicforeach{\textbf{for each}}
\algdef{S}[FOR]{ForEach}[1]{\algorithmicforeach\ #1\ \algorithmicdo}

\algrenewcommand\algorithmicindent{1.0em}%

\makeatletter
\algnewcommand{\LineComment}[1]{\Statex \hskip\ALG@thistlm \(\triangleright\) #1}
\makeatother

\newbox\statebox
\newcommand{\myState}[1]{%
    \setbox\statebox=\vbox{#1}%
    \edef\thealgruleheight{\dimexpr \the\ht\statebox+1pt\relax}%
    \edef\thealgruledepth{\dimexpr \the\dp\statebox+1pt\relax}%
    \ifdim\thealgruleheight<.75\baselineskip
        \def\thealgruleheight{\dimexpr .75\baselineskip+1pt\relax}%
    \fi
    \ifdim\thealgruledepth<.25\baselineskip
        \def\thealgruledepth{\dimexpr .25\baselineskip+1pt\relax}%
    \fi
    \State #1%
    \def\thealgruleheight{\dimexpr .75\baselineskip+1pt\relax}%
    \def\thealgruledepth{\dimexpr .25\baselineskip+1pt\relax}%
}

\newtheorem{theorem}{Theorem}
\newtheorem{corollary}[theorem]{Corollary}
\newtheorem{lemma}[theorem]{Lemma}

\newtheorem{proposition}[theorem]{Proposition}

\newtheorem{definition}[theorem]{Definition}

\newcommand{\norm}[1]{\left\Vert#1\right\Vert}
\newcommand{\abs}[1]{\left\vert#1\right\vert}
\newcommand{\set}[1]{\left\{#1\right\}}

\DeclareMathOperator*{\argmin}{arg\,min}
\DeclareMathOperator*{\argmax}{arg\,max}

\usepackage{mathtools}

\DeclarePairedDelimiter\floor{\lfloor}{\rfloor}

\newcommand{\Gc}{\mathcal G}
\newcommand{\Hc}{\mathcal H}

\newcommand{\Lc}{\mathcal L}

\newcommand{\Pc}{\mathcal P}

\newcommand{\Sc}{\mathcal S}

\newcommand{\Tc}{\mathcal T}
\newcommand{\Uc}{\mathcal U}

\newcommand{\av}{\mathbf a}

\newcommand{\hv}{\mathbf h}

\newcommand{\yv}{\mathbf y}

\newcommand{\muv}{\bm \mu}

\newcommand{\muvstar}{\bm \mu^*}

\newcommand{\migProb}{DMP}
\newcommand{\cpuall}{GFA}
\newcommand{\bu}{BU}

\newcommand{\pushUp}{PU}
\newcommand{\algtop}{BUPU}
\newcommand{\opt}{LBound}

\newcommand{\ffit}{F-Fit}
\newcommand{\cpvnf}{CPVNF}

\newcommand{\figWidth}{8 cm}

\date{}
\title{Dynamic Service Provisioning in the Edge-cloud\\ Continuum with Provable Guarantees
}
\author{
\IEEEauthorblockN
{
Itamar Cohen\IEEEauthorrefmark{1},
Carla Fabiana Chiasserini\IEEEauthorrefmark{1},
Paolo Giaccone\IEEEauthorrefmark{1}, and
Gabriel Scalosub\IEEEauthorrefmark{2}
}

\IEEEauthorblockA
{
\IEEEauthorblockA{\IEEEauthorrefmark{1}Department of Electronics and Telecommunications, Politecnico di Torino, Italy}\\
\IEEEauthorblockA{\IEEEauthorrefmark{2}School of Electrical and Computer Engineering, Ben-Gurion University of the Negev, Beer Sheva, Israel}
\\itamar.cohen@polito.it, carla.chiasserini@polito.it, paolo.giaccone@polito.it, sgabriel@bgu.ac.il}

}

\begin{document}
\maketitle
\begin{abstract}

We consider a hierarchical edge-cloud architecture in which services are provided to mobile users as chains of virtual network functions. Each service has specific computation requirements and target delay performance, which require placing the corresponding chain properly  and allocating a suitable amount of computing resources. Furthermore, chain migration may be necessary to meet the services' target delay, or convenient to keep the service provisioning cost  low.  We tackle such issues by formalizing the problem of optimal chain placement and resource allocation in the edge-cloud continuum, taking into account migration, bandwidth, and computation costs.  Specifically,  we first envision an algorithm that, leveraging resource augmentation,  addresses the above problem and provides an upper bound to the amount of resources required to find a feasible solution. We use this algorithm as a building block to devise an  efficient approach targeting the minimum-cost solution, while minimizing the required resource augmentation. Our results, obtained through trace-driven, large-scale  simulations, show that our  solution can provide a feasible solution by using  half the amount of  resources required by  state-of-the-art alternatives.
\end{abstract}

\section{Introduction}
%
{Today's networks offer an unprecedented level of resource virtualization, available as a continuum from the edge to the cloud~\cite{FollowMe_J, MoveWithMe, Justifies_path_to_root_n_CLP_vehs, SFC_mig, Justify_CLP_tree_ISPs, micado_orchestrator, orch_cloud2edge_survey}. 
These virtual resources are embodied as a collection of datacenters that host service function chains.
These service chains provide a plethora of applications, including infotainment~\cite{MoveWithMe}, road safety~\cite{Justifies_path_to_root_n_CLP_vehs} and 
virtual network functions
~\cite{Justify_CLP_tree_ISPs, CPVNF_proactive_place_in_CDN, mig_or_reinstall, APSR}. These applications have versatile service requirements; for instance, a road safety application requires low latency, which may dictate processing it in an edge datacenter, close to the user. On the other hand, infotainment tasks are more  computation-intensive but less latency-sensitive, and therefore may be offloaded to the cloud, where computation resources are abundant and cheap~\cite{tong2016, SFC_mig}. 

Deploying service function chains is even more challenging when dynamic traffic conditions exist and/or some of the users are mobile. 
In such cases, service chains may need to be migrated in order to follow the mobile user and, thus, reduce latency~\cite{FollowMe_J,MoveWithMe, mig_correlated_VMs}. 
However, when the system is highly-loaded, there may not be enough available resources in the migration's destination. Hence, providing reliable service may compel using some over-provisioning, or {\em resource augmentation} -- at the cost of increasing the system's capital expenses. 
Existing schemes~\cite{CPVNF_proactive_place_in_CDN, MoveWithMe, mig_correlated_VMs, Companion_Fog, dynamic_sched_and_reconf_t, SFC_mig, Dynamic_SFC_by_rtng_Dijkstra, mig_or_reinstall, Avatar}
 perform well when the system load is not too high, but fail to provide a feasible solution under a high load of service requests.
To the best of our knowledge, no previous work provides guarantees of finding a feasible solution for the problem whenever such a solution exists. 

In this work, we study the combined service {\em Deployment and Migration Problem} (\migProb)  in a multi-tier network, where the  service orchestrator~\cite{micado_orchestrator} has to decide: 
\begin{inparaenum}[(i)]
\item where to deploy a service chain  across the cloud-edge continuum,
\item which resources to allocate for each part of every service chain, and
\item which chains to migrate, and to which datacenter, to fulfill the service requirements while minimizing the overall deployment and migration costs.
\end{inparaenum}
Our main contributions are as follows:
\begin{itemize}
\item We  first formalize the \migProb, and show that even finding a feasible solution to the problem -- regardless of its cost -- is NP-hard.
\item 
We take latency as the main Key Performance Indicator  (KPI)~\cite{Okpi}, as specified by the Service Level Agreement (SLA), and show how to calculate the minimal amount of CPU resources required for placing every service chain on any datacenter, while satisfying the latency  requirements.
\item We develop a placement algorithm that, leveraging some bounded amount of resource augmentation, is guaranteed to provide a feasible solution whenever such a solution exists for the case with no resource augmentation.
\item We present an algorithm that, given a feasible solution, greedily decreases its cost, while keeping the required resource augmentation minimal.
\item We compare the performance of our proposed solution to those of existing alternatives using two large-scale vehicular scenarios and real-world antenna locations. 
Our results show that our algorithm can provide a feasible solution using half the computing resources required by existing alternatives. Our evaluation further highlights several system trade-offs, such as the preferred decision period between subsequent runs of the algorithm. 
\end{itemize}
}
The rest of the paper is organized as follows. After introducing  the system model in Sec.~\ref{sec:system}, we formalize the optimal deployment and migration problem in Sec.~\ref{sec:problem}, and overview our solution concept in Sec.~\ref{sec:alg_concept}. 
The problem is decomposed into a computational resource allocation problem, studied and solved in Sec.~\ref{sec:alloc}, and a placement problem, characterized and solved in Sec.~\ref{sec:bu}. Our overall algorithmic solution is described in Sec.~\ref{sec:top_lvl} and its performance is assessed in Sec.~\ref{Sec:sim}. Finally, Sec.~\ref{sec:related} discusses related work, and Sec.~\ref{sec:conc} draws some conclusions.

\section{Modeling the edge-cloud architecture}\label{sec:system}

This section introduces the  model for the network infrastructure and the services offered to mobile users and describes how we compute the service delay. 

\subsection{Network model}\label{sec:netw_model}
We consider a fat-tree edge-cloud hierarchical network architecture. As described in~\cite{tong2016}, the network comprises: 
\begin{inparaenum}[(i)]
\item {\em datacenters} (denoting generic computing resources),
\item {\em switches} (generic switching nodes, as routers, switches, multiple switches associated with Multi-Chassis Link-Aggregation (MCLA)~\cite{mcla}), and
\item {\em radio Points of Access (PoA)}. 
\end{inparaenum}
Datacenters are connected through switches, and PoAs may
have a co-located datacenter~\cite{Mig_in_Mobile_Edge_Clouds}. 
Each user is connected to the network through a PoA, which may vary as the user moves.
An example of such a system is depicted in Fig.~\ref{fig:net}. 

We denote by $\Sc$ the set of datacenters, and model the logical multi-tier network as a directed graph ${\Gc} = (\Sc,\Lc)$  
where the vertices are the datacenters, while the edges are the directed   virtual links connecting them, i.e., $(i,j)\in \Lc$ with $i,j \in \Sc$.
Let $D(\Gc)$ denote the {\em diameter} of $\Gc$, and $r$ the root of the fat tree topology.
For any two datacenters $i,j$,  $\Pc(i,j)$ denotes the directed path from $i$ to $j$, with  $\Pc(i,j)$ referring to  a sequence of physical links, or vertices, depending on the context.
We consider that such a path is loop-free and uniquely predetermined between any two vertices. 

\subsection{Services and chain deployment}\label{sec:service_model}
{Consider a generic user generating} a {\em service request} $u$, originating at the PoA $p^u$ to which the user is currently connected. 
Each service request is addressed through an instance of   VNF {\em chains}, where each VNF is deployed on a dedicated virtual machine (VM) or container in a datacenter.
For the convenience of presentation, hereinafter we refer to VMs only.
{We refer to the instance of the chain for service request $u$ as $\hv^u=(v^u_1,\ldots,v^u_{h_u})$,  where $h_u$ indicates the number of VMs in $\hv^u$.
Let $\Uc$ denote the set of service requests, and $\Hc$  the set of corresponding chains that are currently deployed, or need to be deployed, in the network.}
Furthermore, for every subset of requests $\Uc' \subseteq \Uc$, we let $\Hc' \subseteq \Hc$ denote the subset of chains corresponding to $\Uc'$.
For simplicity of notation, while referring to VMs and datacenters hosting them, we will drop superscripts and subscripts whenever clear from the context.

To successfully serve chain $\hv^u$,  the chain should be fully deployed on one of the datacenters on the path from its PoA $p^u$ to 
root $r$~\cite{Justifies_path_to_root_n_CLP_IEICE, Justifies_path_to_root_n_CLP_vehs, Justify_CLP_tree_ISPs}. We denote that path by $\Pc(p^u,r)$. 
Distinct deployment decisions incur distinct {\em costs} that we detail in Sec.~\ref{sec:problem}.

Each service is associated with an SLA, which specifies the requirements in terms of KPIs~\cite{Okpi},
and with a maximum amount of resources, e.g., for which the user is willing to pay the network provider.
We consider {\em latency} as the most relevant KPI, although 
our model could be extended to others, like throughput and energy consumption.
We thus associate with each chain $\hv^u$ a {\em target delay} $\Delta(\hv^u)$.

\begin{figure}[!t]
\begin{center}
    \includegraphics[width=0.35\textwidth]{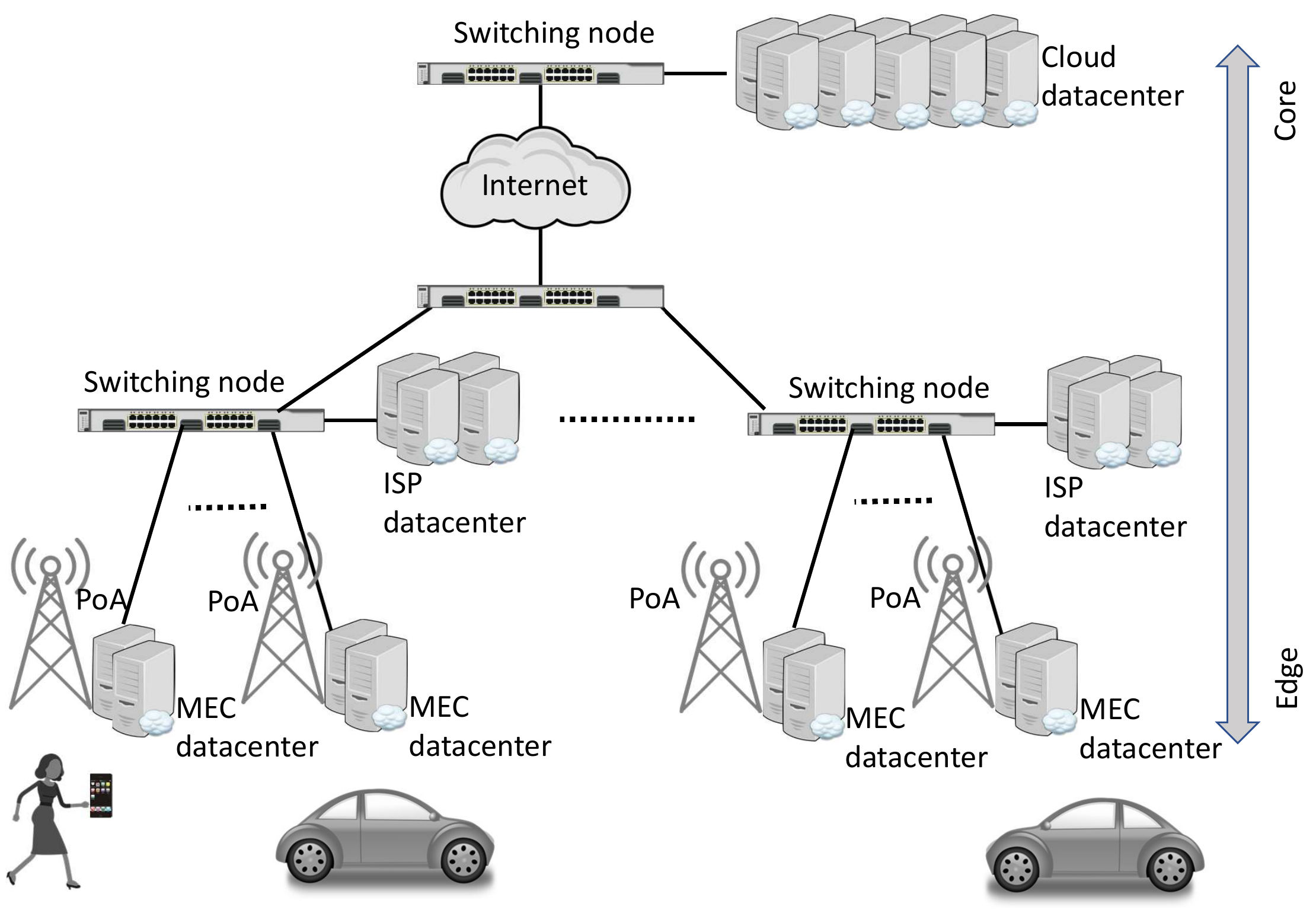}
    \caption{Example of network scenario for mobile service provisioning.  }\label{fig:net}
\end{center}
\vspace{-4mm}
\end{figure}

\subsection{Service delay}\label{sec:delay_model}
The service delay comprises  the computational and the network delays, as detailed below.

\textbf{Computational delay.} 
Given chain $\hv^u$, each VM $v^u_k \in \hv^u$
has some input traffic load, $\lambda^u_k$, expressed in bit/s, which is known a-priori~\cite{mig_correlated_VMs, SFC_mig}.
In particular, $\lambda^u_1$ denotes the input traffic to the chain at the PoA associated with the request.
We let $\theta^u_k$ represent the processing capacity required to handle 
a single unit of traffic corresponding to $v^u_k$, expressed in CPU cycles/bit. Thus, $\lambda^u_k \theta^u_k$ represents the  CPU cycles/s required to process the incoming traffic.
The computation is defined in terms of  single  data units to be processed.
Let $\gamma_k^u$ be the number of bits per data unit, then $\gamma_k^u \theta_k^u$ is the  number of CPU cycles required to process each data unit\footnote{As an example, a single data unit could be a video frame  to process in an  object-recognition VNF. Different image resolutions will result in different values of $\gamma^u_k$ and $\theta_k^u$. In a DPI application, a single data unit would be instead a single data packet.}.

Let chain $\hv^u$ be placed on datacenter $s$.
For each VM $v^u_k \in \hv^u$,  $\mu^{u,s}_k$ denotes the processing capacity allocated to such VM on $s$, expressed in number of CPU cycles/s.
As often done in the  literature~\cite{Okpi, Joint_VNF_placement_n_cpu_alloc_Carla_Francesco, OsS_aware_VNF_placement_using_MM1_model,mm1jsac},  CPU processing at the VM is modeled through an M/M/1 queue. 
The average computational delay at VM $v^u_k$ to process one data unit is given by 
\tagged{long}{\begin{equation}\label{def:per_VM_comp_delay}
D^u_k (\mu^{u,s}_k) = {\gamma_k^u \theta_k^u}/({\mu^{u,s}_k - \theta^u_k \lambda^u_k}),  \end{equation}}
\tagged{short}{$D^u_k (\mu^{u,s}_k) = {\gamma_k^u}/({\mu^{u,s}_k - \theta^u_k \lambda^u_k})$,}%
where we must have $\mu^{u,s}_k > \theta^u_k \lambda^u_k$.
The overall {computational delay} of chain $\hv^u$ is thus given by:
\iftagged{short}{$d^c(\hv^u,\muv^{u,s}) = 
\sum_{k = 1}^{h_u} D^u_k (\mu^{u,s}_k)$.}{
\begin{equation}
\label{Eq:def_chain_comp_delay}
d^c(\hv^u,\muv^{u,s}) = 
\sum_{k = 1}^{h_u} D^u_k (\mu^{u,s}_k).
\end{equation}}
To reflect real-world conditions, when allocating virtual cores to VMs, we consider that $\mu^{u,s}_k$ is an integer multiple of a basic CPU speed, and, thus,   $\mu^{u,s}_k$ is discrete and  will be expressed in  CPU units in the following;  $\theta_k^u$ is coherently expressed in a fractional value of CPU units.

We denote by $\hat{C}_u$ the maximum amount of computing resources that may be allocated to chain $\hv^u$ as per the SLA.
Each datacenter $s\in \Sc$ has a total processing capacity $C_s$, 
expressed in number of CPU cycles/s.
It is fair to assume that $\hat{C}_u$ does not exceed the processing capacity of any single datacenter in the system,
i.e., $\forall s \in \Sc, u \in \Uc, \hat{C}_u \leq C_s$.

\textbf{Network delay.}
{The intra-datacenter communication delay is typically negligible~\cite{justifies_neglecting_intra_DC_comm} compared to the delays  in the network connecting the datacenters. Thus, we will consider here just the delays in the inter-datacenter communications.} 

We consider a deterministic system with token bucket controlled  traffic (as for TSPEC in IntServ) and rate-latency, as in~\cite{netcalc}. 
We consider the delay 
accrued along
the path traversed by chain $\hv^u$'s traffic in the most general case, where such a path includes both uplink and downlink traffic transfers. 
To seamlessly model the downlink traffic,
we denote by $\lambda^u_{h_u+1}$  the traffic from the last VM in the chain back to the PoA of the request.

Each link $(i,j)$  is associated with propagation delay $T_{p}(i,j)$ and bandwidth capacity $C(i,j)$. 
We assume that on each link the bandwidth is partitioned between all the traversing flows, and a link scheduler provides a rate guarantee for each chain $\hv^u$ equal to 
$\lambda_1^u$ for uplinks, and $\lambda_{h_u+1}^u$ for downlinks.
We assume that the bandwidth on each link is
sufficiently provisioned,
thus $C(i,j)$ is large enough to accommodate all the traffic flows between $i$ and $j$, thus avoiding blocking events.

We assume that ingress and egress chain traffic is leaky-bucket regulated, with maximum burstiness $\sigma_u$. 
The link $(i,j)$ scheduler is a rate-latency server (e.g., a PGPS scheduler~\cite{netcalc}) with 
the appropriate rate $\lambda^u$ as specified above
and latency $L_{\max}/C(i,j)$, where $L_{\max}$ is a constant depending upon the adopted scheduler (e.g., equal to the maximum packet size for PGPS). 
Using network calculus~\cite{netcalc} and defining  $\tau(i,j)=L_{\max}/C(i,j)+
T_{p} (i,j)$, the delay experienced on link $(i,j)$ can be upper-bounded by: 
\iftagged{short}{${\sigma_u}/{\lambda^u}+ 
\tau(i,j)$.}
{ 
\begin{equation}\label{eq:single_hop_delay}
\dfrac{\sigma_u}{\lambda^u}+ 
\tau(i,j)   \,.
\end{equation}
}
Then, recalling the ``pay bursts only once'' result~\cite{netcalc},
the {\em network delay} associated with chain $\hv^u$ is:
\begin{align}
\label{Eq:def_chain_netw_delay}
d^n(\hv^u,s) &= \dfrac{\sigma_u}{\lambda_1^u} + \dfrac{\sigma_u}{\lambda_{h_u+1}^u} + 
\sum_{{(i,j) \in \Pc(p^u,s) \cup 
 \Pc(s,p^u)}} \tau(i,j) 
\end{align}
where $s$ is the datacenter on which chain $\hv^u$ is deployed. 

\textbf{Total service delay.}
The total delay of chain $\hv^u$ is then given by the sum of its computational and network delay, i.e.,
\begin{equation}
\label{Eq:def_total_chain_delay}
d(\hv^u, \muv^{u,s}, s) =  d^c(\hv^u,\muv^{u,s}) + d^n (\hv^u,s).
\end{equation}

\section{The deployment and migration problem}\label{sec:problem}

The delay experienced by a chain may vary over time because
\begin{inparaenum}[(i)]
\item the PoA of the request, hence the network delay, has changed, or,  
\item there is a traffic surge/reduction, and the processing time of the chain VMs changes~\cite{Dynamic_user_demands}.
\end{inparaenum}
We assume that a monitoring system predicts the performance of the deployed services every $T$ time units (hereinafter also referred to as the {\em decision period}),
and it%
\tagged{long}{
identifies
the set of}%
\tagged{short}{identifies the set of}
{\em critical} chains, 
whose experienced latency  is expected to  violate the delay constraints due to changes in the requests'
attributes (e.g., PoA, or values of $\lambda^u_k$).
The delay constraint of a critical chain may dictate migrating that chain to reduce its delay. 
{ Every decision period, the service orchestrator  
 decides on the destination datacenter and on the resources to allocate for the new chains and for the critical chains that need to be migrated. Notably, according to our definition, $\Hc$ comprises the new chains, the critical chains, and the remaining currently deployed non-critical chains.}

In what follows, we formulate an optimization problem defining the framework in which such decisions are made with the aim of minimizing the migration and the system operational cost. 
\tagged{long}{We first introduce the problem decision variables, constraints,  system costs, and, then,  our objective function. 
Finally, we discuss the problem  complexity. }

{\bf Decision variables.} 
Let $\yv$ denote the Boolean {\em placement} decision variables.
Namely, $y(u,s) = 1$ iff chain $\hv^u$ is scheduled to run on datacenter $s$ {in the following decision period}.
The {\em allocation} decision variables, $\muv$, determine, for every VM $v^u_k$ of  chain $\hv^u$, the amount of computing capacity $\mu^{u,s}_k$ to be allocated for this VM on datacenter $s$ hosting the chain.
Any choice for the values of the 
$\yv$- and $\muv$-variables comprises a {\em solution} to our problem, specifying 
(i) where new chains are deployed and what computing resources each of their VMs gets, and (ii) which existing chains are migrated, where they are migrated to, and what computing resources each of their VMs use.

{\bf Constraints.} 
The following constraints hold:
\begin{align}
\textstyle
\sum_{s \in \Sc} y(u,s)
&= 1
&\forall \hv^u \in \Hc \label{problem_constraint:single_placement} \\
\textstyle
y(u,s) \cdot 
d(\hv^u, \muv^{u,s}, s) 
&\leq \Delta(\hv^u)
&\forall \hv^u \in \Hc, \forall s \in \Sc \label{problem_constraint:target_delay}\\
\textstyle
y(u,s) \cdot \theta^u_k \cdot \lambda^u_k
&\leq \mu^{u,s}_k
&\forall \hv^u \in \Hc, \forall s \in \Sc \label{problem_constraint:finite_positive_delay}\\
\textstyle
\sum_{k=1}^{h_u} \mu^{u,s}_k
&\leq \hat{C}_u
&\forall \hv^u \in \Hc \label{problem_constraint:chain_computatioal_capacity}\\
\textstyle
\sum_{\hv^u \in \Hc} \sum_{k=1}^{h_u} \mu^{u,s}_k
&\leq C_s
&\forall s \in \Sc\,. \label{problem_constraint:datacenter_computational_residual_capacity}
\end{align}
Indeed, \eqref{problem_constraint:single_placement} ensures that each chain is associated with a single {\em scheduled} placement. 
~\eqref{problem_constraint:target_delay} guarantees that the target maximum delay of each chain is satisfied. 
~\eqref{problem_constraint:finite_positive_delay} ensures a finite delay for each VM. 
~\eqref{problem_constraint:chain_computatioal_capacity} verifies the bound of computing resources allocated for each chain. Finally,~\eqref{problem_constraint:datacenter_computational_residual_capacity} makes sure that the capacity of each datacenter is not exceeded.

{\bf Costs.}
The system costs are due to migration, as well as computation and bandwidth usage, as detailed below.

 Migrating chain $\hv^u$ from datacenter $s$ to datacenter $s'$ incurs a {\em computational migration cost}   $\chi^m(u,s,s')$. 
Let $x(u,s)$ denote the {\em current placement indicator parameters}, i.e., $x(u,s) = 1$ iff chain $\hv^u$ is currently placed on datacenter $s$.\tagged{long}{\footnote{Note that $x(u,s)$ are not decision variables to be determined, but rather represent the {\em current state} of the deployment.}}
\tagged{short}{Note that $x(u,s)$ is not a decision variable.}
The migration cost incurred by a critical chain $\hv^u$ is then:
\iftagged{short}{$\sum_{s\neq s' \in \Sc} x(u, s) \cdot y(u, s') \cdot \chi^m(u,s,s')$.}{ 
\begin{align}
\sum_{s\neq s' \in \Sc} x(u, s) \cdot y(u, s') \cdot \chi^m(u,s,s'). \notag
\end{align}
}

Each unit of computation allocated on datacenter $s$ incurs a {\em computation cost} $\chi^c(s)$. 
Finally, each unit of traffic being routed across link $(i,j)$ incurs a {\em bandwidth cost} $\chi^b(i,j)$. 

{\bf Objective.}
{Our objective is to minimize the cost function}
\begin{multline}
\label{Eq:def_obj_func}
\phi(\yv,\muv)
=\sum_{\hv^u \in \Hc} 
\sum_{{s\neq s' \in \Sc }}
x(u,s) \cdot y(u,s') \cdot \chi^m(u,s,s') \\ 
\quad+ \sum_{\hv^u \in \Hc} 
\sum_{s \in \Sc} y(u,s) \sum_{k=1}^{h_u} \mu^{u,s}_k \cdot \chi^c(s) 
+ \sum_{(i,j) \in \Lc} b(i,j, \yv) \cdot \chi^b(i,j)
\end{multline}
{where $b(i,j, \yv)$ denotes the overall amount of traffic (in bit/s) that traverses link $(i,j)$ when using placement $\yv$. Namely, considering the traffic towards the datacenters and back:}
\begin{align}
\label{Eq:BW_def}
b(i,j,\yv)
= 
\begin{cases}
\displaystyle
\sum_{\hv^u \in \Hc} ~~ \sum_{\substack{s \in \Sc \\ (i,j) \in \Pc(p^u,s)}} y(u,s) \cdot \lambda^u_1 & \mbox {uplink} \\
\displaystyle
\sum_{\hv^u \in \Hc} ~~\sum_{\substack{s \in \Sc \\ (i,j) \in \Pc(s,p^u)}} y(u,s) \cdot \lambda^u_{h_u+1} & \mbox{downlink}.
\end{cases}
\end{align}

Our problem, hereinafter referred to as the {\em Deployment and Migration Problem} (\migProb), is therefore given by:
\begin{equation}\label{eq:min-problem}
\min_{\yv,\muv} \phi(\yv,\muv) \
\textrm{subject to}~\eqref{problem_constraint:single_placement} -
\eqref{problem_constraint:datacenter_computational_residual_capacity}.
\end{equation}

\begin{table}
\centering
    \scriptsize
    \caption{Main notations
    \label{tab:notations}}
    \begin{tabular}{|M{0.17\columnwidth}|m{0.71\columnwidth}|}
        \hline 
	    Symbol & \multicolumn{1}{c|}{Description} \tabularnewline
	    \hline \hline
	    \multicolumn{2}{|c|}{Parameters: network (Sec.~\ref{sec:netw_model})}\\
        \hline
        $\Gc$ & Network graph \tabularnewline
        \hline
        $D(\Gc)$ & The diameter  of network graph $\Gc$ \tabularnewline
        \hline
        $\Sc$ & Set of datacenters \tabularnewline
        \hline
        $\Lc$ & Set of links \tabularnewline
        \hline
        $(i,j)$ & Directed link connecting datacenters $i$ and $j$ \tabularnewline
        \hline
        $\mathcal P(s,s')$ & Path connecting datacenter $s$ to datacenter $s'$\tabularnewline
        \hline \hline
	    \multicolumn{2}{|c|}{Parameters: services, delays and capacities (Secs.~\ref{sec:service_model}--\ref{sec:delay_model})}\\
        \hline
        $\Uc$ & Set of service requests \tabularnewline
        \hline
        $\Hc$ & Set of chains corresponding to $\Uc$ \tabularnewline
        \hline
        $\hv^u$ & Service chain (ordered list of VMs) serving  $u$ \tabularnewline
        \hline
        $h_u$ & Number of VMs in $\hv^u$ \tabularnewline
        \hline
        $p^u$ &PoA where request for $\hv^u$ is generated \tabularnewline
        \hline
        $\Delta (\hv^u)$ & Target delay [s] of chain $\hv^u$ \tabularnewline
        \hline
        $\hat{C}_u$ & Maximum CPU units that may be allocated to chain $\hv^u$ based on the SLA \tabularnewline
        \hline
        $C_s$ & Overall processing capacity of datacenter $s$ [cycles/s] \tabularnewline
          \hline 
        $\lambda^u_k$ & Input traffic load of VM ${v^u_k}$ [bits/sec]\tabularnewline
        \hline
        $\theta^u_k$ & Required processing capacity for VM $v^u_k$'s incoming traffic [cycles / bit]\tabularnewline
        \hline
        $\gamma^u_k$ & Bits per data units [bit]\\
          	\hline
        $D^u_k (\mu^{u,s}_k)$ & 
        Computational delay [s] exhibited by VM $v^u_k$~\eqref{def:per_VM_comp_delay}\tabularnewline
        \hline
        $d^c(\hv^u, \muv^{u,s})$ & Computational delay [s] of chain $\hv^u$~\eqref{Eq:def_chain_comp_delay}\tabularnewline
        \hline
        $\tau(i,j)$ & 
        Network delay experienced on link $(i,j)$ \tabularnewline
        \hline
        $d^n(\hv^u,s)$ & Network delay [s] of chain $\hv^u$ when located on server $s$~\eqref{Eq:def_chain_netw_delay}\tabularnewline
	    \hline \hline
	    \multicolumn{2}{|c|}{Parameters: costs (Sec.~\ref{sec:problem})}\\
        \hline
		\rule{0pt}{2.4ex}    
        $\chi^m(u, s, s')$ & Cost of migrating chain $\hv^u$  from datacenter $s$ 
        to datacenter $s'$\tabularnewline
        \hline
        $x (u,s)$ & Placement indicator: true iff chain $\hv^u$ is currently hosted on datacenter $s$\tabularnewline
        \hline 
        $\chi^c(s)$ & Cost of allocating one CPU unit on datacenter $s$\tabularnewline
        \hline
        $\chi^b(i,j)$ & Cost of having one unit of bandwidth traverse link $(i,j)$\tabularnewline
        \hline
        $\phi$ & Objective function~\eqref{Eq:def_obj_func} \tabularnewline
        \hline
        $b(i,j, \yv)$ & Amount of traffic traversing link $(i,j)$ [bits/s]~\eqref{Eq:BW_def}\tabularnewline
        \hline \hline
	    \multicolumn{2}{|c|}{Decision variables (Sec.~\ref{sec:problem})}\\
        \hline
        $y (u,s)$ & Placement indicator: true iff chain $\hv^u$ is scheduled to run on datacenter $s$\tabularnewline
        \hline 
        $\mu^{u,s}_k$ & Integer allocation: expressing the number of CPU units allocated for VM $v^u_k$ on datacenter $s$ \tabularnewline
        \hline \hline
 	    \multicolumn{2}{|c|}{Sec.~\ref{sec:alloc}} \tabularnewline
 	    \hline
        $\Sc_u$ & Set of delay-feasible datacenters of chain $\hv^u$ \tabularnewline
	    \hline
        $\mu^{u,s}_k$ & CPU allocation for the $k$-th VM of chain $\hv^u$ on datacenter $s$ \tabularnewline
        \hline        
        $\delta^u_k$        & Delay reduction function \eqref{eq:def_delay_reduc_f} \tabularnewline  
        \hline \hline        
	    \multicolumn{2}{|c|}{Sec.~\ref{sec:bu}} \tabularnewline
	    \hline 
        $\Tc(s)$ & The sub-tree rooted by datacenter $s$ \tabularnewline
        \hline
        $\overline{s}(u)$ & Top datacenter in $\Sc_u$\tabularnewline
        \hline
        $\Hc(s)$ & The set of chains whose PoAs are in $\Tc(s)$\tabularnewline
        \hline
        $\tilde{\mu} (u)$ & The minimum required CPU to serve request $u$  as in~\eqref{eq:def:min_required_CPU} \tabularnewline
        \hline        
        $T_{\Hc'}$ & Potential placement tree of $\Hc'$ \\ 
        \hline        
        $R$ & Multiplicative resource augmentation factor\\
        \hline        
        $a_s$ & Available processing capacity of datacenter $s$ \tabularnewline   
        \hline
        \end{tabular}
\end{table}

We can prove the following result on the \migProb\ complexity.
\begin{theorem}\label{thm:NPH}
The \migProb\ is NP-hard.
\end{theorem}
\begin{IEEEproof}
Consider the NP-hard partition problem~\cite{garey79computers}, where we are given a sequence of integers $n_1,\ldots,n_k$, and we seek a set $N \subseteq \set{1,\ldots,k}$ such that $\sum_{i \in N} n_i = \left( \sum_{i=1}^k n_i \right)/2$.
Without loss of generality, assume that $n_i \geq 2$ for all $i$ (otherwise, we may simply consider the integers $\hat{n}_i=2\cdot n_i$ for all $i$ as our input).

We now present a polynomial reduction from the partition problem to the \migProb.
Consider a network with two  datacenters, $s$ and $r$, where $r$ is the root, and $s$ is  co-located with the PoA of all requests.
The
processing capacity in both the root and the PoA is $C_s=C_r=\left( \sum_{i=1}^k n_i \right)/2$.
Define requests $u=1,\ldots,n$ where chain $\hv^u$ has delay constraint $\Delta(\hv^u)=\frac{1}{n_u-1}$.
Each requested chain has a single VM, {with rate $\lambda^u_1=\lambda^u_2=1$, and requires a
processing
capacity $\theta_1^u=\theta_2^u=1$}. The network delay is zero.
Observe that the delay constraint~\eqref{problem_constraint:target_delay} of chain $i$ is satisfied only if it is allocated at least $n_i$ CPU units. Furthermore, since $C_s + C_r=\sum_{i=1}^k n_i$, a feasible solution for \migProb\ may allocate for chain $i$ at most $n_i$ CPU units. It follows that any feasible solution allocates exactly $n_i$ CPU units for chain $i$. 
Hence, there exists a feasible solution for \migProb\ for this input iff there exists a solution to the partition problem.
The result follows.
\end{IEEEproof}

\section{Solution overview and main results}\label{sec:alg_concept}

The \migProb's objective \eqref{Eq:def_obj_func} combines the placement decision variables, $\yv$, and the allocation decision variables, $\muv$. 
A closer look shows that the chain placement and CPU allocation problems are entangled, since each  placement decision impacts the  CPU allocation required to satisfy the target  delay of the service. 
Our solution concept is based on {\em decoupling the chain placement and CPU allocation problems}, which allows applying a combinatorial approach to solving the \migProb, and studying the trade-offs inherent to our solutions. 
{
In more detail, our solution  comprises three steps: (i) solving the CPU allocation problem (Sec.~\ref{sec:alloc}), 
(ii) finding a feasible solution for the chain placement problem (Sec.~\ref{sec:bu}), 
and (iii) reducing cost (Sec.~\ref{sec:top_lvl}).
We now overview these steps.

\subsection{Solving the CPU allocation problem}
In Sec.~\ref{sec:alloc}, we define
a polynomial-time algorithm, called {\em GetFeasibleAllocations} (\cpuall), that identifies for each chain $\hv^u$ its set of {\em delay-feasible} datacenters, namely, the datacenters on which it is possible to place $\hv^u$, while meeting its target delay.
For each chain $\hv^u$ and {delay-feasible} datacenter $s$, \cpuall\  finds an allocation $\muv^{u,s}$ that is {\em provably minimal} in terms of the overall number of required CPU units.

\subsection{Solving the chain placement problem}
\label{sec:roadmap_placement}

{First, we note that given any allocation $\muv^{u,s}$ for all chains $\hv^u$ and datacenters $s$, the DMP~\eqref{eq:min-problem} becomes an integer linear program (ILP). 
We note that the optimal solution for the linear relaxation serves a lower bound for the DMP, and also serves as a witness of feasibility.
}

The proof of Theorem~\ref{thm:NPH} implies that even when the solution for the CPU allocation problem is known (e.g., allocating $n_i$ CPU units to VM $i$ in the proof of Theorem~\ref{thm:NPH}), the \migProb\ is NP-hard. Hence, the following proposition holds.
\begin{proposition}
\label{prop:NP-H}
Finding a feasible solution to the chain placement problem is NP-hard.
\end{proposition}

In Sec.~\ref{sec:bu}, we address the hardness of the chain placement problem using {\em resource augmentation}, i.e., assuming that each datacenter has an augmented processing capacity. 
We develop a polynomial-time algorithm, dubbed {\em Bottom-Up (BU)}. Further, we show an upper bound on the amount of processing capacity augmentation required for \bu\ to find a feasible solution whenever one exists for the non-augmented case.

\subsection{Reducing cost}
In Sec.~\ref{sec:top_lvl} we present the {\em Push-Up (\pushUp)} algorithm, which aims at reducing the cost of any given feasible \migProb\ solution.
Then, we use \bu\ and \pushUp\ as building blocks of our integrated algorithm {\em Bottom-Up-and-Push-Up (\algtop)} for finding a minimal-cost solution to \migProb\ while minimizing the amount of resource augmentation. 
}

\section{Allocating computational resources}
\label{sec:alloc}

In this section, we address the CPU {\em allocation problem}.
In particular, for each chain $\hv^u$, we identify its set of {\em delay-feasible} datacenters, denoted by $\Sc_u$.
Then, for every chain $\hv^u$ and datacenter $s \in \Sc_u$, we calculate the minimal number of CPU units that one must assign to each VM in $\hv^u$ running on $s$, in order to satisfy the target delay. 
Fig.~\ref{fig:alloc_toy_example} provides some intuition on the allocation problem.

\newcommand{\inlinelist}[1]{\begin{inparaenum}[(i)] #1 \end{inparaenum}}

\begin{figure}[!tb]
\begin{center}
    \includegraphics[width=7cm]{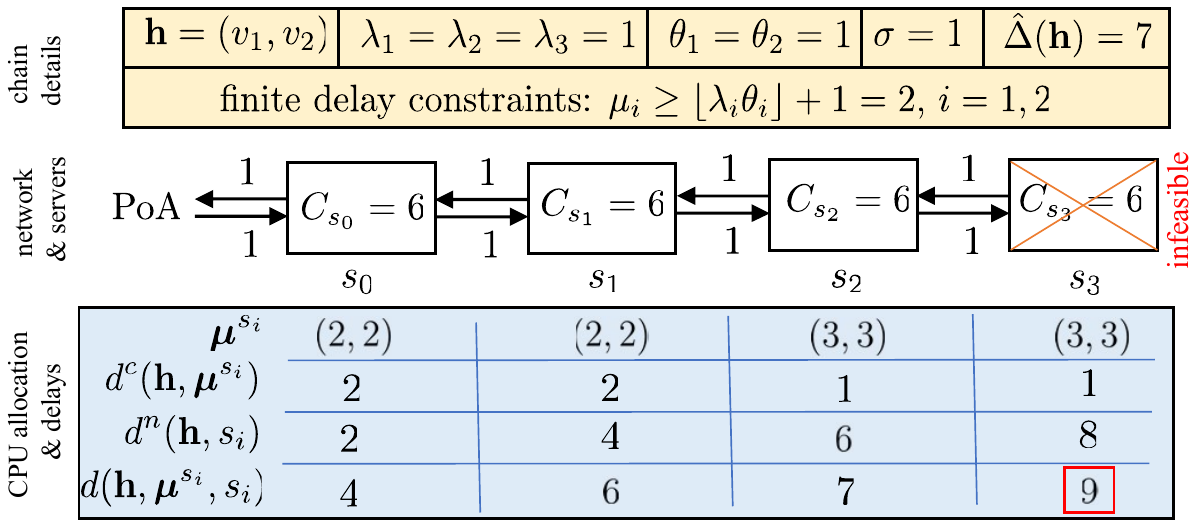}
    \caption{Example of the resource allocation dynamics and their effect on delay and feasibility.
    The input chain details are given at the top; the datacenter topology is given in the middle, with all link delays and
    datacenter capacities set to 1 and 6 (resp.). 
    The table at the bottom indicates:
    \inlinelist{%
     \protect\item the minimal allocation $\muv^{s_i}$ on datacenter $s_i$,
     \protect\item the induced computational delay $d^c(\hv,\muv^{s_i})$,
     \protect\item the network delay incurred for reaching datacenter $s_i$, and
     \protect\item the total delay $d(\hv,\muv^{s_i},s_i)$.
     }
     As can be seen, while datacenters $s_0,s_1,s_2$ are feasible (each with its CPU allocation), datacenter $s_3$ is infeasible.
    }
    \label{fig:alloc_toy_example}
    \end{center}
    \end{figure}

\subsection{The CPU allocation algorithm}\label{sec:alloc:descibe_alg}

Our polynomial-time CPU allocation algorithm, named GetFeasibleAllocations (or \cpuall, for short), takes as input the network graph $\Gc$ and a given set of chains $\tilde{\Hc}$.
\cpuall\ computes the set of delay-feasible datacenters $\Sc_u$ for each chain $\hv^u\in \tilde{\Hc}$, and, for each such datacenter $s$, calculates 
a CPU allocation vector $\muv^{u,s}$ satisfying the delay constraint and minimizing the  CPU  allocated to $\hv^u$ on  $s$. Formally, 
$\muv^{u,s} \in \argmin_{\muv} \{ \norm{\muv}_1 | \ d(\hv^u, \muv) \leq \Delta_s(\hv^u)\}$,
where $\norm{\cdot}_1$ is the $\ell_1$ norm; 
such a CPU allocation is referred to as {\em minimal allocation}.


\cpuall\ is detailed in Alg.~\ref{alg:alloc}.
For each chain $\hv^u\in\tilde{\Hc}$, initially  all datacenters in $\Pc(p^u,s)$ are assumed to be feasible (ln.~\ref{alg:alloc:init_Scfu}).
For each feasible datacenter, going from the PoA towards the root of the network, \cpuall\ initializes the CPU allocation for each VM in $\hv^u$ to the minimal necessary to ensure a finite computational delay (ln.~\ref{alg:alloc:verify_finite_delay_loop_begin}--\ref{alg:alloc:verify_finite_delay_loop_end}).
The algorithm then computes  (ln.~\ref{alg:alloc:if_within_slack}-\ref{alg:alloc:while_below_Cu_end}) the minimum amount of CPU  required to meet the delay constraint, while not violating the bound~\eqref{problem_constraint:chain_computatioal_capacity} on the computing resources allocated to the chain.
This is done using the method described below.
Finally, once the current datacenter $s$ is delay infeasible, which by our model also implies that all its ancestors are deemed infeasible,
$s$ and all its ancestors are removed from the set of feasible datacenters (ln.~\ref{alg:alloc:if_s_infeasible}-\ref{alg:alloc:loop_s_end}), and \cpuall\ returns the set of feasible datacenters, and the corresponding minimal allocations (ln.~\ref{alg:alloc:rtrn}).

{\bf Computing the minimal allocation.}
If the current allocation $\muv$ leads to a delay constraint violation, \cpuall\ increases the total number of CPU units it uses by one using a gradient method: it increments the number of CPU units allocated to the VM that, owing to this change, maximizes the delay reduction  (ln.~\ref{alg:alloc:argmax}--\ref{alg:alloc:inc_argmax}).
To this end, we define the 
{\em delay reduction function}, which captures the residual reduction in the computational delay corresponding to increasing the CPU allocation of $v^u_j$ by one. 
Formally, for  VM $v^u_k$ and CPU allocation $\mu$, we have:
\begin{equation}\label{eq:def_delay_reduc_f}
\delta^u_k (\mu) = D^u_k (\mu) - D^u_k (\mu+1)\,.
\end{equation}
As can be verified by  algebraic manipulation, $\delta^u_k (\mu)$ is monotonically decreasing for every $\mu > \floor{\theta^u_k \lambda^u_k}$.
Next, we prove that our approach indeed finds a minimal CPU allocation.

\begin{algorithm}[t!]\caption{\small \cpuall($\Gc, { \tilde{\Hc}}$)}
\scriptsize
\label{alg:alloc}
\begin{algorithmic}[1]
    \For {$\hv^u \in { \tilde{\Hc}}$}    \Comment{for each chain} \label{ alg:alloc:for_u_begin}
        \State $\Sc_u = \Pc(p^u,r)$ \label{alg:alloc:init_Scfu}
        \Comment{ordered list of datacenters from PoA to the root}
        \ForEach {$s\in \Sc_u $} \label{alg:alloc:loop_s_begin}\Comment{for each datacenter in $\Sc_u$, from PoA to the root}
            \For {$k = 1, \dots, h_u$} \label{alg:alloc:verify_finite_delay_loop_begin}\Comment{for each chain in the VM}
                \State $\mu^{u,s}_k = \floor*{\theta^u_k \lambda^u_k} + 1$ \label{alg:alloc:init_mu} 
                \Comment{ensure finite computation delay}
            \EndFor \label{alg:alloc:verify_finite_delay_loop_end}
            \While {     $d^c (\hv^u, \muv^{u,s}) > \Delta(\hv^u) - d^n(\hv^u,s)${\bf and} $\sum_{k=1}^{h_u} \mu^{u,s}_k \leq \hat{C}_u$  }

            \label{alg:alloc:if_within_slack}
                \Statex 
                \Comment{delay constraint is still unsatisfied and CPU is available}
                \State $k^* = \argmax_{1 \leq k \leq h_u} \set{\delta^u_k (\mu^{u,s}_k)}$ \label{alg:alloc:argmax}
                \Comment{find VM with max del.\ reduction}
                \State $\mu^{u,s}_{k^*} = \mu^{u,s}_{k^*} + 1$ \label{alg:alloc:inc_argmax}\Comment{increase CPU to reduce delays}
            \EndWhile \label{alg:alloc:while_below_Cu_end}
            \If {$\sum_{k=1}^{h_u} \mu^{u,s}_k > \hat{C}_u$}
            \label{alg:alloc:if_s_infeasible}
            \Comment{if not enough CPU capacity, $s$ is infeasible}
                \State $\Sc_u = \Sc_u \setminus \Pc(s,r)$\Comment{remove all the datacenters from $s$ to the root}
                \State {\bf break}\Comment{it is not worth anymore to go further towards the root}
            \EndIf
        \EndFor \label{alg:alloc:loop_s_end}
    \EndFor \label{alg:alloc:for_u_end}
    \State {\bf return} $\Sc_u, \forall \hv^u \in { \tilde{\Hc}}$ and $\muv^{u,s}$, $\forall \hv^u \in {\tilde{\Hc}},s \in \Sc$ \label{alg:alloc:rtrn}
\end{algorithmic}
\end{algorithm}

\subsection{Performance analysis}\label{sec:alloc:analysis}


We begin by defining the {\em $B$-minimal} CPU allocation for a given chain and CPU budget.
\begin{definition} 
A CPU allocation for chain $\hv^u$ on  datacenter $s$ is $B$-minimal for a given CPU budget $B$, if it minimizes the computational delay of $\hv^u$ on $s$ while using 
$B$ CPU units, i.e.,
\iftagged{short}{$\tilde{\muv}^{u,s} (\hv^u, B) = \argmin_{\muv} \{d^c(\hv^u, \muv^{u,s}) | \norm{\muv^{u,s}}_1 = B\}$}{ 
\begin{equation}
\tilde{\muv}^{u,s} (\hv^u, B) = \argmin_{\muv} \{d^c(\hv^u, \muv^{u,s}) | \norm{\muv^{u,s}}_1 = B\}.
\end{equation}
}
\end{definition}

To prove that \cpuall\ finds a minimal CPU allocation, we will use the following lemma on $B$-minimal CPU allocations. 

\begin{lemma}
\label{lem:alloc:opt_j_minus_1}
Let $\muv^*$ be a $B$-minimal allocation for chain $\hv^u$ on datacenter $s$. Then, for any $1 \leq i, j \leq h_u$: 
$\delta^u_{j} (\mu^*_j - 1) \geq \delta^u_i (\mu^*_i)$.
\end{lemma}

\begin{proof}
Assume by contradiction that
\begin{multline}
\delta^u_{j} (\mu^*_j - 1)
= D^u_j (\mu^*_j - 1) - D^u_j (\mu^*_j) \\
< D^u_i (\mu^*_i) - D^u_j (\mu^*_i+1) 
=\delta^u_i (\mu^*_i).
\end{multline}
This implies that
$
D^u_j (\mu^*_j - 1) + D^u_j (\mu^*_i+1)
< D^u_j (\mu^*_j) + D^u_i (\mu^*_i).
$
Consider the allocation $\muv'$, defined as: $\mu'_j = \mu^*_j - 1$, $\mu'_i = \mu^*_i + 1$, and for any $k \neq i, j$: $\mu'_k = \mu^*_k$. Thus,  
$
d^c(\hv^u, \muv')
-
d^c(\hv^u, \muv^*) < 0,
$
which contradicts the $B$-minimality of $\muv^*$. 
\end{proof}

The following lemma shows that \cpuall\ never increases the CPU allocation to a VM above some level, before it exploits any chance to gain more delay reduction by increasing the CPU allocated to any other VM in that chain. 

\begin{lemma}
\label{lem:alloc_wont_pick_smaller}
If $\delta^u_j(a) > \delta^u_i(b)$,  \cpuall\ does not assign more than $b$ CPU units to $v^u_i$ 
before assigning at least $a+1$  units to $v^u_j$.
\end{lemma}

\begin{proof}
If $a \leq \floor*{\theta^u_k \lambda^u_k}$,  \cpuall\ initializes $\mu^{u,s}_j$ to at least $a+1$  units (ln.~\ref{alg:alloc:init_mu}), and the claim holds true in the first iteration of the while loop.
For any subsequent iteration, 
 the above lemma holds by construction (ln.~\ref{alg:alloc:argmax}--\ref{alg:alloc:inc_argmax}) since \cpuall\ will not assign $k^*=i$ (since it picks the VM that maximizes the improvement) before $v_j^u$ is assigned at least $a+1$ CPU units.
Note that the inequality $\delta^u_j(a) > \delta^u_i(b)$ is independent of any change made to the allocation of CPU units to any VM distinct from both $i$ and $j$.
\end{proof}

The following lemma bounds the $\ell_\infty$ distance
between the allocations that \cpuall\ considers, and any minimal allocation using an identical budget.

\begin{lemma}
\label{lem:chebyshev_leq1}
Let $\muv$ be the allocation for chain $\hv^u$ on datacenter $s$ in 
some iteration of \cpuall's {\em while} loop, and let $B = \sum_{k=1}^{h_u} \mu_k^{u,s}$. 
Let $\muv^*$ denote a $B$-minimal allocation for $\hv^u$ on datacenter $s$.
{Then, $\|\muv-\muv^*\|_\infty\leq 1$.}
\end{lemma}

\begin{proof}
Assume by contradiction that the $\|\muv-\muv^*\|_\infty> 1$.
Then, there exists an index $i$ for which 
either (1) $\mu_i - 1  > \mu^*_i$, or (2) $\mu^*_i - 1  > \mu_i$.
We have now two cases.

{\em Case 1}: $\mu_i - 1  > \mu^*_i$ .
As $B=\sum_{k=1}^{h_u} \mu_k = \sum_{k=1}^{h_u} \mu^*_k$, there exists another index $j \neq i$ s.t.\ $\mu^*_j - 1 \geq \mu_j$. Namely, 
\begin{equation}\label{eq:lem:alloc_i_ks_opt_pairs_gt1}
\mu_i - 1 >  \mu_i^* \quad , \quad \mu^*_j - 1 \geq  \mu_j. 
\end{equation}
By Lemma~\ref{lem:alloc:opt_j_minus_1},
\begin{equation}\label{eq:lem:opt_alloc_uj_dig_gt_!}
\delta^u_j (\mu^*_j - 1) \geq \delta^u_i (\mu^*_i).
\end{equation}

Combining (\ref{eq:lem:alloc_i_ks_opt_pairs_gt1}) and the fact that $\delta^u(\cdot)$ is monotone decreasing, we have that
$\delta^u_i (\mu^*_i) > \delta^u_i (\mu_i - 1),\, 
\delta^u_j (\mu_j) \geq \delta^u_j (\mu^*_j - 1)$. 
Combining these inequalities with (\ref{eq:lem:opt_alloc_uj_dig_gt_!}), 
we have 
$\delta^u_j (\mu_j) > \delta^u_i (\mu_i-1)$.
However, applying Lemma~\ref{lem:alloc_wont_pick_smaller} on this latter inequality
implies that \cpuall\ does not assign 
$\mu_i$ CPU units for $v^u_i$ before assigning at least 
$\mu_j+1$ CPU units to $v^u_j$. Hence, this case is impossible. 

{\em Case 2}: $\mu^*_i - 1  > \mu_i$. As $B=\sum_{k=1}^{h_u} \mu_k = \sum_{k=1}^{h_u} \mu^*_k$, there exists an index $j \neq i$ s.t. $\mu_j - 1 \geq \mu^*_j$. Namely,
\begin{equation}\label{eq:lem:alloc_case_2b}
\mu_i^* - 1 >  \mu_i \quad , \quad \mu_j - 1 \geq \mu_j^*. 
\end{equation}
Applying Lemma~\ref{lem:alloc:opt_j_minus_1}, while exchanging the roles of $i$ and $j$, we have 
$\delta^u_{i} (\mu^*_i - 1) \geq \delta^u_j (\mu^*_j)$.
Combining (\ref{eq:lem:alloc_case_2b}) and the fact that $\delta^u(\cdot)$ is monotone decreasing, we obtain
$\delta^u_{i} (\mu_i) > \delta^u_i (\mu^*_i - 1)$.
Combining the latter two inequalities, 
we have 
\begin{equation}\label{eq:alloc_proof:case2a_eq3}
\delta^u_{i} (\mu_i) > \delta^u_j (\mu^*_j).
\end{equation}
By setting $a = \mu_i$ and $b = \mu^*_j$ in Lemma~\ref{lem:alloc_wont_pick_smaller}, we know that \cpuall\ does not assign $v^u_j$ more than $\mu^*_j$ units before assigning to $v^u_i$ at least $\mu_i+1$ units. However, this contradicts our assumption that vector $\muv$ assigns only $\mu_i$ CPU units to $v^u_i$, while $v^u_j$ is already allocated some $\mu_j > \mu^*_j$ units. Therefore, also case 2 is impossible, and the thesis
 follows.
\end{proof}

The following lemma shows that \cpuall\ considers only  $B$-minimal allocations.
\tagged{short}{
This can be proved by using Lemma~\ref{lem:chebyshev_leq1} to show that any $B$-minimal allocation can be transformed into $\muv$.
The full proof is omitted.
} 
\begin{lemma}
\label{lem:alloc_considers_only_opt}
Let $\muv$ be the allocation in some iteration  of \cpuall's {\em while} loop for chain $\hv^u$ on datacenter $s$, and let $B = \norm{\muv}_1$. 
Then, $\muv$ is $B$-minimal.
\end{lemma}

\begin{proof}
Let $\muv^*$ denote a $B$-minimal allocation for chain $\hv^u$ on datacenter $s$, and assume by contradiction that $\muv$ is not $B$-minimal. 
By Lemma~\ref{lem:chebyshev_leq1}, 
\tagged{long}{
for any} $ 1 \leq i \leq h_u: \abs{\mu_i - \mu^*_i} \leq 1$. Then we can partition the non-equal indices in $\muv$ and
$\muv^*$ into $K$ pairs, where pair $k$ consists of two indices $1 \leq i_k, j_k \leq h_u$ s.t.\
\begin{equation}\label{eq:lem:alloc_i_ks_opt_pairs_leq1}
\mu^*_{i_k} = \mu_{i_k}-1, \quad \mu^*_{j_k} = \mu_{j_k}+1. 
\end{equation}
Applying Lemma~\ref{lem:alloc:opt_j_minus_1} with $i = i_k$ and $j = j_k$, we have $\delta^u_{j_k} (\mu^*_{j_k} - 1) \geq \delta^u_{i_k} (\mu^*_{i_k})$.
Combining this with (\ref{eq:lem:alloc_i_ks_opt_pairs_leq1}),
we have
\iftagged{short}{$\delta^u_{j_k} (\mu_{j_k}) = \delta^u_{j_k} (\mu^*_{j_k} - 1) \geq \delta^u_{i_k} (\mu^*_{i_k}) = \delta^u_{i_k} (\mu_{i_k} - 1)$, which }{ 
\begin{equation}\label{eq:lem:opt_alloc_uj_geq_ui}
\delta^u_{j_k} (\mu_{j_k}) = \delta^u_{j_k} (\mu^*_{j_k} - 1) \geq \delta^u_{i_k} (\mu^*_{i_k}) = \delta^u_{i_k} (\mu_{i_k} - 1). 
\end{equation}

Eq.~\eqref{eq:lem:opt_alloc_uj_geq_ui}} implies that either there exists $k$ s.t. $\delta^u_{j_k} (\mu_{j_k}) > \delta^u_{i_k} (\mu_{i_k} - 1)$, or for all $k$, it holds that $\delta^u_{j_k} (\mu_{j_k}) = \delta^u_{i_k} (\mu_{i_k} - 1)$. 
In the former case, 
by assigning $a = \mu_{j_k}$ and $b = \mu_{i_k}-1$ in Lemma~\ref{lem:alloc_wont_pick_smaller}, we know that \cpuall\ will not assign $\mu_{i_k}$ CPU units to $v^u_i$ before assigning at least 
$\mu_{j_k}+1$ CPU units to $v^u_j$. 
Since the current allocation is $\mu_i$ units to $v^u_i$, and $\mu_j$ units to $v^u_j$, a contradiction arises. 
Thus, we have that for all $k$, 
\begin{equation}\label{eq:alloc_proof:case_1b_def}
\delta^u_{j_k} (\mu_{j_k}) = \delta^u_{i_k} (\mu_{i_k} - 1).
\end{equation}

We now show that if~\eqref{eq:alloc_proof:case_1b_def} holds, then the total delay obtained by $\muv$ equals that obtained by $\muvstar$, thus contradicting the assumption that $\muv$ is not minimal.

By the definition of $\delta(\cdot)$ in~\eqref{eq:def_delay_reduc_f}, we get
\begin{equation}\label{eq:alloc_proof:case_1b_i}
{ \delta^u_{i_k}} (\mu_{i_k} - 1) 
 = { D^u_{i_k}}(\mu_{i_k}-1) - {D^u_{i_k}}(\mu_{i_k}) 
= {D^u_{i_k}}(\mu^*_{i_k}) - { D^u_{i_k}}(\mu_{i_k}),
\end{equation}
where the latter equation is by~\eqref{eq:lem:alloc_i_ks_opt_pairs_leq1}. 
Similarly, we have 
\begin{equation}\label{eq:alloc_proof:case_1b_j}
{\delta^u_{j_k}} (\mu_{j_k})  =
{ D^u_{j_k}}(\mu_{j_k}) - { D^u_{j_k}}(\mu_{j_k}+1) 
 = {D^u_{j_k}}(\mu_{j_k}) - { D^u_{j_k}}(\mu^*_{j_k}),
\end{equation}

Combining~\eqref{eq:alloc_proof:case_1b_def}, \eqref{eq:alloc_proof:case_1b_i} and~\eqref{eq:alloc_proof:case_1b_j}, we obtain: 
\iftagged{short}{$D^u_i(\mu_{i_k}) - D^u_i(\mu^*_{i_k}) + 
D^u_j(\mu_{j_k}) - D^u_j(\mu^*_{j_k}) 
= 0$.}{ 
\begin{align}\label{eq:alloc_proof:case_1b:material_anti_material}
{D^u_{i_k}} (\mu_{i_k}) - 
{D^u_{i_k}} (\mu^*_{i_k}) + 
{D^u_{j_k}} (\mu_{j_k}) - 
{ D^u_{j_k}}(\mu^*_{j_k}) 
= 0.
\end{align}
}
By definition of $K$, the difference in computational delay due to $\muv$ and $\muvstar$ is
\tagged{long}{
\begin{multline}\label{eq:alloc_proof:case_1b:final}
d^c (\hv^u, \muv) - d^c (\hv^u, \muv^*) = \sum_{i=1}^{h_u} 
\Big[D^u_{i}(\mu_i) - D^u_i(\mu^*_i) \Big]
=  \\
\sum_{k=1}^K 
\Big[
{D^u_{i_k}} (\mu_{i_k}) - 
{ D^u_{i_k}} (\mu^*_{i_k}) + 
{ D^u_{j_k}} (\mu_{j_k}) - 
{ D^u_{j_k}} (\mu^*_{j_k}) \Big]  = 0.
\end{multline}

By~\eqref{eq:alloc_proof:case_1b:final},
}
 the delay due to $\muv$ is the same as that due to $\muv^*$, thus contradicting our assumption on $\muv$ not being $B$-minimal.
\end{proof}

By Lemma~\ref{lem:alloc_considers_only_opt},  the following corollary holds.
\begin{corollary}\label{cor:alloc_finds_min_allocation}
For every chain $\hv^u$ and datacenter $s$, if there exists a feasible allocation for $\hv^u$ on $s$, then the allocation $\muv^{u,s}$ given by \cpuall\ satisfies
$
\muv^{u,s} \in \argmin_{\muv} \{ \norm{\muv}_1 | \ d^c(\hv^u, \muv) \leq \Delta_s(\hv^u)\},
$
i.e., it minimizes the  number of allocated CPU units over all feasible allocations for $\hv^u$ on $s$.
\end{corollary}
\begin{proof}
As \cpuall\ increments the total used CPU budget by one at each iteration (ln.~\ref{alg:alloc:inc_argmax}), we know that the budget that \cpuall\ used in the previous iteration, if exists, was $B-1$. 
By Lemma~\ref{lem:alloc_considers_only_opt},  no other allocation obtains lower delay for chain $\hv^u$ with budget $B-1$. Hence, no allocation satisfies the target delay constraint using  budget $B-1$. 
\end{proof}

\paragraph*{Run-time analysis} 
The time complexity of allocating CPU  to $\hv^u$ on $s$ is $\tilde{O}(\hat{C}_u+h_u)$.
For each of the $\Hc$ chains, \cpuall\  considers at most $D(\Gc)$ possible servers.
Thus, the overall time complexity of running \cpuall\ is
$\tilde{O}(\abs{\Hc} \cdot D(\Gc) \cdot \max_{\hv^u \in \Hc} (\hat{C}_u+h_u))$.

\section{Feasible solution to the placement problem}\label{sec:bu}
In light of the complexity of the placement problem  (see Proposition~\ref{prop:NP-H}), 
we now introduce the BU algorithm, which finds a feasible solution to chain placement with some resource augmentation.
In particular, after introducing some preliminaries, we prove that, using some bounded resource augmentation, BU always finds a feasible solution if such a solution exists in the non-augmented case.

\subsection{Preliminaries}

Let $\tilde{\mu}$ be a lower bound on the number of CPU units required to successfully serve any chain instance. By Corollary~\ref{cor:alloc_finds_min_allocation}, $\tilde{\mu}$ can be computed using the allocations $\muv$ found by \cpuall:
\iftagged{short}{$\tilde{\mu} = \min_{s \in \Sc, u \in \Uc} \norm{\muv^{u,s}}_1$.}{ 
\begin{equation}\label{eq:def:min_required_CPU}
\tilde{\mu} =
\min_{s \in \Sc, u \in \Uc} \norm{\muv^{u,s}}_1.
\end{equation}
}
Let $\Tc(s)$ denote the sub-tree rooted at datacenter $s$. 
Further, denote by $\Hc(s)$ the set of chains whose PoAs are in $\Tc(s)$. Denote by $\overline{s}(u)$ the top datacenter in $\Sc_u$ (i.e., the farthest datacenter from  PoA $p^u$ that is delay-feasible for $\hv^u$). 
Given a set of chains $\Hc' \subseteq \Hc$, we define the {\em potential placement tree} of $\Hc'$ as
\iftagged{short}{$\Tc_{\Hc'} = \bigcup_{\hv^u \in \Hc'} \Sc_u$.}{ 
\begin{equation}\label{eq:def_pot_loc_tree}
\Tc_{\Hc'} = \bigcup_{\hv^u \in \Hc'} \Sc_u.
\end{equation}
}
\begin{figure}[!t]
\begin{center}
    \includegraphics[width=6cm]{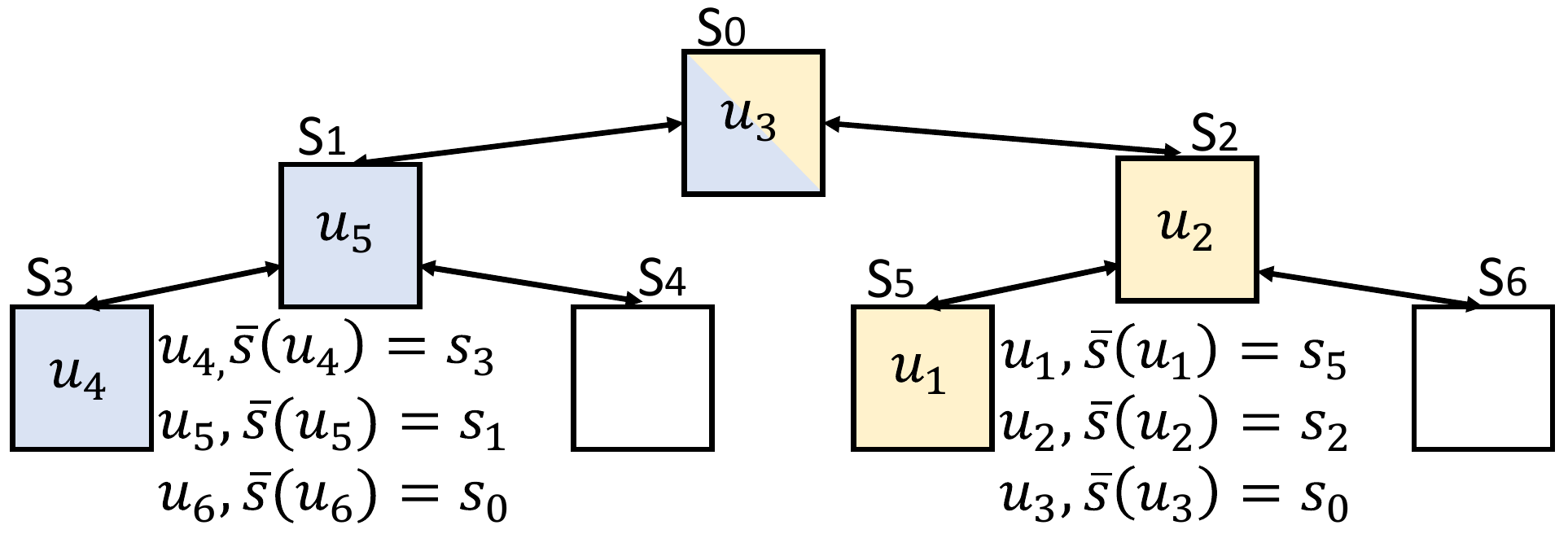}
    \caption{\label{fig:BU_fails_Tu'}
    Example of placement tree. 
    The PoA for requests $u_1$, $u_2$, and $u_3$ is $s_5$, while that for requests $u_4$, $u_5$, and $u_6$ is $s_3$.
    For each request,  the top datacenter of the relevant chain is denoted by $\overline{s}(u)$. 
    The potential placement trees 
    of $\set{u_1, u_2, u_3}$ 
    and of $\set{u_4,, u_5, u_6}$ are   highlighted in yellow and blue (resp.), while their union ($\Tc(s_0)$) is that of $\set{u_1, u_2, u_3, u_4, u_5, u_6}$. 
    }
\end{center}
\vspace{-4mm}
\end{figure}
An illustration clarifying the above notation is provided in Fig.~\ref{fig:BU_fails_Tu'}. 


\subsection{Motivation for a Bottom-Up approach}
\label{sec:naive_top_down}
When one looks at the problem of minimizing the placement cost, it may seem prudent to try and place chains as high as possible.
Indeed, processing costs are lower the farther the datacenter is from the chain's PoA. For instance, running a VM in a datacenter residing in the cloud is much cheaper than running it at a MEC datacenter attached to the PoA, where capacity is scarce and expensive.
Such an approach can be combined with a ``back-pressure'' mechanism, which tries to push previously placed chains lower in the hierarchy, whenever a request cannot be placed at its highest delay-feasible datacenter.
However, when targeting a feasible placement, such a mechanism might end up performing an exhaustive search.

To see this, consider a scenario where a chain $\hv^u$, originating at PoA $p^u$, cannot be accommodated in any of the datacenters along the sequence $\Sc_u$, due to previously placed chains.
In such a case, there might be a {\em unique} sequence of chains $\hv^{u_0},\ldots,\hv^{u_t}$ and a sequence of datacenters $s_0, \ldots s_{t+1}$, such that
\begin{inparaenum}[(i)]
\item $s_0 \in \Sc_u$,
\item $\hv^{u_i}$ is currently placed in $s_i$, for $i=1,\ldots,t$, and
\item re-placing $\hv^{u_i}$ in $s_{i+1}$ for $i=1,\ldots,t$, and placing $h_u$ in $s_0$ is feasible (given the placement of all other chains already handled).
\end{inparaenum}
Since there is no clear criteria for identifying such a sequence of chains and datacenters, simply testing all back-pressure adjustments may be prohibitively costly.
Since such a sequence may be unique (i.e., no other sequence can meet the above requirements), performing such an exhaustive search might be {\em necessary} to find a feasible solution.

The above scenario serves as a motivation for our {\em Bottom-Up (BU) algorithm} described in the sequel, which, by design, avoids such predicaments altogether.
We stress that the BU algorithm is meant to find a feasible solution efficiently. Once such a solution is found, we enhance it by pushing-up chains to decrease the total cost, as detailed in Sec.~\ref{sec:top_lvl}.

\subsection{The BU algorithm}

The \bu\ algorithm tries to place every chain as {\em low} (namely, closest to the corresponding PoA) as possible, and climbs higher only when there is insufficient processing capacity in any of the lower levels datacenters.
While this  approach may seem a poor choice in terms of cost, it is  effective in ensuring a feasible placement. 
Intuitively, if the \bu\ algorithm fails to place some chain, then there exists a sub-tree that is over-loaded with service requests, so that other algorithms are also unlikely to find a feasible solution.
Our algorithm will be using an {\em augmented processing capacity}, namely, \bu\ may allocate on datacenter $s$ up to $R \cdot C_s$ CPU units, where $R \geq 1$ is the multiplicative resource augmentation factor. 
We denote by $a_s$ the currently available (i.e., residual) processing capacity in datacenter $s$ during the execution of \bu.

\bu\ is detailed in Alg.~\ref{alg:BU}.
It gets as input the network graph $\Gc$;  a set {$\tilde{\Hc}$ of chains to be placed or migrated}; the CPU allocation and the set of delay-feasible datacenters of each chain; and the currently available processing capacity $\av$.
Given the above inputs, \bu\  computes placement $\yv$ and the new residual available processing capacity $\av$.
{\bu\ first releases the resources of all the chains, that will potentially be re-placed (ln.~\ref{alg:BU:rlz_rsrc_of_crit_chains}).}
Then, \bu\ scans the datacenters in a {Depth First Search (DFS)} order, and for each datacenter $s$ in this scan, it tries to deploy yet unplaced chains on $s$, while satisfying the (residual) capacity constraint on $s$ (ln.~\ref{alg:bu:if_enough_aviail_CPU}).
The yet unplaced chains are scanned in such a way that chains $\hv^u$ for which there are fewer remaining placement options 
{are considered first. The number of remaining placement options is the number of datacenters that are delay-feasible for $\hv^u$ but found above $s$, namely $\abs{\Sc_u \setminus \Tc(s)}$.}
In case a chain cannot be placed on any of its feasible datacenters, \bu\ returns an empty placement, which serves as a signal that the problem is infeasible (as proved in the sequel).

\begin{algorithm}[t!]\caption{\small\bu\ ($\Gc, {\tilde{\Hc}}, \set{\muv^{u,s}}_{u \in \Uc, s \in \Sc}, \set{\Sc_u}_{u \in \Uc}, \av$)}
\label{alg:BU}
\scriptsize
\begin{algorithmic}[1]
\State $\av \gets  \av $ after releasing resources of chains in ${\tilde{\Hc}}$\Comment{Release resources}\label{alg:BU:rlz_rsrc_of_crit_chains}
    \State $\yv \equiv 0$ \label{alg:bu:reset_y}\Comment{Init placement}
    \ForEach {datacenter $s$ in DFS order} \label{alg:bu:loop_over_s_begin}
            \For {$\hv^u \in { \tilde{\Hc}} \cap \Hc(s)$  in non-decreasing order of {
            ${\abs{\Sc_u \setminus \Tc(s)}}$}}  \label{alg:bu:loop_over_u_begin}
                \If {$a_s \geq \norm{\muv^{u,s}}_1$} \Comment{Check if enough residual capacity} \label{alg:bu:if_enough_aviail_CPU}
                    \State $y(u, s) = 1$ \Comment{Deploy the chain}
                    \label{alg:bu:set_y}
                    \State $a_s = a_s - \norm{\muv^{u,s}}_1$\Comment{Update the residual capacity}
                    \label{alg:bu:dec_avail_CPU}
                \ElsIf {$\overline{s}(u) = s$}\label{alg:bu:if_fail}\Comment{Check if the top datacenter has been reached}
                    \State return $\yv \equiv 0, \av \equiv 0$ \label{alg:bu:fail}\Comment{Infeasible}
                \EndIf
            \EndFor             
    \EndFor             \label{alg:bu:loop_over_s_end}
    \label{alg:BU:find_feasible_phase_end}
    \State return $\yv$, $\av$
\end{algorithmic}
\end{algorithm}


\subsection{Ensuring feasibility}

We now  upper-bound the amount of processing capacity resource augmentation, henceforth  
$R$,  that  guarantees that \bu\  finds a feasible solution, if such a solution exists in the case with no resource augmentation. 
Intuitively, our proof  shows that if \bu\ fails, then there exists a sub-tree $\Tc_{\Hc'}$ that is over-loaded by some set of chains $\Hc'$ that must be placed on $\Tc_{\Hc'}$.
Further, we show that in such a case {\em any} algorithm that does not use augmented processing capacities will fail as well.

We begin by characterizing the case where \bu\ may fail to place a chain on some datacenter $s$.

\begin{definition}
\label{def:suf_full}
A datacenter $s$ is almost-full if $a_s < \max_{u \in \Uc} \hat{C}_u$, and  
a set of datacenters $\Sc'$ is almost-full if every $s \in \Sc'$ is almost-full. 
\end{definition}

The notion of an almost-full datacenter helps to
upper-bound the amount of resource augmentation used by \bu. 
Further, our evaluation study in Sec.~\ref{Sec:sim} shows that the amount of resource augmentation used by \bu\ is significantly lower than our worst-case guarantees.
The following lemma provides a lower bound on the number of chains that \bu\ places on a datacenter before it becomes almost-full.
\begin{lemma}
\label{lem:num_of_chains_on_full}
If $s$ is almost-full, then \bu\ has placed on $s$ at least $\floor{R \cdot C_s  / \max_{u \in \Uc} \hat{C}_u}$ chains.
\end{lemma}
\begin{proof}
The augmented processing capacity on datacenter $s$ is $R \cdot C_s$. Now \bu\ allocates at most $\max_{u \in \Uc} \hat{C}_u$ CPU units per chain. The result follows.
\end{proof}
In the sequel we assume that $R \cdot C_s  / \max_{u \in \Uc} \hat{C}_u$ is an integer.
The following claim follows from the fact that \bu\ attempts to deploy 
any unplaced chain $\hv^u$ as close as possible to $p^u$.

\begin{lemma}
\label{lem:bu_all_paths_are_full}
If \bu\ tries to place chain $\hv^u$ on datacenter $s$, then all the datacenters in $\Sc_u$ below $s$ are almost-full. 
\end{lemma}


The following lemma follows directly from the order in which unplaced chains are considered at any datacenter.
\begin{lemma}
\label{lem:bu_located_earlier_so_within_Tu}
Consider two chains, $\hv^u$, $\hv^{u'}$, unplaced when considering datacenter $s \in \Sc_u \cap \Sc_{u'}$. Assume that
\begin{inparaenum}[(i)]
\item \bu\ places $\hv^{u'}$ on $s$ before it tries to place $\hv^u$ on $s$, and
\label{lem:bu_located_earlier_so_within_Tu:condition1}
\item $\Sc_u$ is almost-full after \bu\ terminates.
\end{inparaenum}
Then $\Sc_{u'}$ is also almost-full after \bu\ terminates. 
\end{lemma}
\begin{proof}
{    Let $s'$ be a datacenter in $\Sc_{u'}$. We will show that $s'$ is almost-full. Consider three cases, corresponding to $s'$ being a descendent of $s$, an ancestor of $s$, or equals to $s$.
    
    {\em Case 1:} $s'$ is a descendent of $s$. By condition (i), $s$ places $\hv^{u'}$. 
    Hence, by Lemma~\ref{lem:bu_all_paths_are_full}, all the datacenters in $\Sc_{u'}$ below $s$ are almost-full. It follows that $s'$ is almost-full.
    
   {\em Case 2:} $s'$ is an ancestor of $s$.
    By condition~\eqref{lem:bu_located_earlier_so_within_Tu:condition1} and the order in which \bu\ considers the chains (ln.~\ref{alg:bu:loop_over_u_begin} in Alg.~\ref{alg:BU}), we know that
    {    $\abs{\Sc_{u'} \setminus \Tc(s)} \leq \abs{\Sc_u \setminus \Tc(s)}$. In addition, $s \in  \Sc_u \cap \Sc_{u'}$. Combining the reasoning above, every ancestor of $s$ belonging to $\Sc_{u'}$ belongs also to $\Sc_u$. 
    }
    In particular, $s'$ is an ancestor of $s$ belonging to $\Sc_{u'}$, and therefore $s' \in \Sc_u$. As $\Sc_u$ is almost-full, $s'$ is almost-full. 

   {\em Case 3:} $s'=s$. As it is given that $s \in \Sc_u$ and $\Sc_u$ is almost-full, $s'$ is almost-full.
}
\end{proof}

The following lemma shows that if \bu\ fails, then 
there exists a set of chains requiring a total amount of CPU resources that is higher than the overall (augmented) processing capacity of the relevant datacenters. 

\tagged{short}{The proof, omitted for brevity, is by explicitly constructing a set of chains $\Hc'$ for which $\Tc_{\Hc'}$ is almost-full, and using Lemma~\ref{lem:num_of_chains_on_full} to obtain a lower bound on the number of chains in $\Hc'$.}

\begin{lemma}\label{lemma:if_bu_fails}
If \bu\ fails to place a chain, 
then there exists a set of chains $\Hc'$ s.t.\
\iftagged{short}{$ \abs{\Hc'} >    \frac{R}{\max_{u \in \Uc} \hat{C}_u}
    \sum_{s \in T_{\Hc'}} C_s $.}{ 
\begin{align}\label{eq:lemma:if_bu_fails}
    \abs{\Hc'} > 
    \frac{R}{\max_{u \in \Uc} \hat{C}_u}
    \sum_{s \in T_{\Hc'}} C_s \,.
\end{align}
}
\end{lemma}

\begin{algorithm}[t!]\caption{Constructing $\Hc'$ satisfying~\eqref{eq:lemma:if_bu_fails}}
\label{alg:lemma:if_bu_fails}
\scriptsize
\begin{algorithmic}[1]
    \State $\hv^{\tilde{u}}$ = chain that \bu\ failed to place, unmarked
    \label{alg:lemma:if_bu_fails:init_failing_u}
    \State $\Hc' = \set{\hv^{\tilde{u}}}$ 
    \label{alg:lemma:if_bu_fails:init_U_prime}
    \ForEach {unmarked chain $\hv^u \in \Hc'$}
        \label{alg:lemma:if_bu_fails:for_u}
        \State mark $\hv^u$
        \ForEach {un-visited datacenter $s$ from $\overline{s}(u)$ to $p^u$}
        \label{alg:lemma:if_bu_fails:for_s}
            \State mark $s$ as visited
            \ForEach {unmarked chain $\hv^{u'}$ placed on $s$}
            \Statex \Comment{w.l.o.g., in reverse order of placement by \bu}
            \label{alg:lemma:if_bu_fails:for_u_prime}
                \State add $\hv^{u'}$ to $\Hc'$, unmarked
                \label{alg:lemma:if_bu_fails:add_u_prime}
            \EndFor
        \EndFor
    \EndFor
\end{algorithmic}
\end{algorithm}

\begin{proof}
Let $\hv^{\tilde{u}}$ be the chain that \bu\ fails to place. 
Our construction of $\Hc'$ is detailed in Alg.~\ref{alg:lemma:if_bu_fails}, which works as follows. 
First, it initializes $\Hc'$ to 
$\set{\hv^{\tilde{u}}}$; then, it  repeatedly visits all the datacenters that are delay-feasible for chains that already belong to $\Hc'$, and adds to $\Hc'$ all the chains placed on those datacenters.
The algorithm finishes once there are no further datacenters to visit and no further chains to add, and returns $\Hc'$.

We first provide some intuition for Alg.~\ref{alg:lemma:if_bu_fails}, and for the validity of the claim. Consider Fig.~\ref{fig:BU_fails_Tu'}, and assume for simplicity that the network delay is zero, $R=1$, and $\max_{\hv^u \in \Hc} \hat{C}_u = \max_{s \in \Sc} C_s = 1$.
Hence, once \bu\ places a single chain on a datacenter, the datacenter becomes 
(almost) full.

Fig.~\ref{fig:BU_fails_Tu'} depicts a scenario where \bu\ fails to place chain $\hv^{u_6}$. Hence, algorithm~\ref{alg:lemma:if_bu_fails} assigns
$\tilde{u} = u_6$ (ln.~\ref{alg:lemma:if_bu_fails:init_failing_u}), and $\Hc' = \set{\hv^{u_6}}$ (ln.~\ref{alg:lemma:if_bu_fails:init_U_prime}). 
Next (ln.~\ref{alg:lemma:if_bu_fails:for_u}), the algorithm visits every datacenter on the path from $\overline{s}(u_6) = s_0$ to $p_{u_6} = s_3$. 
Namely, the algorithm visits $s_0$, $s_1$ and $s_3$, and consequently adds to $\Hc'$ all the chains placed on these datacenters (ln.~\ref{alg:lemma:if_bu_fails:for_u_prime}-\ref{alg:lemma:if_bu_fails:add_u_prime}). 
At this stage, we have $\Hc' = \set{\hv^{u_6}, \hv^{u_3}, \hv^{u_5}, \hv^{u_4}}$.  

Next, the algorithm visits the datacenters belonging to $\Sc_{u_3}$ that were not visited yet, namely, $s_2$ and $s_5$. Consequently, the algorithm adds $\hv^{u_2}$ and $\hv^{u_1}$, that are placed on these datacenters, to $\Hc'$. 

At this stage, after marking each unmarked chain in $\Hc'$, there are no more un-visited datacenters.
The algorithm then halts with $\Hc' = \set{\hv^{u_1}, \hv^{u_2}, \hv^{u_3}, \hv^{u_4}, \hv^{u_5}, \hv^{u_6}}$. By the definition of the potential placement tree ~\eqref{eq:def_pot_loc_tree}, $T_{\Hc'}$ consists of all the datacenters that are delay-feasible for chains belonging to $\Hc'$, namely, $T_{\Hc'} = \set{s_0, s_1, s_2, s_3, s_5}$. 
Recall that we assume that $R=1$ and $\max_{u \in \Uc} \hat{C}_u = C_s = 1$.
Hence, $\abs{\Hc'} = 6 > 5 = \frac{R}{\max_{u \in \Uc} \hat{C}_u} \sum_{s \in T_{\Hc'}}  C_s$, and the claim holds true.

We now turn to prove the claim. We first show that $\Tc_{\Hc'}$ is almost-full. We use induction over the chains added to $\Hc'$ (ln.~\ref{alg:lemma:if_bu_fails:add_u_prime} in Alg.~\ref{alg:lemma:if_bu_fails}). 
For the base, we have
$\Hc' = \set{\hv^{\tilde{u}}}$. As \bu\ fails to place $\tilde{u}$, we know that 
$\Tc_{\Hc'} = \Tc_{\set{\hv^{\tilde{u}}}} = \set{\Sc_{\tilde{u}}}$ is almost-full. 

For the induction step, 
consider a chain $\hv^{u'}$ that is added to $\Hc'$ while considering datacenter $s$ (ln.~\ref{alg:lemma:if_bu_fails:add_u_prime} in Alg.~\ref{alg:lemma:if_bu_fails}).
Let $\hv^u$ denote the concrete chain considered in ln.~\ref{alg:lemma:if_bu_fails:for_u} of the algorithm at that iteration. 
 Alg.~\ref{alg:lemma:if_bu_fails} visits datacenters in a top-down order (that is, advancing towards the leaves), while \bu\ visits datacenters in DFS-order, i.e., in bottom-up order 
(that is, from the leaves towards the root). 
Furthermore, Alg.~\ref{alg:lemma:if_bu_fails} handles chains in reverse order of placement by \bu. 
It therefore follows that both $\hv^u$ and $\hv^{u'}$ were unplaced when \bu\ considered datacenter $s$, and $s \in \Sc_u \cap \Sc_{u'}$.
Hence, \bu\ placed $\hv^{u'}$ on $s$ before it tried to place $\hv^u$.
By the induction hypothesis, $\Sc_u$ is almost-full.
Hence, by Lemma~\ref{lem:bu_located_earlier_so_within_Tu}, $\Sc_{u'}$ is also almost-full. 
Hence, after $\hv^{u'}$ is added to $\Hc'$, $\Tc_{\Hc'}$ is still almost-full.
We therefore proved by induction that when Alg.~\ref{alg:lemma:if_bu_fails} halts, 
$\Tc_{\Hc'}$ is almost-full.

As $\Tc_{\Hc'}$ is almost-full, every datacenter $s \in \Tc_{\Hc'}$ is almost-full. Hence, by Lemma~\ref{lem:num_of_chains_on_full},  \bu\ placed at least $\frac{R \cdot C_s}{\max_{u \in \Uc} \hat{C}_u}$ chains on each such datacenter $s$.
Hence, the overall number of chains that \bu\ placed on $\Tc_{\Hc'}$ is at least
$\frac{R}{\max_{u \in \Uc} \hat{C}_u}\sum_{s \in \Tc_{\Hc'}} C_s $. 
By the construction of $\Hc'$, {\em every} chain that \bu\ placed on $\Tc_{\Hc'}$ belongs to $\Hc'$. Hence, $\Hc'$ contains all the $
\frac{R}{\max_{u \in \Uc} \hat{C}_u}
\sum_{s \in T_{\Hc'}} C_s$ chains that \bu\ placed on $\Tc_{\Hc'}$. In addition, $\Hc'$ includes $\hv^{\tilde{u}}$, and the result follows.
\end{proof}

After having characterized the scenarios where \bu\ fails, the following lemma shows
a sufficient condition for the problem being
infeasible without resource augmentation (recalling the definition of $\tilde{\mu}$ in Eq.~\eqref{eq:def:min_required_CPU}).

\begin{lemma}\label{lemma:then_opt_fails}
If there exists a set of chains $\Hc'$ s.t.\
$\abs{\Hc'} 
\cdot \tilde{\mu}
> 
\sum_{s \in T_{\Hc'}} C_s
$, 
then the problem is infeasible without resource augmentation.
\end{lemma}

\begin{proof}
Any feasible solution must allocate for each chain $\hv^u \in \Hc'$ at least 
$\tilde{\mu}$
CPU units on some datacenter(s) in $T_{\Hc'}$. The result follows. 
\end{proof}


The theorem below, which is our main result in this section, now follows from combining Lemma~\ref{lemma:if_bu_fails} and Lemma~\ref{lemma:then_opt_fails}.

\begin{theorem}
\label{thm:max_resource_augmentation}
Assume \bu\ uses in each datacenter a multiplicative resource augmentation of processing capacity
\iftagged{short}{$R = 
{\max_{u \in \Uc}  \hat{C}_u}/
{\tilde{\mu}}$.}{ 
\begin{equation}\label{eq:def:bu_max_rsrc_aug_needed}
R = 
\frac
{\max_{u \in \Uc}  \hat{C}_u}
{\tilde{\mu}} \,.
\end{equation} 
}
Then, \bu\ finds a feasible solution whenever such a solution exists in a system without resource augmentation.
\end{theorem}

\begin{proof}
Assume that \bu\ fails. 
Assigning the value of $R$ ~\eqref{eq:def:bu_max_rsrc_aug_needed} in Lemma~\ref{lemma:if_bu_fails}, there exists a set of chains $\Hc'$ s.t. 
\begin{align}
\abs{\Hc'} & > 
\frac
{1}
{\tilde{\mu}} 
\sum_{s \in T_{\Hc'}} C_s. \notag
\end{align}
Applying Lemma~\ref{lemma:then_opt_fails}, the result follows.
\end{proof}

\section{Algorithmic solution to the \migProb }\label{sec:top_lvl}

In this section, we use the algorithms for the CPU allocation problem (Sec.~\ref{sec:alloc}) and for the placement problem (Sec.~\ref{sec:bu}) as building blocks for solving the \migProb. %
We first present an algorithm that takes as input a feasible solution for \migProb, and greedily reduces its cost (Sec.~\ref{sec:push_up}).
Later, 
we present our algorithmic solution for \migProb, \algtop, which  targets  a minimal cost solution, using the minimum amount of resource augmentation (Sec.~\ref{sec:subsec:top-lvl-alg}).

\subsection{Reducing costs}
\label{sec:push_up}

\begin{algorithm}[t!]
    \caption{\small \pushUp($\Gc, {\tilde{\Hc}}, \set{\muv^{u,s}}_{u,s} 
    \set{\Sc_u}_{u}$, feasible $\yv, \av$)}
    \label{alg:push_up}
    \scriptsize
    \begin{algorithmic}[1]

        \While {True}\label{alg:push_up:while1_begin}
            \For {$\hv^u \in { \tilde{\Hc}}$ in non-decreasing order of $\yv(u,s) \cdot  \muv^{u,s}$}\label{alg:push_up:pick_chain}
                \State push $\hv^u$ as high as possible in the tree as long as this decreases $\hv^u$'s cost, and update $\yv$, $\av$ accordingly\label{alg:push_up:push_up}
            \EndFor
            \If {did not succeed to push up any chain} \label{alg:push_up:if_then_stop}
                \State {\bf break}
             \label{alg:push_up:break}
            \EndIf
        \EndWhile \label{alg:push_up:while1_end}
        \State return $\yv, \av$
    \end{algorithmic}
\end{algorithm}

We first note that the placement costs and latency constraints are {\em separable} between different chains. That is, for any two chains $\hv^u$, $\hv^{u'}$, if $\hv^u$ is placed on datacenter $s$ using CPU allocation $\muv^{u,s}$,  the cost of the placement and CPU allocation of $\hv^u$, and the latency inflicted on $\hv^u$, are independent of the placement and CPU allocation of $\hv^{u'}$.

Based on this observation, we devise our algorithm, {\em PushUp} (or {\em \pushUp} for short), for reducing the cost of a given feasible solution.
{\pushUp, formally described in Alg.~\ref{alg:push_up},  scans the provided feasible solution, and greedily tries to improve it by pushing each chain as high up as possible in the network topology (while reducing the cost).}
\pushUp\ considers chains in non-increasing order of the CPU units allocated to them, thus prioritizing chains that potentially offer the largest gain from being pushed-up.

Intuitively, this ``push-up'' operation serves two goals: 
\begin{inparaenum}[(i)]
\item decreasing the total cost and, 
\item decreasing the total number of migrations, as a datacenter located higher in the tree can serve users located in a larger physical area. \end{inparaenum}
Importantly, \pushUp\ always outputs a feasible solution. Hence, one may run \pushUp\ as a heuristic {to improve a feasible solution found by any algorithm solving the placement problem.}

\paragraph*{Run-time analysis}
\pushUp\ moves a chain to another datacenter only if this reduces the cost. Hence, once it moves chain $\hv^u$ from a datacenter, it never moves it back (due to separability). 
\tagged{short}{It follows that the running time of \pushUp\ is
$\tilde{O}(\abs{\Hc}^2\cdot D(\Gc)^2)$.
}
\tagged{long}{Thus, the number of iterations of the {\em while} loop is at most $\abs{\tilde{\Hc}} \times \text{height}(\Gc)$.
As each iteration requires $O(\abs{\tilde{\Hc}} D(\Gc))$ steps, the time complexity of \pushUp\ is    $O(\abs{\tilde{\Hc}}^2 D(\Gc)^2)$.}

\begin{algorithm}[t!]
    \caption{\small \algtop\ ($\Gc, \Hc, \tilde{\Hc}, \av$)
    }
    \label{alg:algtop}
    \scriptsize
    \begin{algorithmic}[1]
        \State $\set{\muv^{u,s}}_{{\hv^u \in \tilde{\Hc}}, s \in \Sc}, \label{bupu1}
        \set{\Sc_u}_{\hv^u \in \tilde{\Hc}}$ = \cpuall ($\Gc, { \tilde{\Hc}}$)
        \State $\yv, \av$ = \bu\ ($\Gc, {\tilde{\Hc}}, \set{\muv^{u,s}}_{{\hv^u \in \tilde{\Hc}}, s \in \Sc}, \set{\Sc_u}_{{ \hv^u \in \tilde{\Hc}}}, \av$)\label{bupu2}
        \If {$\yv \not\equiv 0$} \Comment{found a feasible solution}
            \State $\yv, \av$ = \pushUp($\Gc, { \tilde{\Hc}}, \set{\muv^{u,s}}_{{\hv^u \in \tilde{\Hc}}, s \in \Sc}, \set{\Sc_u}_{{ \hv^u \in \tilde{\Hc}}}, \yv, \av$)\label{bupu3}
        \Else
            \State $\yv, \av =$ feasible solution with minimal resource augmentation\label{bupu4}
            \Comment{binary search using \bu($\Gc, \Hc, \set{\muv^{u,s}}_{u \in \Uc, s \in \Sc}, \set{\Sc_u}_{u \in \Uc}, \av$)}
            \State $\yv, \av$ = \pushUp($\Gc, \Hc, \set{\muv^{u,s}}_{u \in \Uc, s \in \Sc}, \set{\Sc_u}_{u \in \Uc}, \yv,\av$)\label{bupu5}
 \label{alg:algtop:while_begin}
         \EndIf 
    \end{algorithmic}
\end{algorithm}

\begin{figure*}[tb!]
\centering
\includegraphics[width=15.5cm]{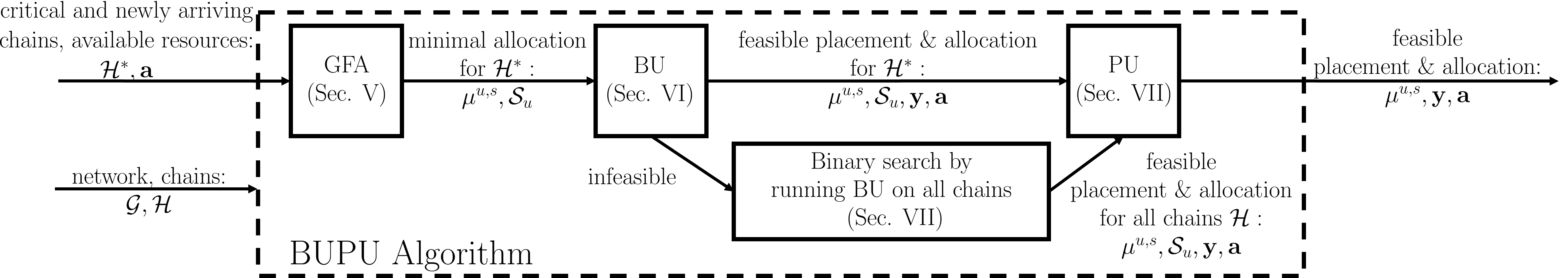}
\caption{Structure of the \algtop\ algorithmic solution.}\label{fig:toplevel}
\end{figure*}

\subsection{The \algtop\ algorithm}
\label{sec:subsec:top-lvl-alg}

Theorem~\ref{thm:max_resource_augmentation} provides an upper bound on the amount of resource augmentation that \bu\ requires to find a feasible solution for the placement problem. However,  the amount of resource augmentation needed in practice might  be significantly lower than that provided by the theorem.  
{Our algorithm for solving the \migProb, Bottom-Up Push-Up (\algtop), aims at finding a feasible, minimal-cost solution while using a minimal amount of resource augmentation. The algorithm 
is summarized in Fig.~\ref{fig:toplevel}, 
and formally defined in Alg.~\ref{alg:algtop}.

The algorithm gets as input the structure of the network $\Gc$, and the set of chains $\Hc$. These inputs are used by all the modules of the algorithm (in Fig.~\ref{fig:toplevel} we omit the arrows connecting $\Gc$ and $\Hc$ to each module to improve clarity). 
In addition, \algtop\ gets the set of critical and newly arriving chains $\Hc^* \subseteq \Hc$, and the currently available resources $\av$.

\algtop\ first runs \cpuall\ to obtain for every critical or newly arriving chain $\hv^u$ its list of delay-feasible datacenters $\Sc_u$, and the minimal CPU allocation required for placing $\hv^u$ on each server belonging to $\Sc_u$ (ln.~\ref{bupu1}). Recall that this minimal CPU allocation is denoted $\muv^{u, s}$.
Next, \algtop\ runs \bu\ to see if a feasible solution exists given the current resource augmentation (ln.\ref{bupu2}).
If so, then \pushUp\ is applied on the set of critical and newly arriving chains to reduce the cost (ln.~\ref{bupu3}).
However, if \bu\ does not find a feasible solution when considering (re)placing only critical and newly arriving chains, the algorithm turns to solve the placement problem for {\em all} chains in the system, to ensure feasibility.
In this case, \algtop\ might end up ``reshuffling'' the placement of many of the chains. 
To adjust the amount of resource augmentation, \algtop\ performs a binary search for the minimal amount of resource augmentation required for obtaining a feasible solution (ln.~\ref{bupu4}).
Given a feasible solution that minimizes the amount of resource augmentation, \algtop\ runs \pushUp\ on all the chains, to reduce the solution cost (ln.~\ref{bupu5}).

Intuitively, using more resource augmentation allows locating more chains in the cloud, thus reducing both the computation costs, and the need for future migrations. 
Hence, there exists a tradeoff between the amount of resource augmentation, and the solution cost. 
We study this tradeoff in Sec.~\ref{sec:sim:cost}.
}

\section{Numerical Evaluation}\label{Sec:sim}
In this section, we evaluate BUPU against existing alternatives and highlight some trade-offs, thus providing insights that go beyond our analytical results. 

\subsection{Simulation settings}\label{Sec:sim_settings}
\newcommand{\subfloatWidth}{0.445\columnwidth}
\newcommand{\subfloatNarrowWidth}{0.23\columnwidth}
\newcommand{\subfloatNarrowerWidth}{0.21\columnwidth}

We now describe the settings of our baseline scenario, and later vary some of them to study their impact on performance.

\begin{table*}
    \scriptsize
    \centering
    \caption{Simulated scenarios}
    {    \begin{tabular}{l|c|c|c|c|c|c|c|c|}
        \cline{2-9}
        & simulated & telecom & antennas & \#distinct & avg. traffic & avg.\ density & linear density & speed \\
        & area (km$^2$)& provider & (\#PoAs) & vehicles & (\#vehicles/s) & (\#requests/km$^2$) & (\#vehicles/km) & (km/h) \\
        \hline
        \multicolumn{1}{|l|}{Luxemburg} & $6.8 \times 5.7$ & Luxembourg Post & 1524 & 25,497 & 2191 & 56 & 2.48 & 15.4 \\
        \hline
        \multicolumn{1}{|l|}{Monaco} & $3.1 \times 1$ & Monaco Telecom & 231 & 13,788 & 7121 & 2297 & 32.7 & 9.0 \\
        \hline
    \end{tabular}
    }
    \label{tab:sim_settings}
\end{table*}

\textbf{Service area.}  
We consider two real-world scenarios, capturing mobility patterns with different characteristics.
We focus on the centers of the cities of
{{\em Luxembourg} and {\em Principality of Monaco},}
using mobility traces~\cite{Luxembourg, Monaco}
and real-world antenna locations, publicly available in~\cite{OpenCellid}. For each simulated area, we consider the antennas of the cellular telecom provider having the largest number of antennas in the simulated area.
The settings of the service areas are detailed in Table~\ref{tab:sim_settings}, which details also some traffic parameters, which we shortly explain.
{For both traces we consider the rush hour period between 7:30 am and 8:30 am.}

\textbf{Network and datacenters.}
Our simulated networks are illustrated in Fig.~\ref{Fig:sim_netw}, where each PoA is co-located with a leaf datacenter. 
At each decision period, each service chain is associated with the nearest PoA datacenter. Figures~\ref{fig:Lux_Voronoi} and~\ref{fig:Monaco_Voronoi} depict a Voronoi diagram of the simulated areas. 
The network topology connecting the datacenters is a 6-height tree, structured as follows. 
Denote a topology level by $\ell\in\{0, 1, \dots, 5\}$, with $\ell=0$ corresponding to the leaf datacenters (co-located with the PoAs), and $\ell=5$ corresponding to the root datacenter. 

We build levels $5, 4, 3, 2, 1$ in Luxembourg's network by recursively partitioning the simulated area into $1, 4, 16, 64, 256$ rectangles, respectively.
Namely, each datacenter at level $\ell \in \set{5,4,3,2}$ has 4 children, each of them responsible for 1/4 of its area.
The children of each datacenter at level $\ell=1$ are the PoAs (antennas and datacenters), of the telecom provider Luxemburg Post, located in its rectangle. 
Finally, if no PoAs exist in a certain rectangle, the respective datacenters are pruned from the tree.
Figures~\ref{fig:Lux_map} and~\ref{fig:Lux_tree} detail the four top levels in Luxembourg's network. 

Monaco's network is built in a similar fashion to that of Luxembourg. However, as Monaco's center makes a long, narrow rectangle (3.1 km $\times$ 1 km), at the top-level (level 5), we partition the simulated area into three horizontal almost-square rectangles. We build levels $3,2,1$ by recursively partitioning these squares into quadrants, so that levels $4,3,2,1$ comprise $3, 12, 48$, and 192 datacenters, respectively. Finally, the leaf datacenters are the 231 antennas (and co-located datacenters) of Monaco Telecom. 
Figures~\ref{fig:Monaco_map} and~\ref{fig:Monaco_tree} detail the three highest levels in Monaco's network. Note the pruned datacenters in  Fig.~\ref{fig:Monaco_tree}, corresponding to areas where no PoAs exist (e.g., areas in the sea).

The datacenter processing capacity depends upon the level $\ell$ and is set to $\ell \cdot C_\text{cpu}$, for a given  $C_\text{cpu}$, to reflect the increase of datacenter capacities when moving from the edge to the cloud.
The total link delay  $\tau(i,j)$ is set to 2~ms for every link $(i,j)$, resulting in a maximum round trip network delay of 20~ms from the PoA (level 0) to the root (level 5).
{Given the minimal allocations obtained by GFA, we use the ILP relaxation of the problem, as discussed in Sec.~\ref{sec:roadmap_placement}. 
Using this ILP, through binary search, one can find the minimal $C_\text{cpu}$ for which there exists a feasible solution for the relaxation. This capacity is denoted by $\hat{C}_\text{cpu}$, and serves as the baseline value of the CPU available at every server, where we consider resource augmentation with respect to this value.
We note that $\hat{C}_\text{cpu}$ serves as a lower bound on the required capacity for such capacity allocation settings.}

\begin{figure*}[tb!]
\centering
    \subfloat[\label{fig:Lux_Voronoi}Luxembourg's cell coverage]{
        \includegraphics[height=3.5cm]{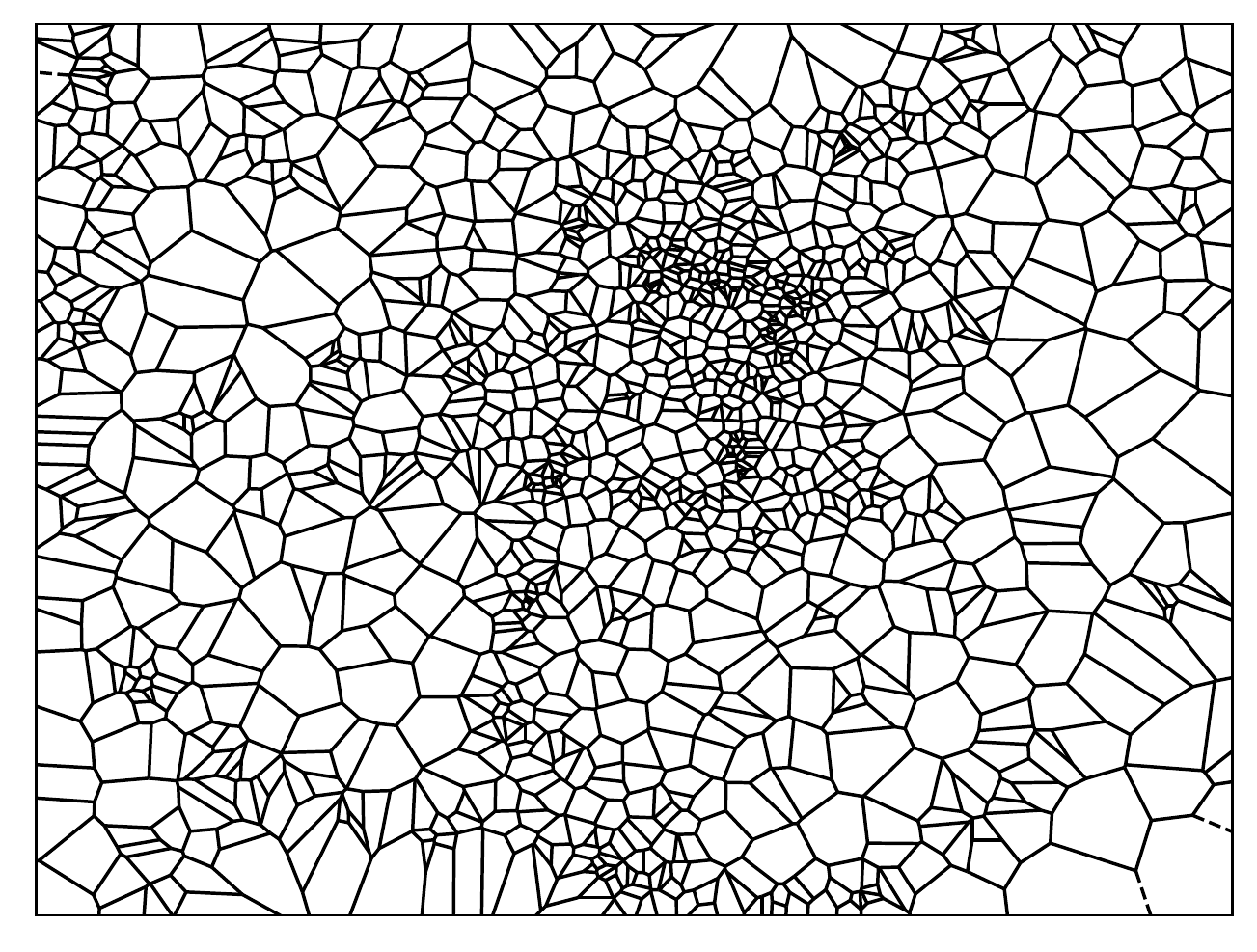}
    }
    \subfloat[\label{fig:Lux_map}Luxembourg's datacenters coverage area]{
        \includegraphics[height=3.5cm]{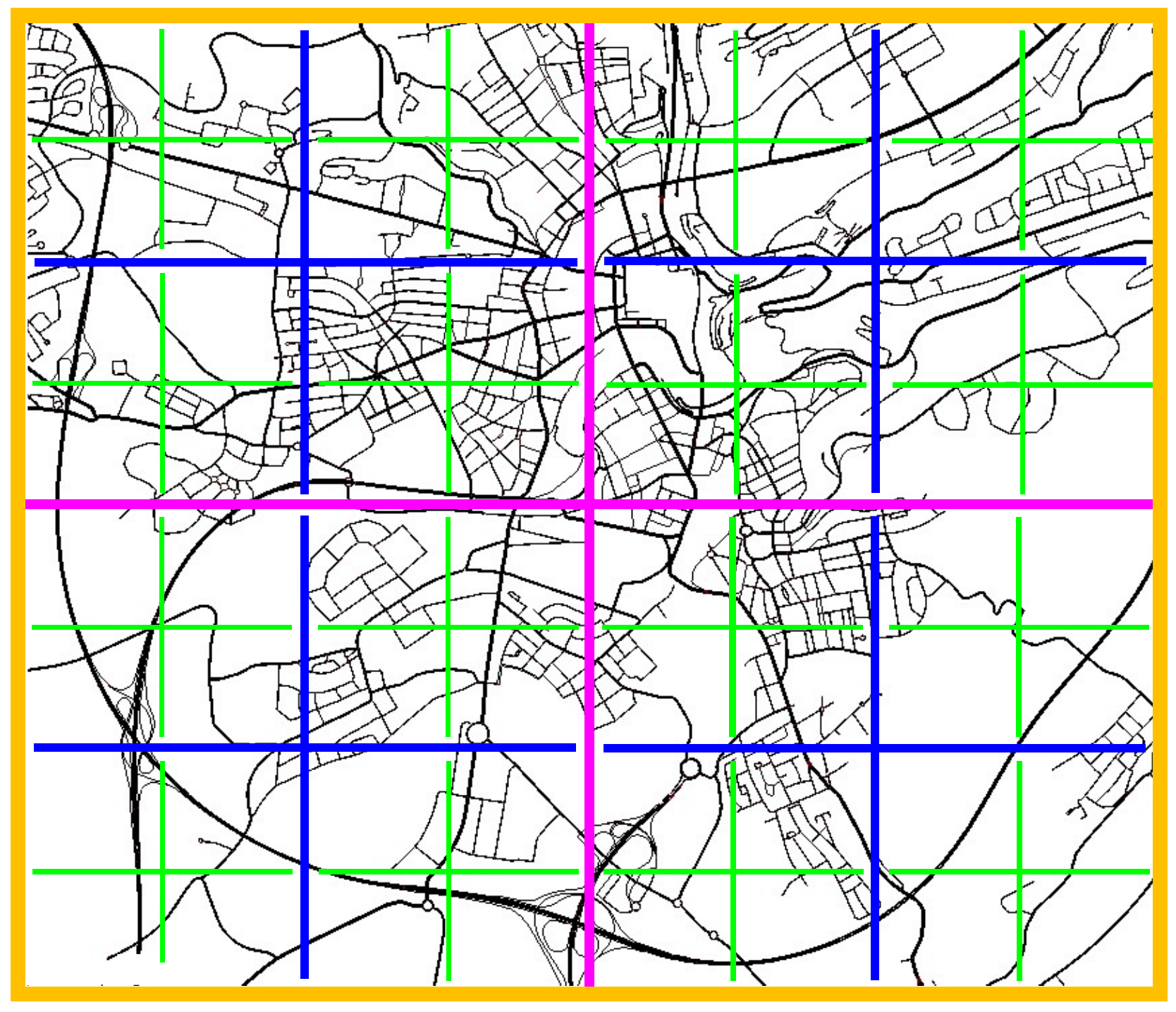}
    }
    \subfloat[\label{fig:Lux_tree}Luxembourg's datacenters network]{
        \includegraphics[height=3cm,width=7.5cm]{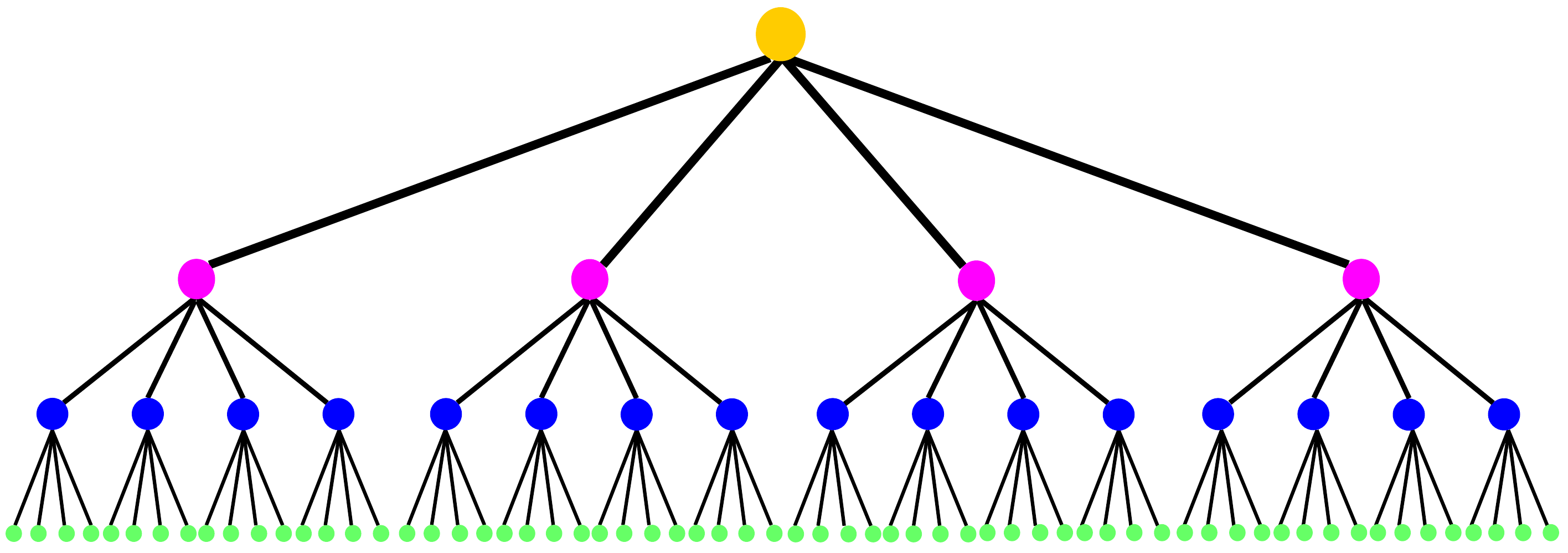}
    }
    
    \vspace{0.25 cm}
    
    \subfloat[\label{fig:Monaco_Voronoi}Monaco's cell coverage]{
        \includegraphics[width=6.1cm, height=2.3cm]{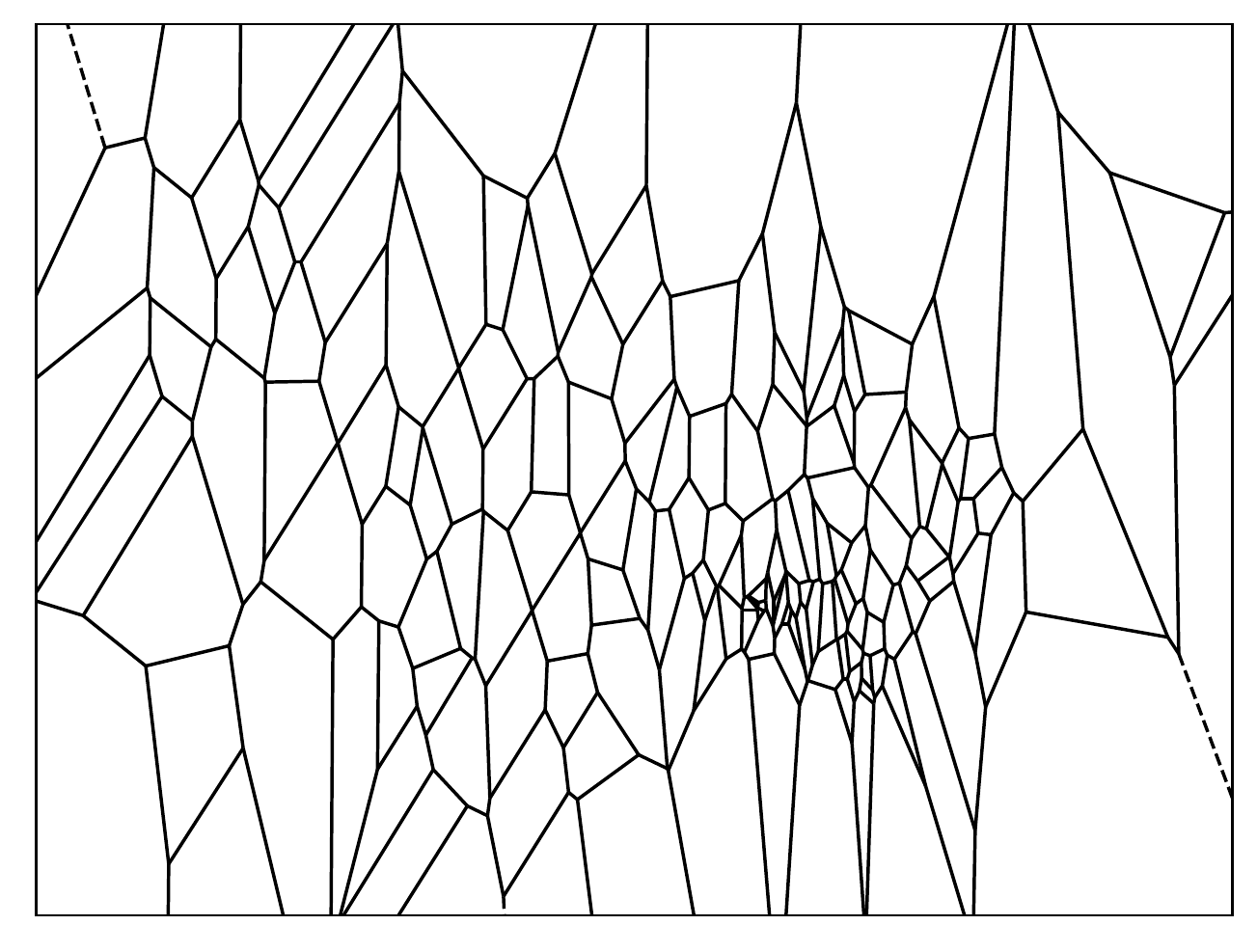}
    }
    \subfloat[\label{fig:Monaco_map}Monaco's datacenters coverage area]{
        \includegraphics[width=6.1cm,height=2.3cm]{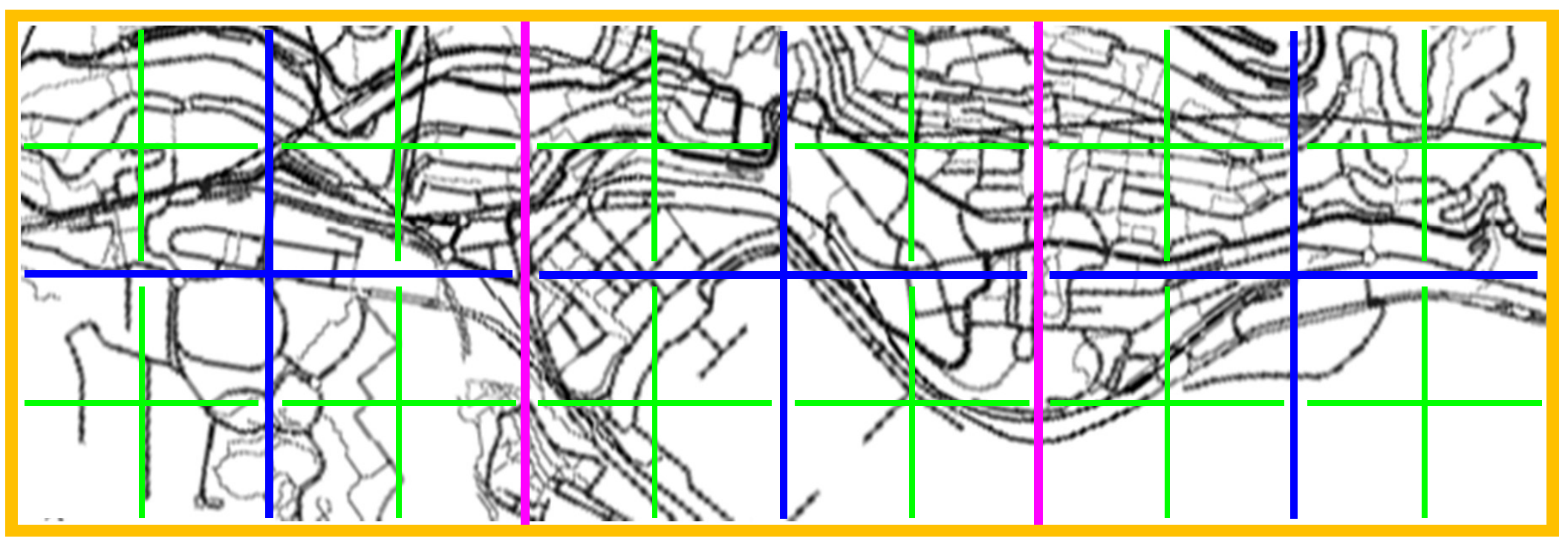}
    }
    \subfloat[\label{fig:Monaco_tree}Monaco's datacenters network]{
        \includegraphics[width=5cm,height=2.3cm]{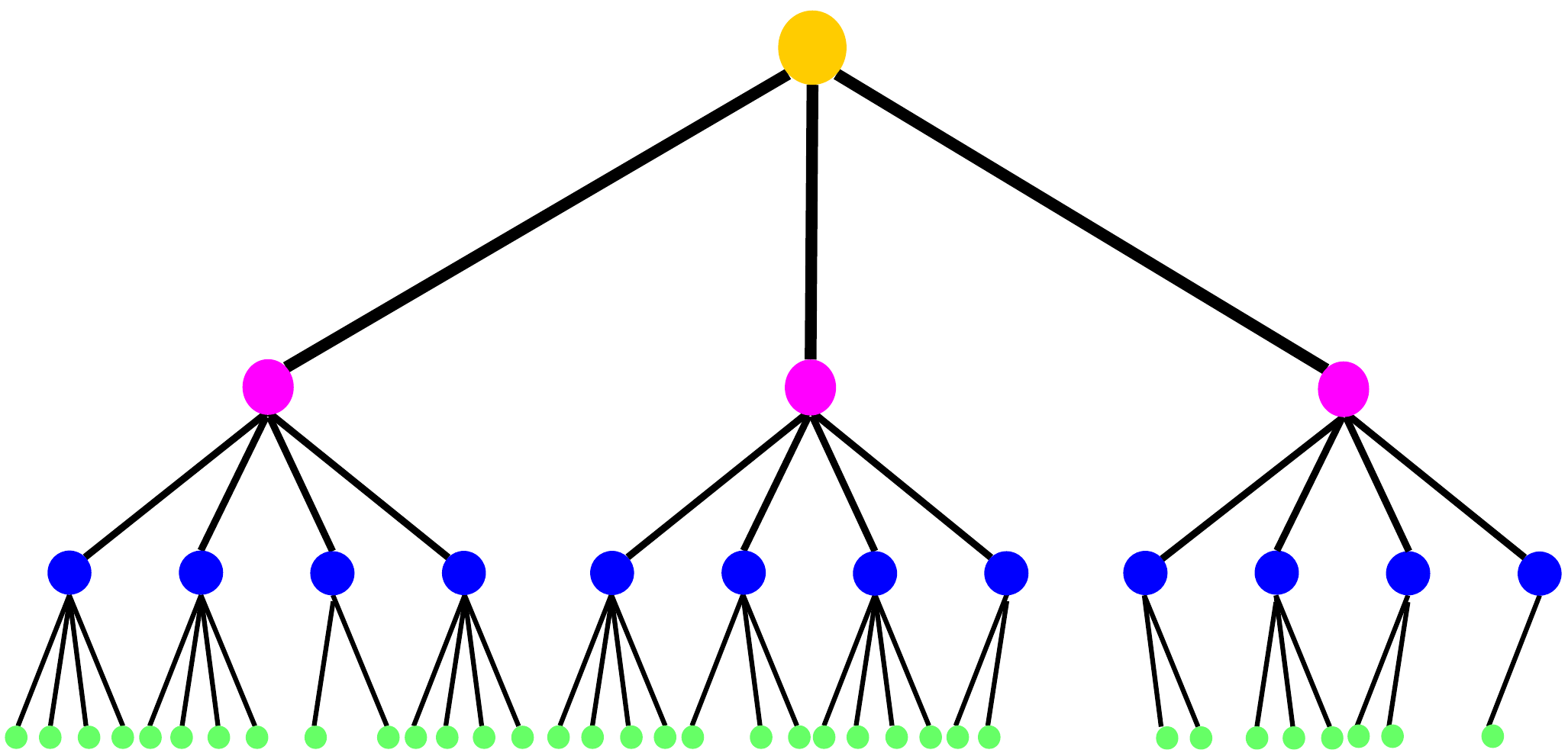}
    }
     \caption{\label{Fig:sim_netw}The service network in Luxembourg (top) and Monaco (bottom). Figures~\ref{fig:Lux_Voronoi} and~\ref{fig:Monaco_Voronoi} present Voronoi diagrams of the PoAs, corresponding to the leaves (level 0). Figures~\ref{fig:Lux_map} and~\ref{fig:Monaco_map} illustrate the iterative partition of the area into rectangles. The rectangles highlighted in yellow, pink, blue, and green, correspond to levels 5 (root), 4, 3, and 2, respectively, in the network, as depicted in Figures~\ref{fig:Lux_tree} and~\ref{fig:Monaco_tree}. Some of the leaves in Monaco's datacenters network are pruned from the tree, as no PoAs exist in the respective rectangles.}

\end{figure*}
\captionsetup[subfloat]{labelformat=parens}

\textbf{Traffic and mobility.}
To characterize the traffic in the simulated scenarios, we consider the {\em linear vehicle density}, defined as the average number of cars per km of lane\footnote{A {\em lane} is a uni-directional path on the road; a single road may contain one or more lanes in each direction.}.  
Fig.~\ref{fig:lin_density} depicts the linear vehicle density within the ``coverage area'' of each datacenter at levels $1,2,3,4$. 
Note the significantly higher values in Monaco's scenario, capturing the heavier traffic in this network. 

Fig.~\ref{fig:left} captures the average number of cars that have moved to another rectangle (or left the simulated area) during the sampling period. This can be seen as the offered migration rate, since a car changing cell may dictate migrating the corresponding service chain.
The higher traffic density in Monaco is translated to lower mobility. 

Consider again Table~\ref{tab:sim_settings}. The table presents the average number of active vehicles, and the average demand density, defined as the number of service requests (vehicles) per square kilometer.
Observe that Monaco's higher linear vehicle density results in a significantly lower average speed (9.0 km/h in Monaco vs. 15.4 km/h in Luxembourg). 

\captionsetup[subfloat]{labelformat=empty}
\begin{figure}[t!]
    \centering
    \subfloat [\label{fig:Lux_lin_density_L1}Lux., level 1] {    
        \includegraphics[height=1.75cm, width=\subfloatNarrowWidth]
            {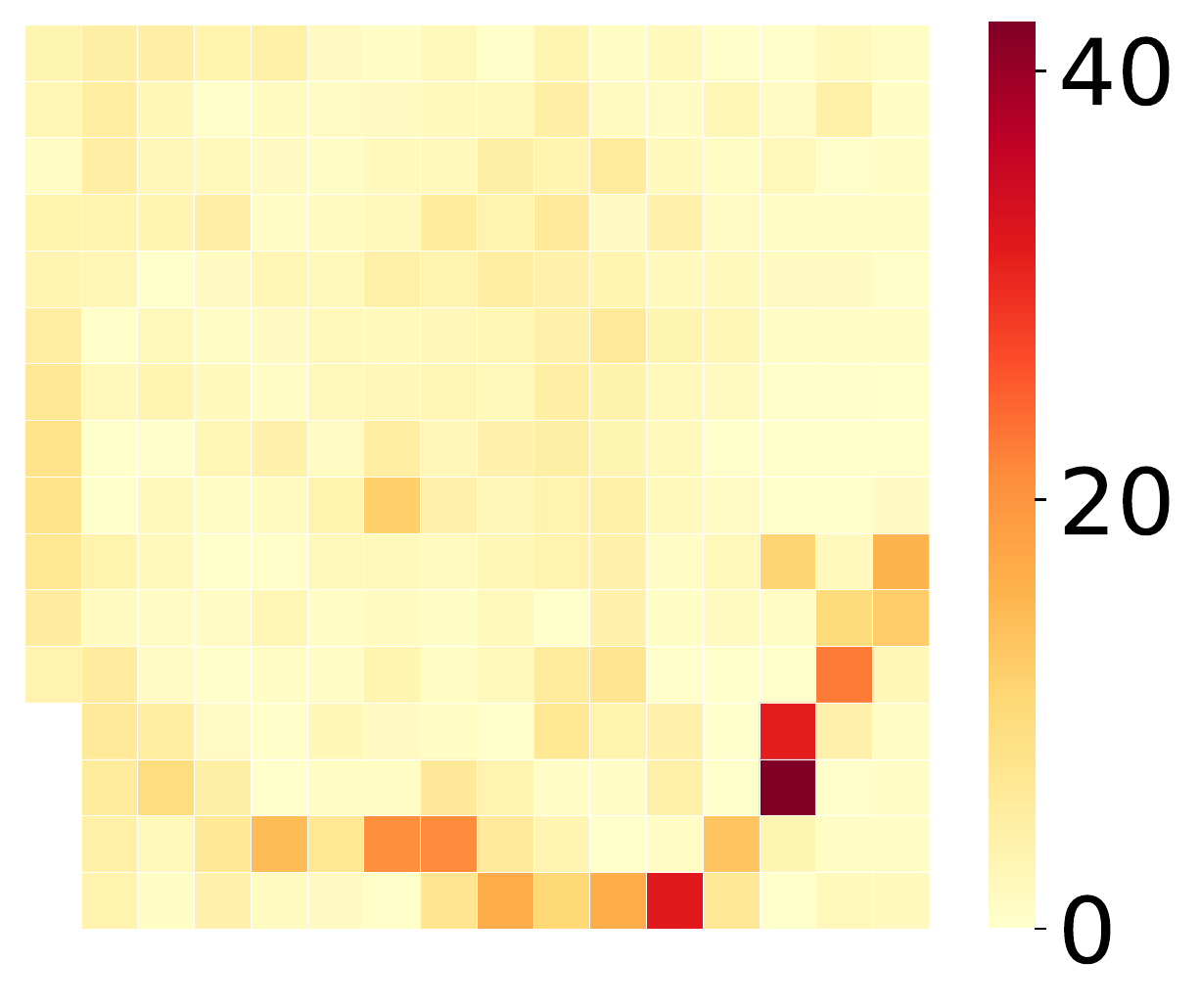}
    }
    \subfloat [\label{fig:Lux_lin_density_L2}Lux., level 2] {    
        \includegraphics[height=1.75cm, width=\subfloatNarrowWidth]
            {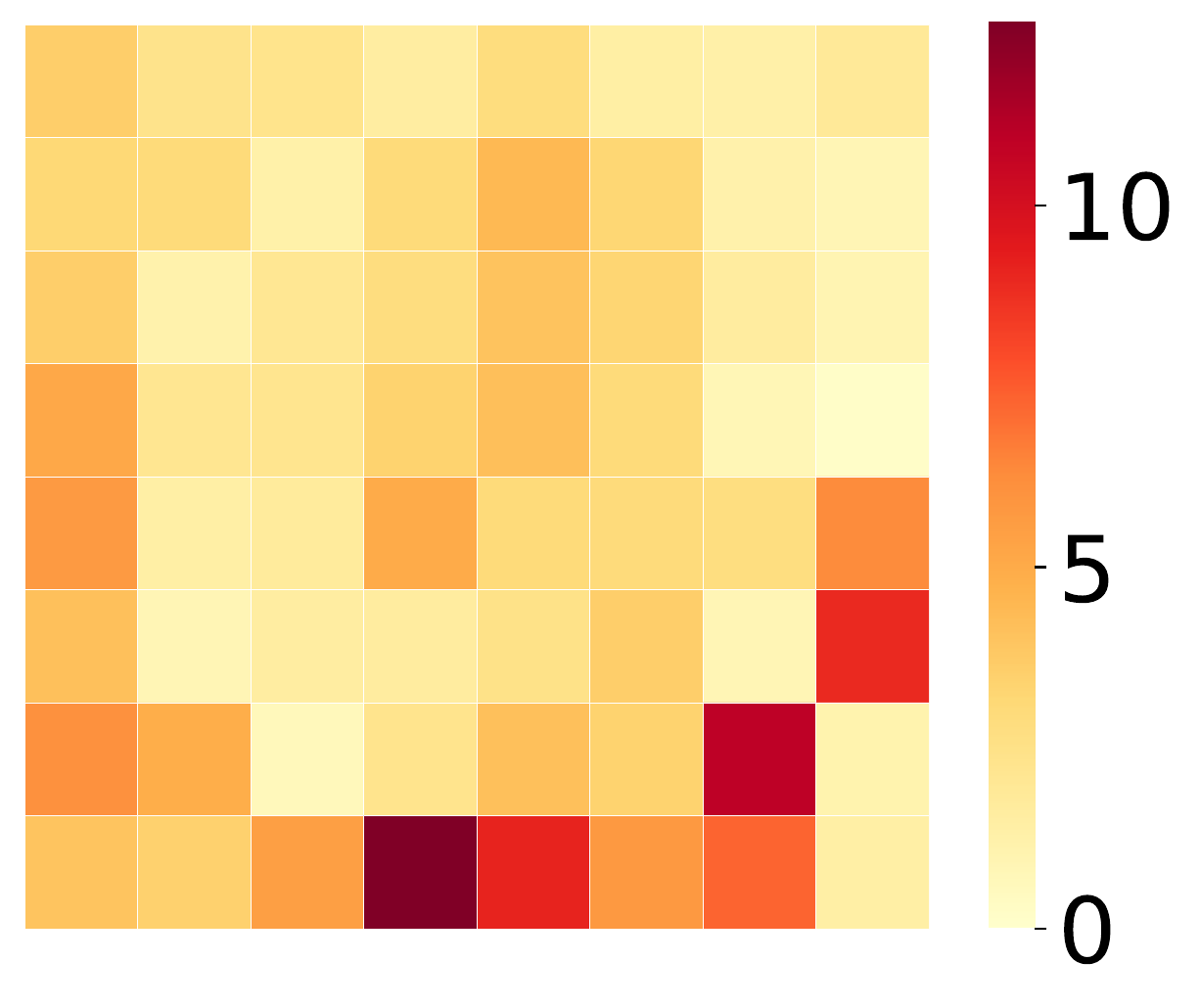}
    }
    \subfloat [\label{fig:Lux_lin_density_L3}Lux., level 3] {    
        \includegraphics[height=1.75cm, width=\subfloatNarrowWidth]
            {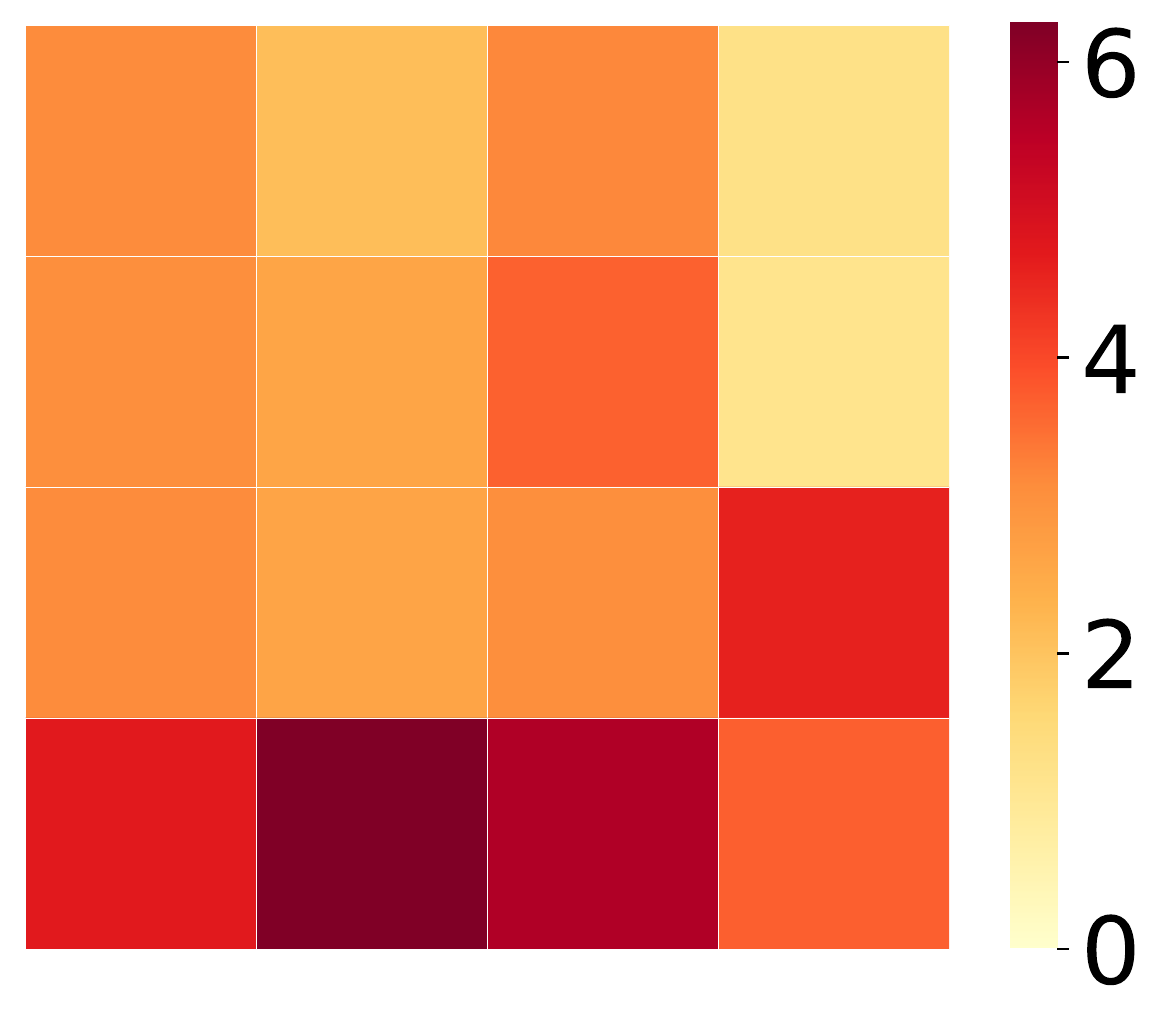}
    }
    \subfloat [\label{fig:Lux_lin_density_L4}Lux., level 4] {    
        \includegraphics[height=1.75cm, width=\subfloatNarrowWidth]
            {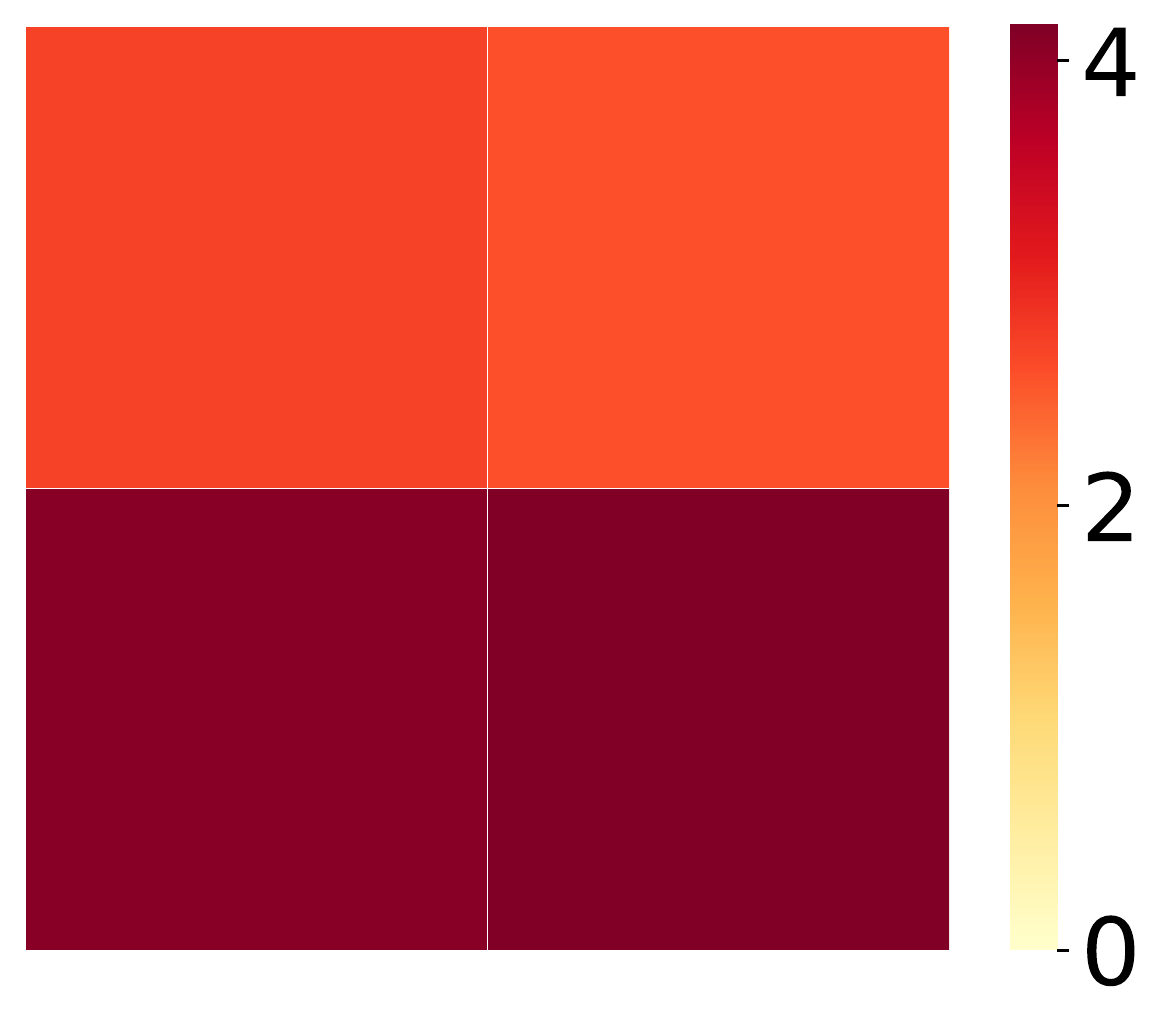}
    }
    
    \vspace{0.25 cm}

    \subfloat [\label{fig:Monaco_lin_density_L1}Monaco, level 1] {
        \includegraphics[height=0.75cm, width=\subfloatNarrowWidth]{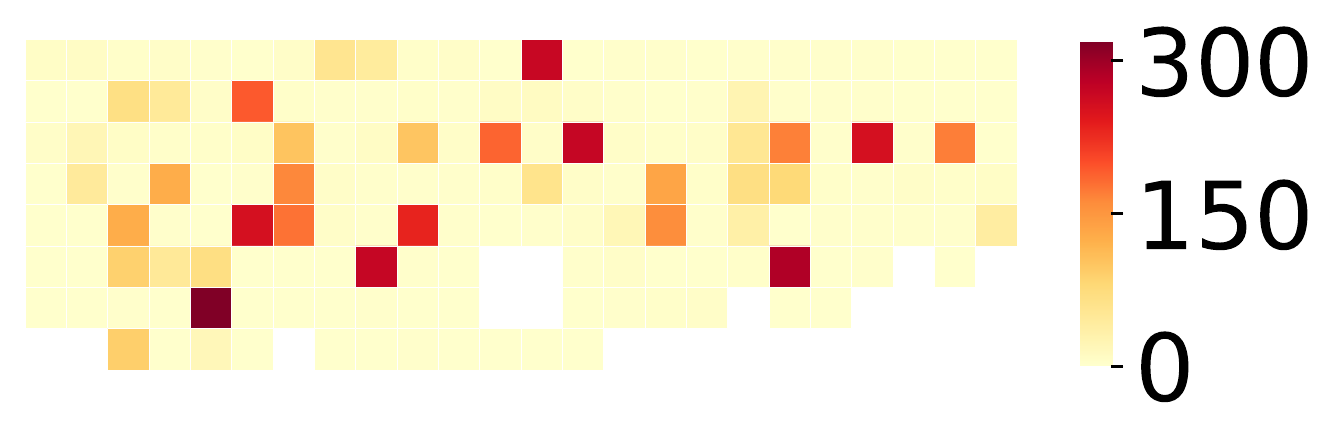}
    }
    \subfloat [\label{fig:Monaco_lin_density_L2}Monaco, level 2] {    
        \includegraphics[height=0.75cm, width=\subfloatNarrowWidth]{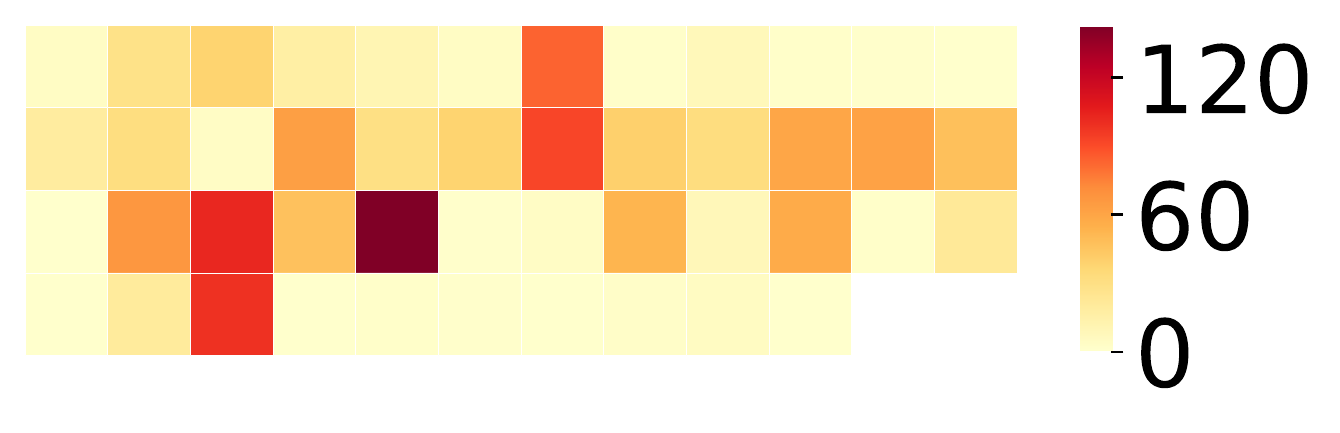}
    }
    \subfloat [\label{fig:Monaco_lin_density_L3}Monaco, level 3] {
        \includegraphics[height=0.75cm, width=\subfloatNarrowWidth]{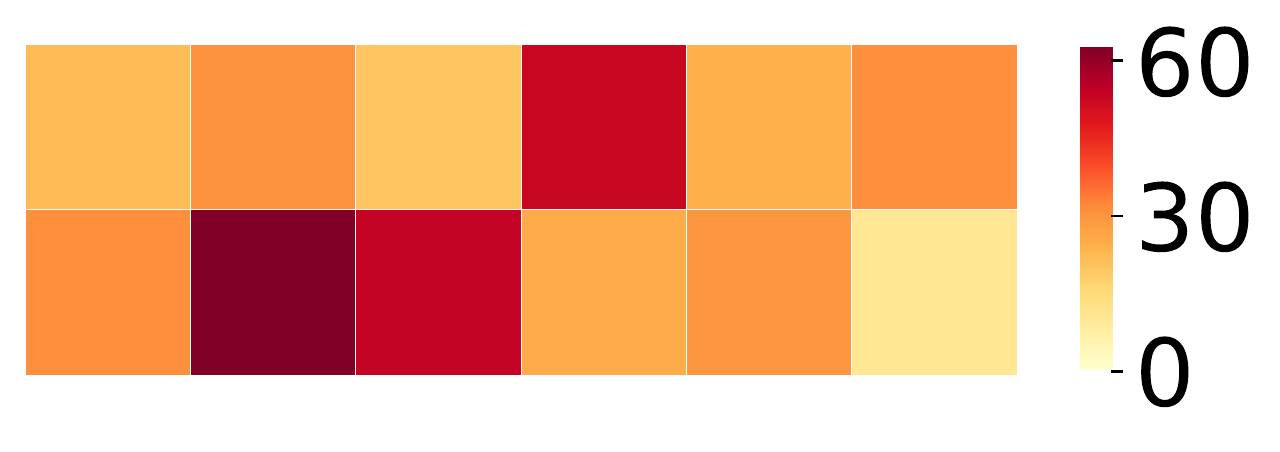}
    }
    \subfloat [\label{fig:Monaco_lin_density_L4}Monaco, level 4] {
        \includegraphics[height=0.75cm, width=\subfloatNarrowWidth]{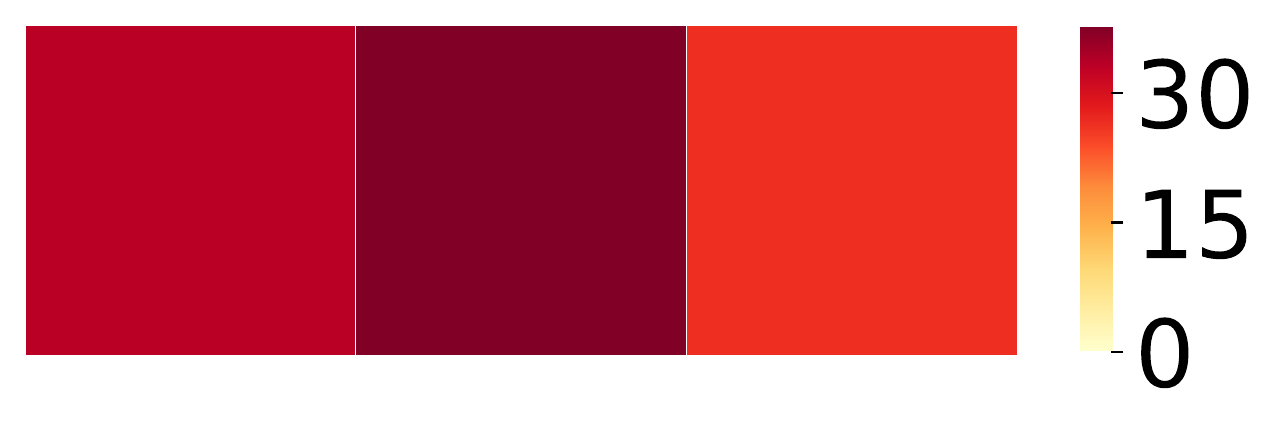}
    }
    \caption{Average linear vehicular density 
    [vehicles/km]
    during the 7:30-8:30\ am interval in each rectangle in Luxembourg (top) and Monaco (bottom). White rectangles correspond to areas where no roads exist.
}\label{fig:lin_density}
\end{figure}
\captionsetup[subfloat]{labelformat=parens}

\captionsetup[subfloat]{labelformat=empty}
\begin{figure}[t!]
    \centering
    \subfloat [Lux., level 1\label{fig:lux:left_L1}] {    
        \includegraphics[height=1.75cm, width=\subfloatNarrowerWidth]
        {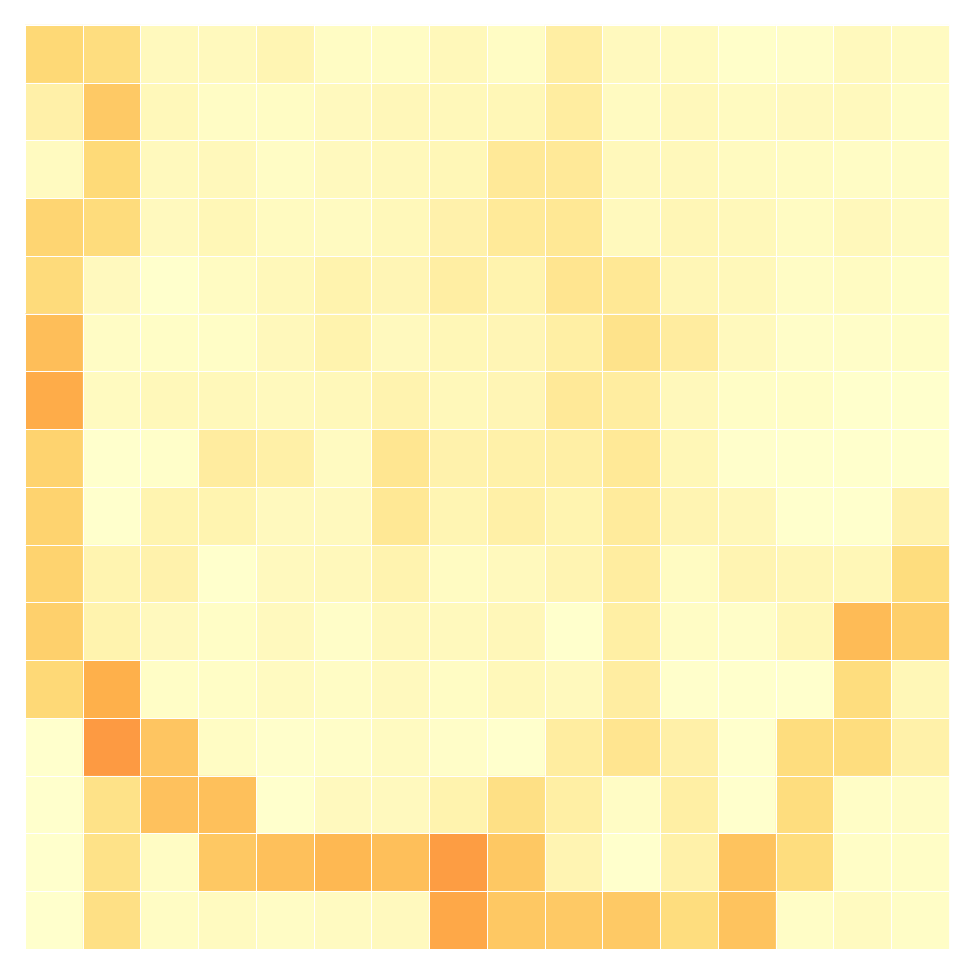}
    }
    \subfloat [Lux., level 2\label{fig:lux:left_L2}] {    
        \includegraphics[height=1.75cm, width=\subfloatNarrowerWidth]
            {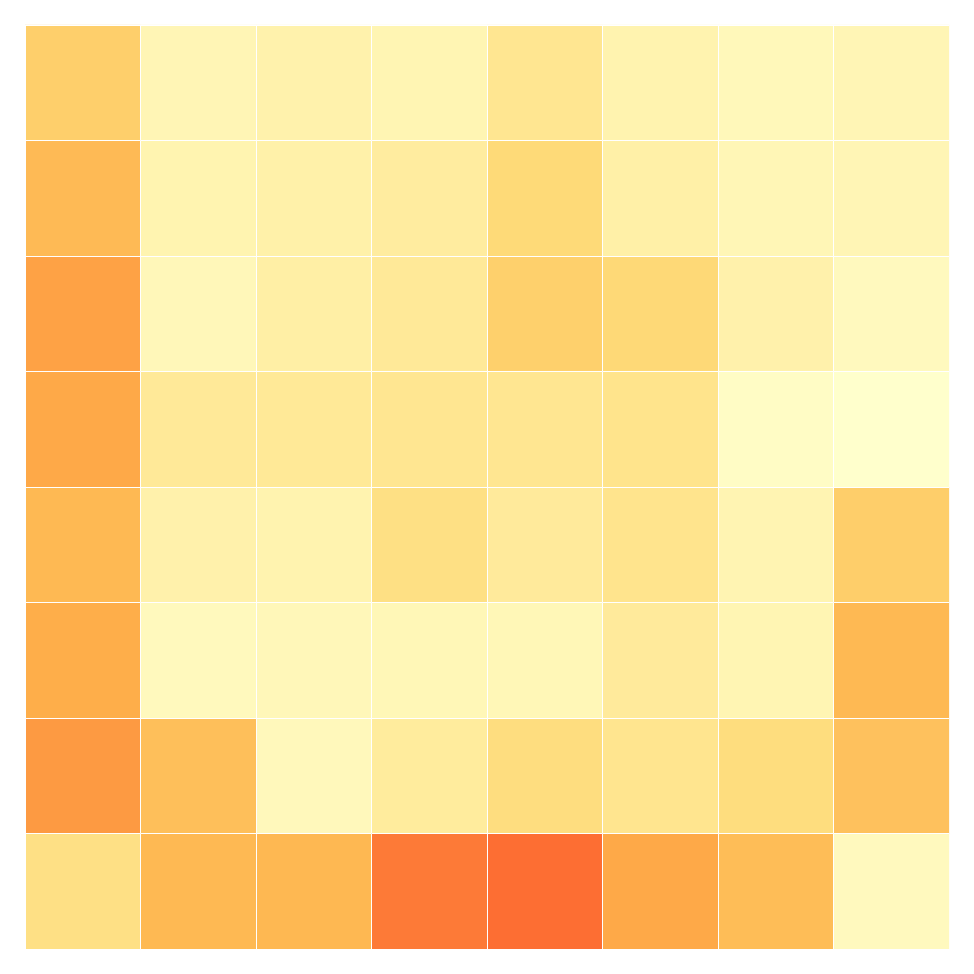}
        }
    \subfloat [Lux., level 3\label{fig:lux:left_L3}] {    
        \includegraphics[height=1.75cm, width=\subfloatNarrowerWidth]
            {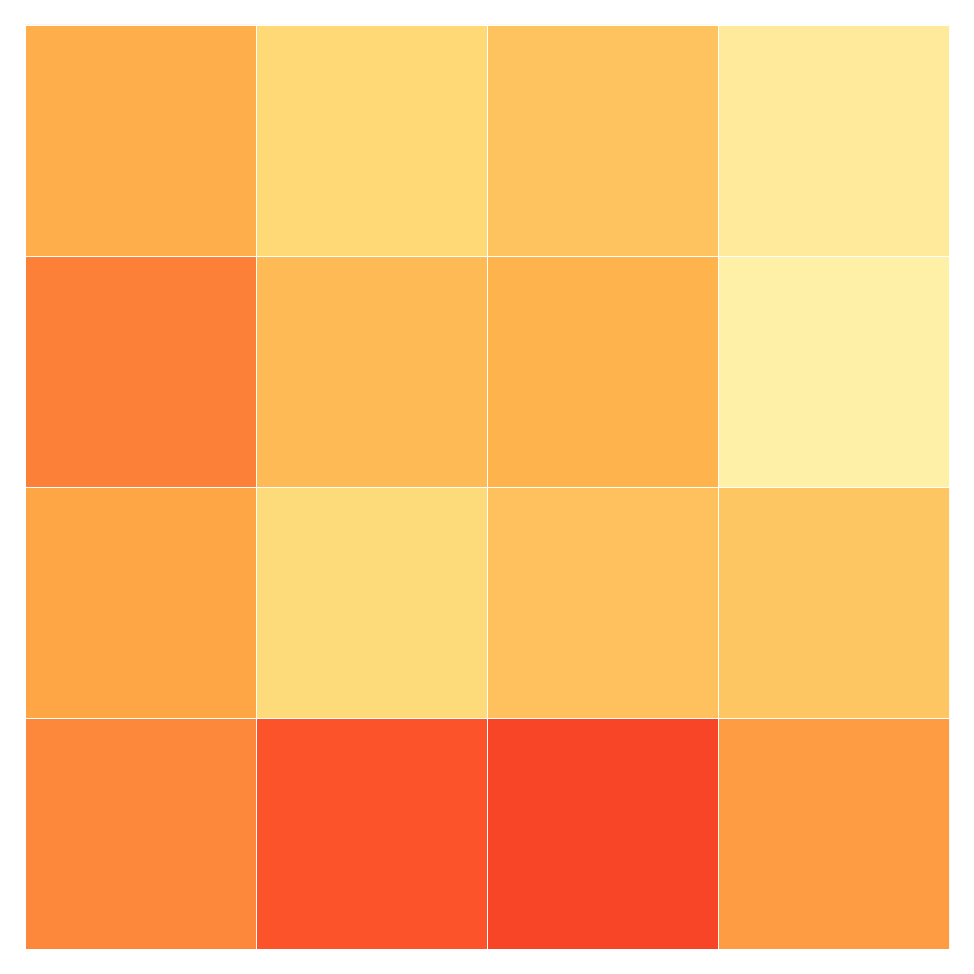}
    }
    \subfloat [Lux., level 4\label{fig:lux:left_L4}] {    
        \includegraphics[height=1.75cm, width=\subfloatNarrowerWidth]
            {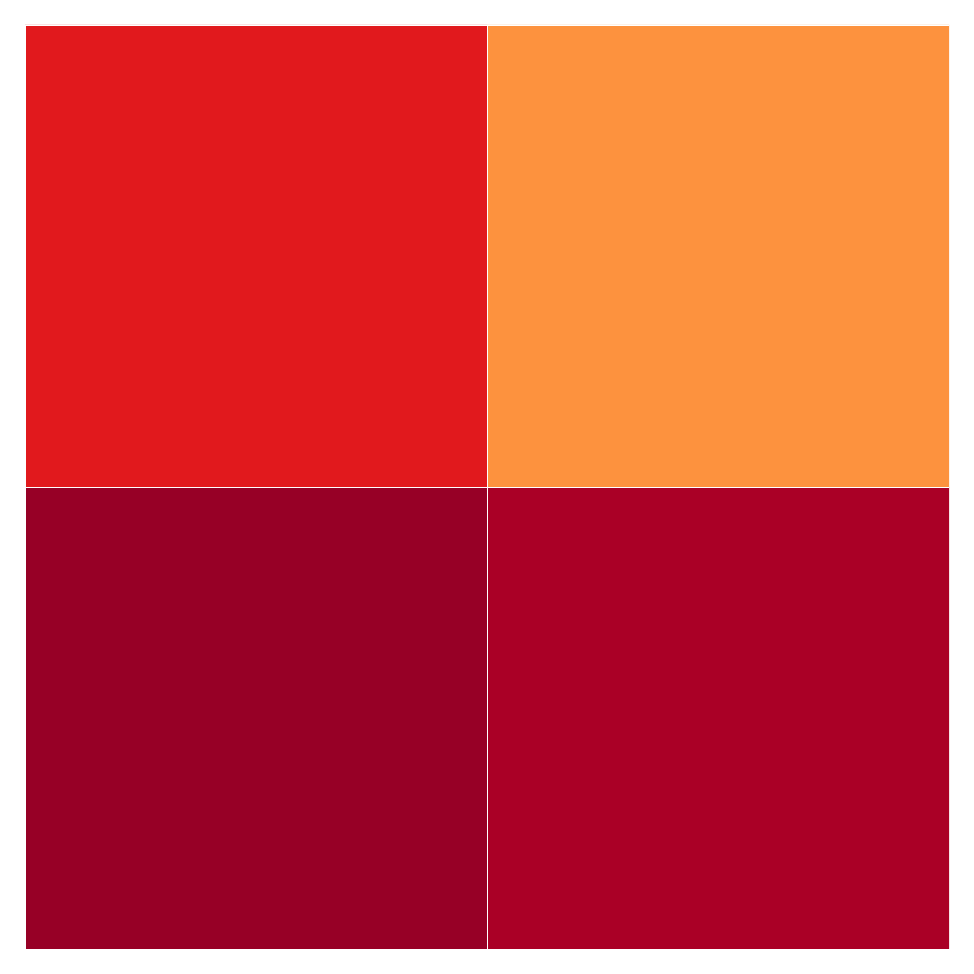}
    }
    \subfloat{    
        \includegraphics[height=1.7cm]{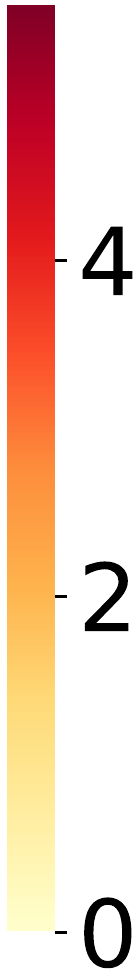}
    }
    
    \vspace{0.3cm}
    \subfloat [\label{fig:monaco:left_L1}Monaco,\\ level 1] {    
        \includegraphics[height=0.75cm, width=\subfloatNarrowerWidth]
            {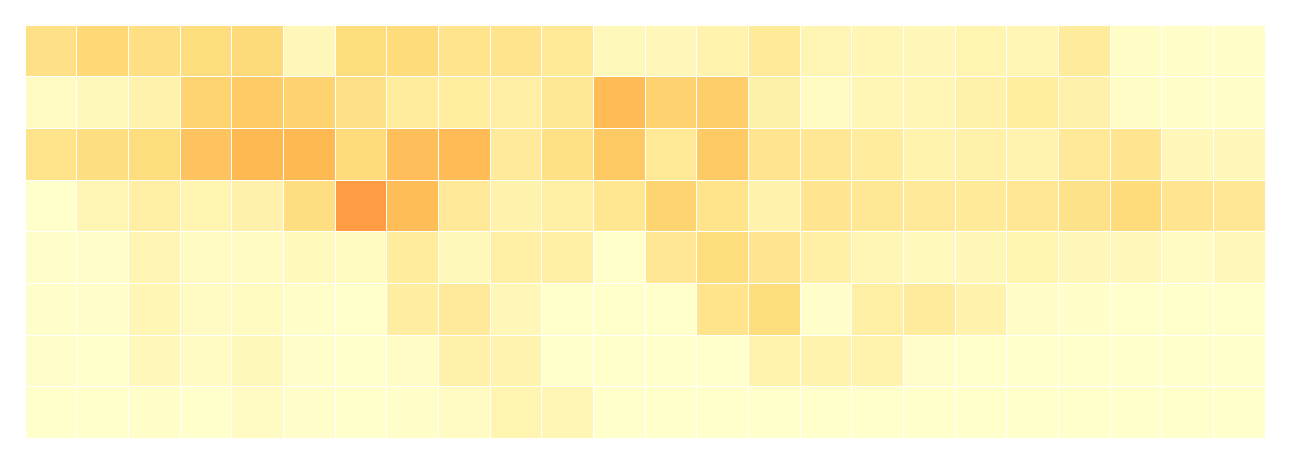}
    }
    \subfloat [\label{fig:monaco:left_L2}Monaco,\\ level 2] {    
        \includegraphics[height=0.75cm, width=\subfloatNarrowerWidth]
            {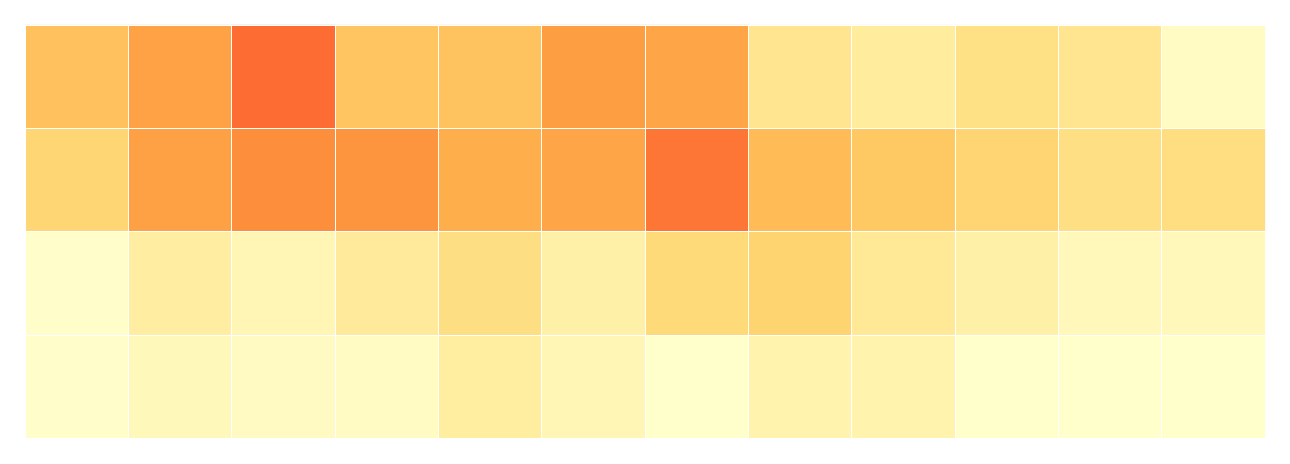}
    }
    \subfloat [\label{fig:monaco:left_L3}Monaco,\\ level 3] {    
        \includegraphics[height=0.75cm, width=\subfloatNarrowerWidth]
            {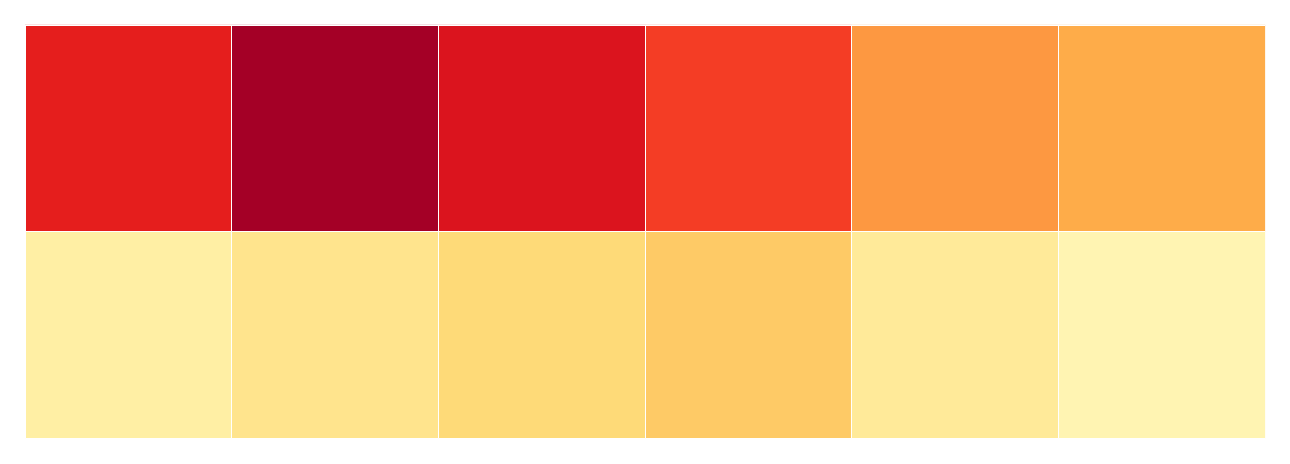}
    }
    \subfloat [\label{fig:monaco:left_L4}Monaco,\\ level 4] {    
        \includegraphics[height=0.75cm, width=\subfloatNarrowerWidth]
            {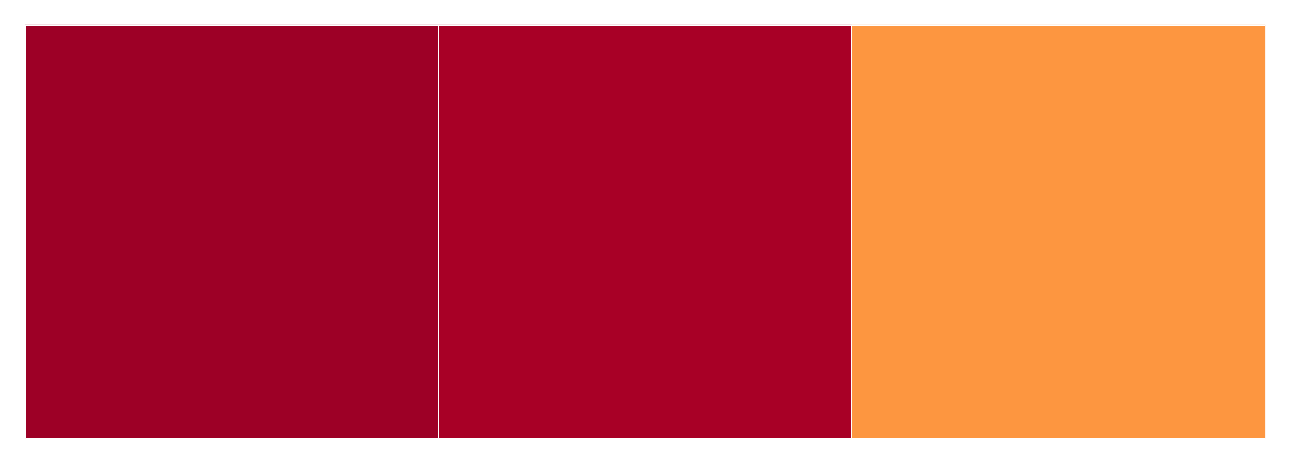}
    }
    \subfloat {    
    \includegraphics[height=0.71cm]{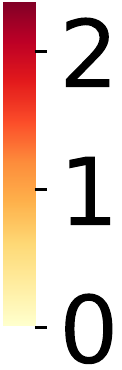}
    }
    \caption{Average number of vehicles leaving a rectangle every second during the 7:30-8:30 am interval in Luxembourg (top) and Monaco (bottom).
 }\label{fig:left}
\end{figure}

\captionsetup[subfloat]{labelformat=parens}

\textbf{Services and service chains.}
We consider two types of  time-critical automotive safety services,  one requiring a maximum delay of 10~ms (e.g., collision avoidance \cite{ca}) and the other a maximum delay of 100~ms (e.g., see-through \cite{see-through}). 
{For simplicity of description, we hereinafter refer to the services with the tighter delay constraints of 10~ms as {\em RT} (real-time) services, and we refer to the corresponding chains as {\em RT chains}.} As reported in \cite{ca}, a service chain consists of 3~VMs. Also,  
for each chain $\hv^u$, $\theta^u_1 \lambda^u_1 = \theta^u_3 \lambda^u_3 = 200$~MHz, and $\theta^u_2 \lambda^u_2 = 1$~GHz, so as to reflect a chain with a front-end and back-end VM with low computation load and a central VM with high computation load.
{For simplicity, we assume $\gamma_k^u\theta_k^u$ to be a constant.}
Upon entering the considered geographical area, each vehicle requests one of the two services at random, with some probability, to be defined later.

\textbf{Cost parameters.} 
We choose the cost values so that all the three components of our objective function  in~\eqref{Eq:def_obj_func} -- namely, computation cost,  bandwidth cost, and migration cost -- are no more than an order of magnitude apart, and none of them becomes negligible. 
The cost  of 100~MHz of CPU at level $\ell$  is $\chi^c= 2^{5-\ell}$ cost units, to reflect the decrease in computation costs when moving from the edge to the cloud~\cite{tong2016, SFC_mig}.
The bandwidth cost is $\chi^b=3$ cost units, while the migration cost is $\chi^m=600$ cost units.

{Note that by this choice of parameters, we have
the computation cost of a chain greater than $14 \cdot 2^{5-\ell}$ (due to our choice of $\theta \cdot \lambda$ above), which ranges between 14 and 448, depending on the level. Furthermore,
 the bidirectional bandwidth cost of a chain placed on level $\ell$ is $3 \cdot 2 \cdot \ell$, which is between 0 and 30, depending on the level.
} 

\textbf{Benchmark algorithms.}
We set the default decision period to $T = 1$\,s, and assume that there exists a mobility prediction scheme that associates each vehicle  with the closest PoA in the next decision period. 
Based on this prediction, we run GFA
to get the CPU allocation $\muv$ for each  chain, as detailed in Sec.~\ref{sec:alloc}. 
This allows us to fix the CPU allocations and compare our approach to existing solutions for the placement problem. 
Notably, if running placement without GFA, alternative deployment solutions (not dealing with CPU allocation and unaware of all the constraints in the \migProb) will get infeasible solutions most of the time.  
For our comparison, we consider the following benchmarks:

{\em Lower-bound\,(LBound):}
{Given the allocation, we use the ILP discussed in Sec.~\ref{sec:roadmap_placement} and solve the linear relaxation of this problem. 
This provides a lower bound on the cost of any feasible solution for the placement problem.
}%
In contrast to our algorithm \algtop, the fractional solution may place parts of the same chain (or even ``fractions of VMs'') on distinct datacenters. Furthermore, the LP formulation considers at each iteration all the chains in the system (not only the critical or newly arriving chains), thus significantly increasing the possible solution space.
Hence, LBound provides a lower bound on the minimal cost for the placement problem.

{\em First-fit (\ffit):} This scheme places each chain on the first delay-feasible datacenter with sufficient available capacity on the path from the root to the chain request's PoA. 

{\em CPVNF~\cite{CPVNF_proactive_place_in_CDN}:} This algorithm orders the critical and newly-arriving chains in a non-increasing order of the CPU capacity they require, if placed on an edge datacenter. It then places each such chain on the feasible available datacenter incurring the lowest cost according to~\eqref{Eq:def_obj_func}. This benchmark is an adaptation to our problem of the CPVNF algorithm~\cite{CPVNF_proactive_place_in_CDN}, which was used as a benchmark also in~\cite{VNF_place_in_public_cloud}. 

\textbf{Feasibility.}
Our \algtop\ algorithm,  as well as \ffit\ and \cpvnf, first considers only critical and newly arriving chains; if merely (re)placing these chains does not yield a feasible solution, the algorithm considers placing from scratch {\em  all} the chains in the system. Thus, referring to the framework presented in Fig.~\ref{fig:toplevel}, the benchmark algorithms (\ffit\ and \cpvnf) are used instead of \bu\ and \pushUp.
{Note that our comparison methodology 
over-estimates the performance of  our benchmark algorithms, as we give them 
the same (optimal) solution for the CPU allocation problem (found by GFA) ``for free".}

\textbf{Simulation methodology.}
The performance of LBound is deterministic. However, the performance of \ffit, \cpvnf\ and \algtop\ depends upon the arbitrary order of handling requests to which the algorithm in question gives the same priority. Hence, in each experiment we run each of these three algorithms 20 times, considering a random order of handling the requests.

The simulator has been developed in  Python and it is fully available on~\cite{SFC_mig_Github}.

\subsection{Resources required for finding a feasible solution}\label{sec:sim:feasibility}
As the first  step, we study the resource augmentation required by each algorithm to find a feasible solution. 
We focus on ten minutes of the trace, referring to a busy morning rush hour (08:20-08:30), with 6,859 and 9,351 distinct vehicles passing in the simulated area in Luxembourg and Monaco, respectively. 
We vary the ratio of RT chains, i.e., corresponding to requests with a maximum delay of 10\,ms. 
For each setting, we use binary search to find the minimum amount of resources (captured by the minimum CPU at the leaf datacenters) required by the considered algorithm  to find a feasible solution for {\em every} 1-second slot along the 10-minute trace.

Fig.~\ref{fig:cpu_vs_RT_prob} shows the results of this experiment. 
The amount of processing capacity required by \algtop\ is extremely close to the lower bound (LBound) in the Luxembourg scenario (Fig.~\ref{fig:cpu_vs_RT_prob_Lux}), and perfectly matches it in the Monaco one (Fig.~\ref{fig:cpu_vs_RT_prob_Monaco}). In contrast, the processing capacity required by \cpvnf\ and \ffit\ for finding a feasible solution is much higher: in Luxembourg's scenario, \cpvnf\ and \ffit\ typically need a processing capacity that is 50\%-100\% higher than the capacity required by \algtop.

As expected, the amount of resources required for obtaining a feasible solution consistently increases for a larger fraction of RT chains.  This happens because tighter timing constraints may require  allocating more CPU resources for each chain, to decrease the computational delay. Further, RT service requirements also dictate placing the chain close to the edge datacenter, thus reducing the use of processing capacities at higher-level datacenters.
However, in the Luxembourg scenario, the processing capacity required by \algtop\ with 100\% of RT chains is still lower than that required by \cpvnf\ and \ffit\ when no RT chain is present. 
Finally, comparing the Luxembourg  and Monaco scenarios, we note that the higher car density in Monaco,  and the smaller number of leaf-datacenters in that scenario (only 231 in Monaco, compared to 1,524 in Luxembourg), dictate using higher computation capacity in the leaf datacenters for finding a feasible solution. 

\subsection{Cost comparison}\label{sec:sim:cost}
 We now consider 30\% of RT chains and compare the costs obtained by different algorithms. Similarly to Sec.~\ref{sec:sim:feasibility}, here we also
 focus on ten minutes of the trace, referring to a busy morning rush hour (08:20-08:30). 

\begin{table}[tb!]
    \centering
{\scriptsize
    \caption{\label{tbl:cost_vs_rsrc}{Normalized cost of chain deployment and migration  vs.\  resource augmentation. The costs are normalized with respect to the cost obtained by the lower bound (LBound). An infinite cost indicates that the algorithm cannot find a feasible solution}} 
    \begin{tabular}{|c|c|c|c|}
        \hline
        \multicolumn{4}{|c|}{Luxembourg}\\
        \hline
        \multirow{2}{*}{$C_{cpu} / \hat{C}_{cpu}$} &
        \multicolumn{3}{c|}{Normalized Cost}\\ 
        \cline{2-4}
                & BUPU          & F-Fit         & CPVNF \rule{0pt}{2ex} \\ \hline
        1.00	& $\infty$	  	& $\infty$	 & $\infty$	 \\ \hline 
        1.06	& 1.76 			& $\infty$	 & $\infty$	 \\ \hline 
        1.50	& 1.91 			& $\infty$	 & $\infty$	 \\ \hline 
        2.00	& 1.17	 		& $\infty$	 & $\infty$	 \\ \hline 
        2.35	& 1.06	 		& 1.06	 	 & $\infty$	 \\ \hline 
        2.40	& 1.06	 		& 1.06	 	 & 1.06	  	 \\ \hline 
        2.50	& 1.05	 		& 1.05	 	 & 1.05	  	 \\ \hline 
      \multicolumn{4}{c}{}\\
      \hline
        \multicolumn{4}{|c|}{Monaco}\\
        \hline
        \multirow{2}{*}{$C_{cpu} / \hat{C}_{cpu}$} &
        \multicolumn{3}{c|}{Normalized Cost}\\ 
        \cline{2-4}
                & BUPU      & F-Fit      & CPVNF \rule{0pt}{2ex} \\ \hline
        1.000	& $\infty$	& $\infty$	 & $\infty$	 \\ \hline 
        1.002	& 1.36	 	& $\infty$	 & $\infty$	 \\ \hline 
        1.50	& 1.08	 	& $\infty$	 & $\infty$	 \\ \hline 
        1.58	& 1.09 		& 1.10		 & 1.10		 \\ \hline 
        2.00	& 1.05		& 1.05	 	 & 1.05		 \\ \hline 
        2.50	& 1.01		& 1.01	 	 & 1.01		 \\ \hline 
    \end{tabular}
}
\end{table}

\begin{table}[tb!]
    \centering
{\scriptsize    
    \caption{\label{tbl:migs_cost_vs_rsrc}{Normalized migration cost vs.\ resource augmentation}}
    \centering
    \begin{tabular}{|c|c|c|c|}
        \cline{2-4}
	    \multicolumn{1}{c|}{} &
        \multicolumn{3}{c|}{$C_{cpu} / \hat{C}_{cpu}$}\\  \cline{2-4}
	    \multicolumn{1}{c|}{} 
                   & 1.10   & 1.50    & 2.00    \\ \hline  
        Luxembourg & 639,609 & 556,739  & 84,138   \\ \hline
        Monaco     & 194,236 & 12,032   & 8,702    \\ \hline
    \end{tabular}
}
\end{table}
Tab.~\ref{tbl:cost_vs_rsrc} shows the average cost per second for each algorithm
{when varying the  resource augmentation factor, $C_\text{cpu}/\hat{C}_\text{cpu}$.
To make a meaningful comparison, the costs are normalized with respect to the cost obtained by LBound for the same amount of resources.
The table shows the normalized costs when the resource augmentation factor varies between 1 and 2.5.
}
In addition, for each scenario and  algorithm, the table presents the minimal amount of resource augmentation for which the considered algorithm finds a feasible solution with a confidence level of at least 99\%. When no feasible solution is found, the corresponding cost is infinite.

Interestingly, for a wide range of resource augmentation, {\em only} \algtop\ finds a feasible solution. 
When the resource augmentation is very small (e.g., only 100.2$\%$ of the resources required by LBound for finding a feasible solution in the Monaco scenario), the cost of \algtop's solution is significantly higher than that of \opt. However, when increasing the amount of resources, the gap between \algtop\ and \opt\ reduces very substantially. 
As for  \ffit\ and \cpvnf\, they need a significant resource augmentation for finding a feasible solution, e.g., 2.4 times  the processing capacity used by LBound in the Luxembourg scenario. 
Finally, for a high amount of resource augmentation (i.e., when resources are abundant), all strategies provide a solution of comparable cost.

{Tab.~\ref{tbl:migs_cost_vs_rsrc}, instead, captures the impact of the resource augmentation on the {\em migration cost} experienced by \algtop: the larger the resource augmentation, the lower the migration cost is.
Indeed, a tight resource budget forces \algtop\ to migrate many, even non-critical, chains to find a feasible solution.
On the other hand, high resource augmentation allows \algtop\ to find a feasible solution while placing more chains at a higher network level, thus reducing the current overall cost and mitigating the need for future migrations when the users move. 
}

\begin{figure}
    \centering
    \subfloat[\label{fig:cpu_vs_RT_prob_Lux}Luxembourg] {
        \includegraphics[height=3cm]{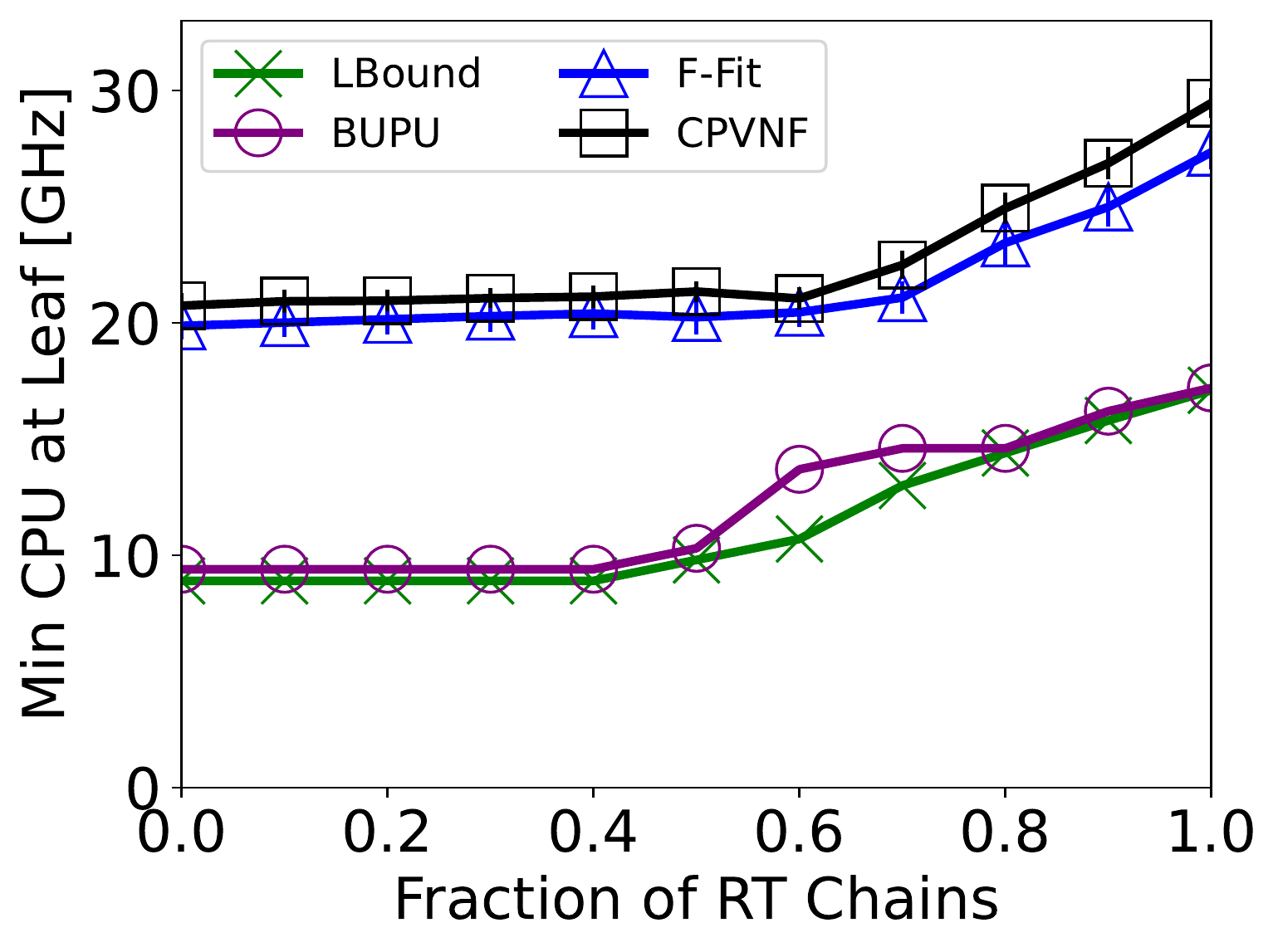}
    }
    \subfloat[\label{fig:cpu_vs_RT_prob_Monaco}Monaco] {
    \includegraphics[height=3cm]{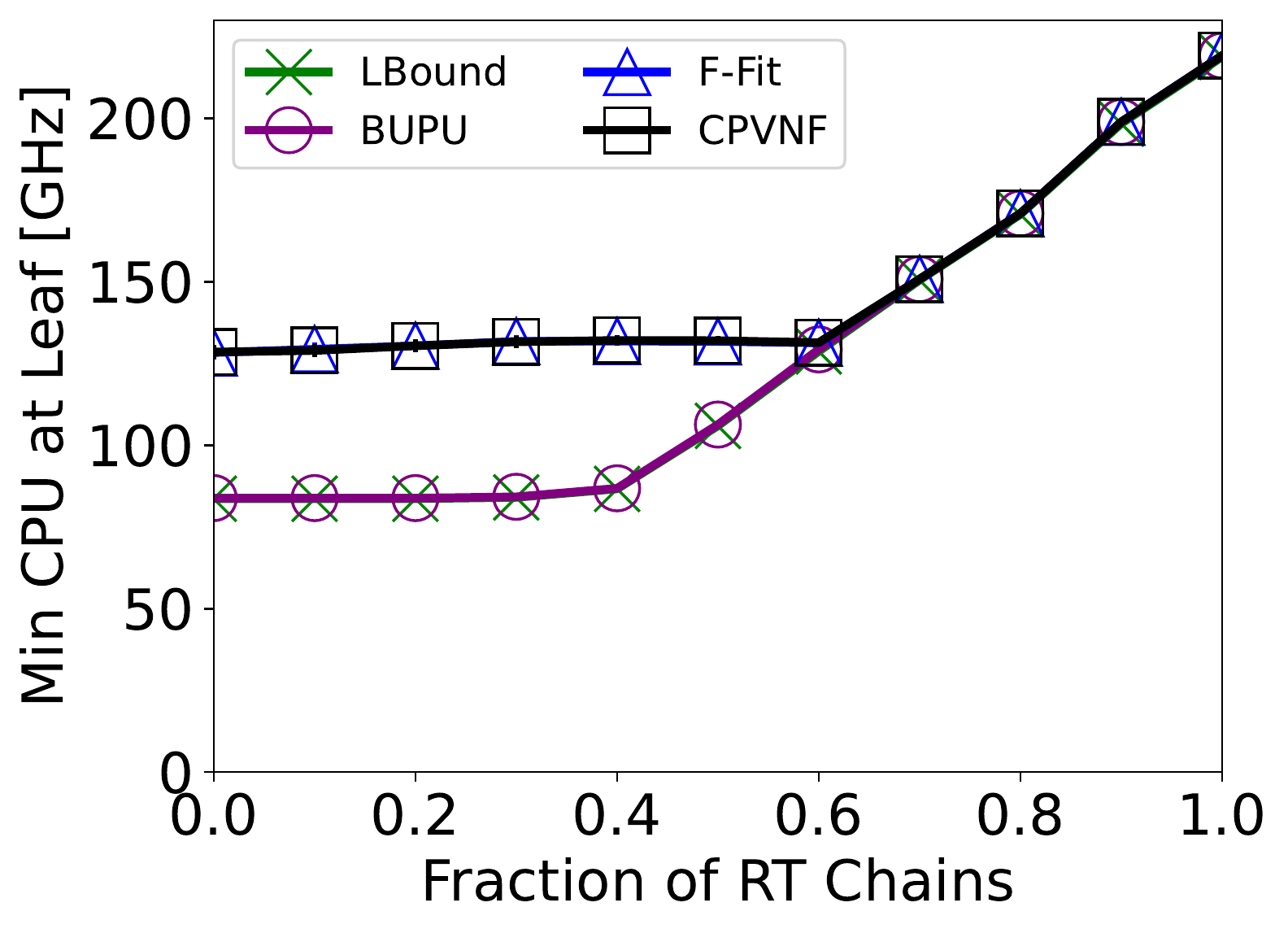}
    }
    \caption{Minimum required processing capacity for finding a feasible solution when varying the ratio of RT service requests.}
    \label{fig:cpu_vs_RT_prob}
\end{figure}

\subsection{Decision period}\label{sec:sim:decision_period}

Next, we focus on the impact of  decision period $T$ on the migration cost and on the SLA violation. To this end, we consider the 07:30-08:30 trace, and vary $T$.  The amount of CPU is set to 110\% the minimal capacity that \algtop\ needs for finding a feasible solution when $T=1$. This choice enforces tight capacity constraints, while being sufficient for allowing \algtop\ to find a feasible solution, even when varying  decision period $T$.

Fig.~\ref{fig:mig_cost_vs_T} depicts  
{the cost of migrating critical and non-critical chains,
}%
as $T$ varies. 
The overall migration cost is governed by  {\em non-compulsory migrations} (namely, migrations of non-critical chains) that the algorithm occasionally performs, as it cannot find a feasible solution otherwise.
In such a case the algorithm is forced to ``reshuffle'' the placement of possibly many chains, incurring a high overall migration cost (well beyond that imposed by critical chains alone).
Clearly, increasing $T$ leads to a lower cost of such non-compulsory migrations, as a smaller $T$ implies reducing the number of times the algorithm runs during the 1-hour trace, and in particular it reduces the number of such potential ``reshuffles''.
The migration of critical chains, instead, may be considered {\em compulsory}, and is  determined mainly by the user mobility, and not by decision period $T$.
Correspondingly, in the Monaco scenario, the cost of migrating critical chains hardly changes when varying $T$.
In the Luxembourg scenario, however, the cost of critical chains migration slightly decreases when increasing $T$. 

To better understand this phenomenon, consider Fig.~\ref{fig:mig_cost_vs_T_per_slot}, showing {the cost of migrating critical and non-critical chains, but now normalized per decision (i.e., per single run of \algtop). 
We consider that 
every second, some number $X$ of chains become critical. One can consider a {\em pessimistic} estimator, such that when using a decision period of $T$ seconds, the number of chains that become critical during the decision period is $X \cdot T$. 
In the Monaco scenario, the per-decision migration cost is indeed very close to our pessimistic estimator. This phenomenon can be attributed to the low mobility of vehicles in this trace, captured by the low average speed (recall Tab.~\ref{tab:sim_settings}).
However, when vehicles move faster, as in the Luxembourg scenario,
a sufficiently large  $T$ may translate to a single migration (per vehicle, per decision period), as opposed to several migrations of that vehicle when using a decision period of one second. This translates to a lower migration cost, compared to our pessimistic estimator. However, this comes at the cost of SLA violation. 

To study such violations incurred by increasing the decision period $T$, we considered the 
average violation time of critical chains. 
As expected, with a decision period $T$  an average SLA violation lasting roughly $T/2$ will be experienced. 
This is consistent with a model where each 
critical chain starts experiencing an SLA violation at a random time instant, uniformly distributed in $(0, T]$. 
We verified this behaviour in our evaluation for  both the Monaco and Luxembourg scenarios.}


Finally, we note that at each run, there is some probability that \algtop\ also migrates non-critical chains for finding a feasible solution, and this probability need not depend upon $T$. This indeed follows by considering the light-blue zone which is an almost constant additive increase of the migration cost, on top of that related to critical chains. 

\begin{figure}
    \centering
    \subfloat[Cost of migrating critical and non-critical chains during the 1-hour trace in Luxembourg (left) and in Monaco (right).\label{fig:mig_cost_vs_T}]{
        {\includegraphics[width=\subfloatWidth]{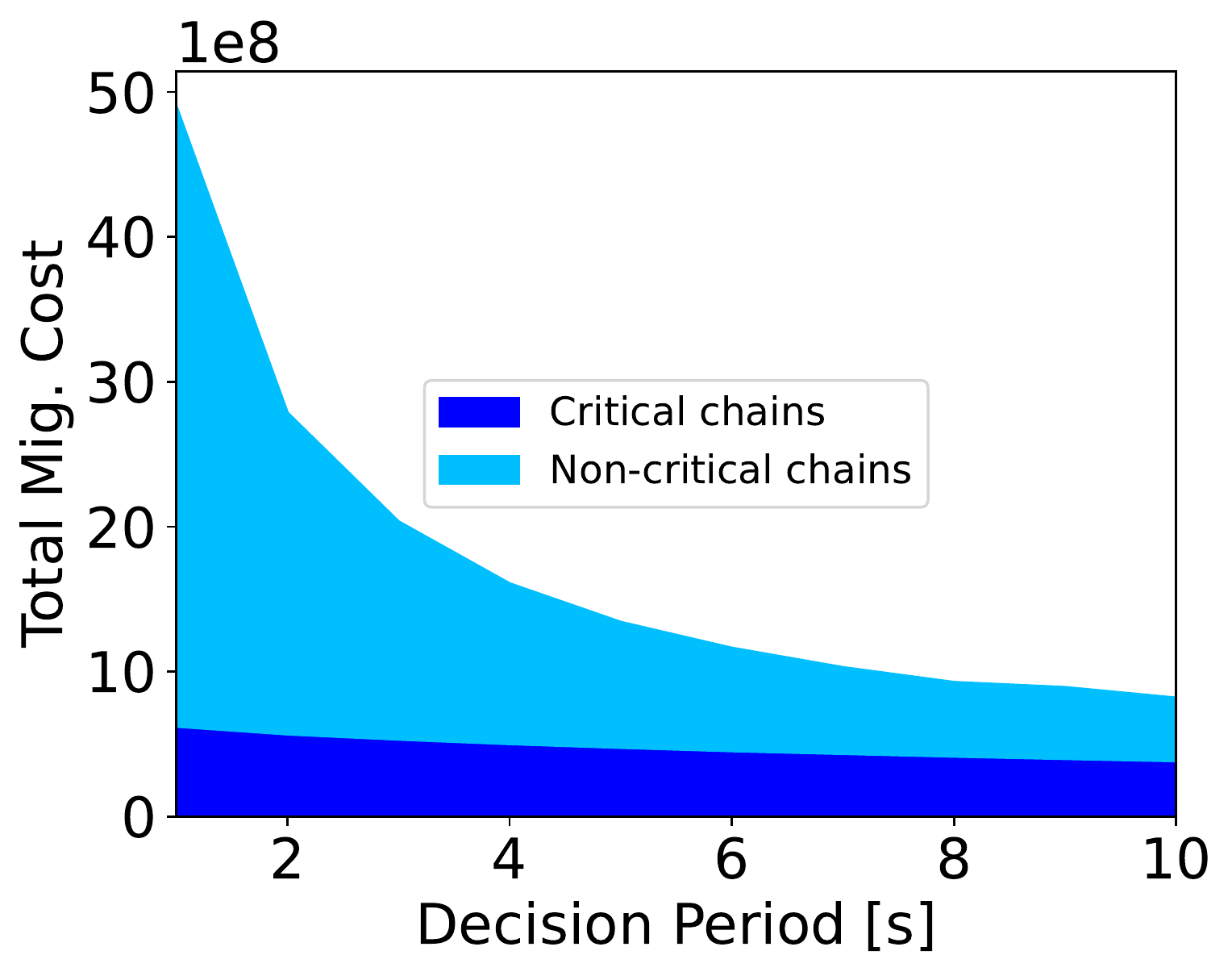}
        {
        \includegraphics[width=\subfloatWidth]{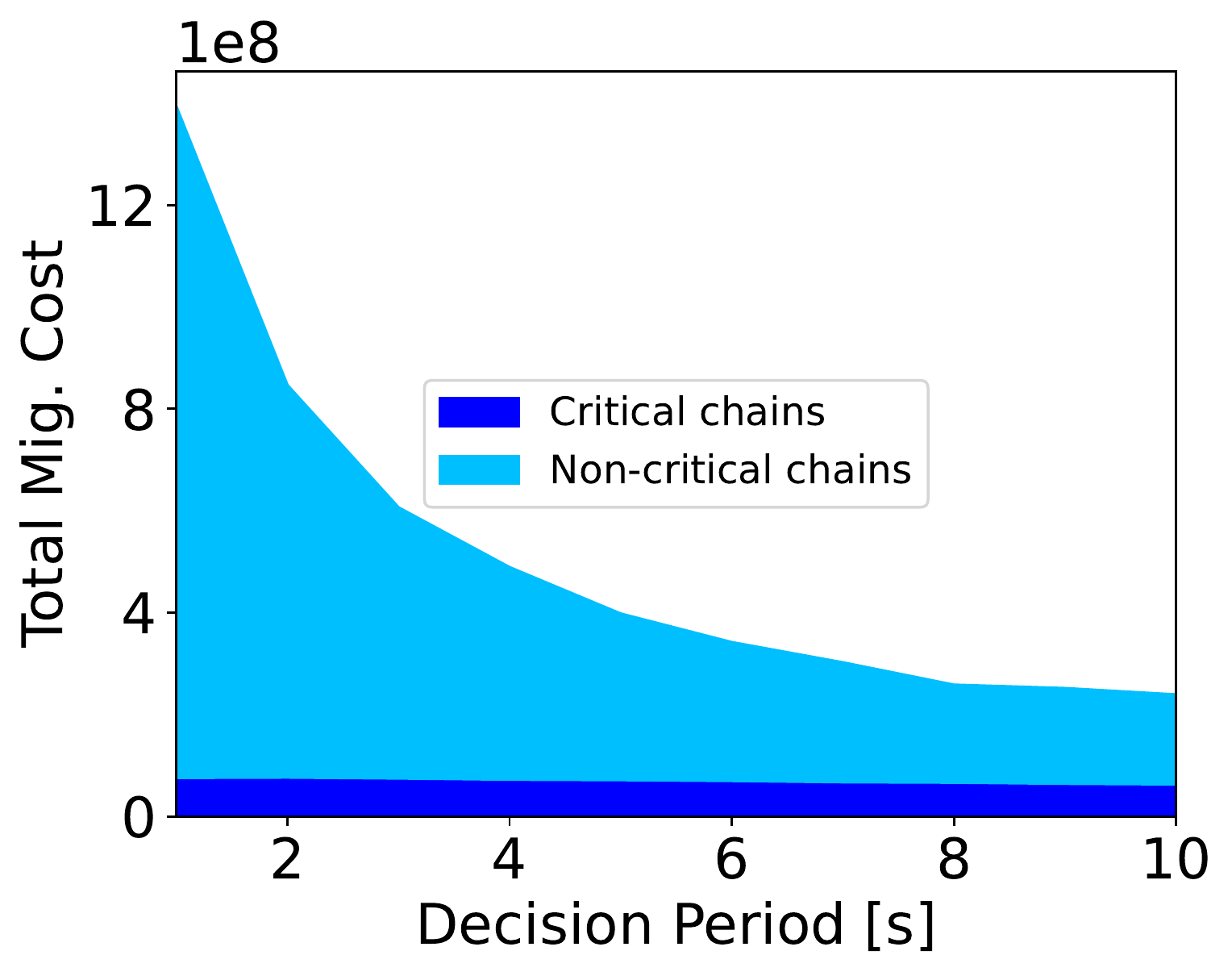}}
    }}

    \subfloat[The per-decision cost of migrating critical and non-critical chains in Luxembourg (left) and Monaco (right). The black line depicts our pessimistic estimator.\label{fig:mig_cost_vs_T_per_slot}] {
        {\includegraphics[width=\subfloatWidth]{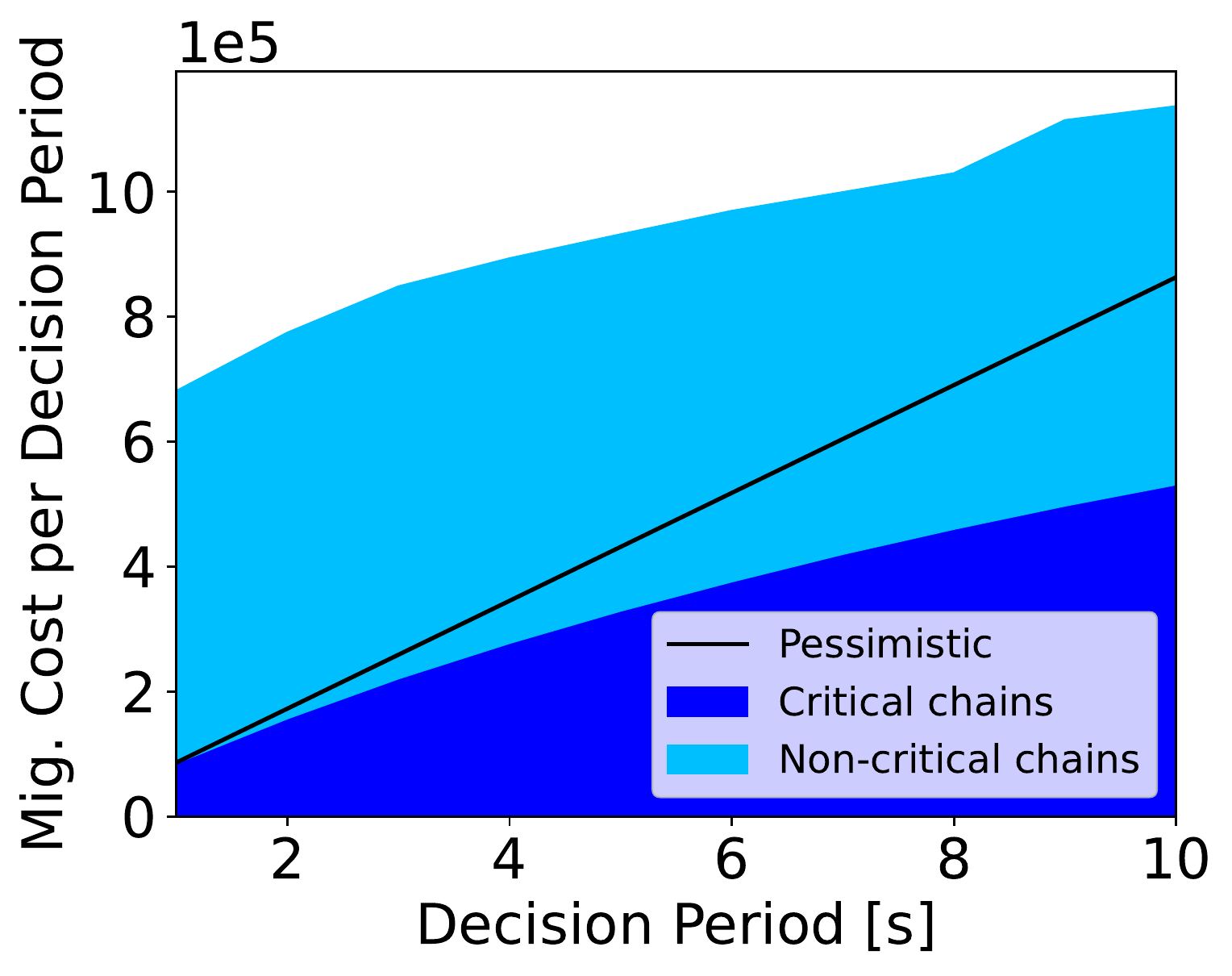} }
        {\includegraphics[width=\subfloatWidth]{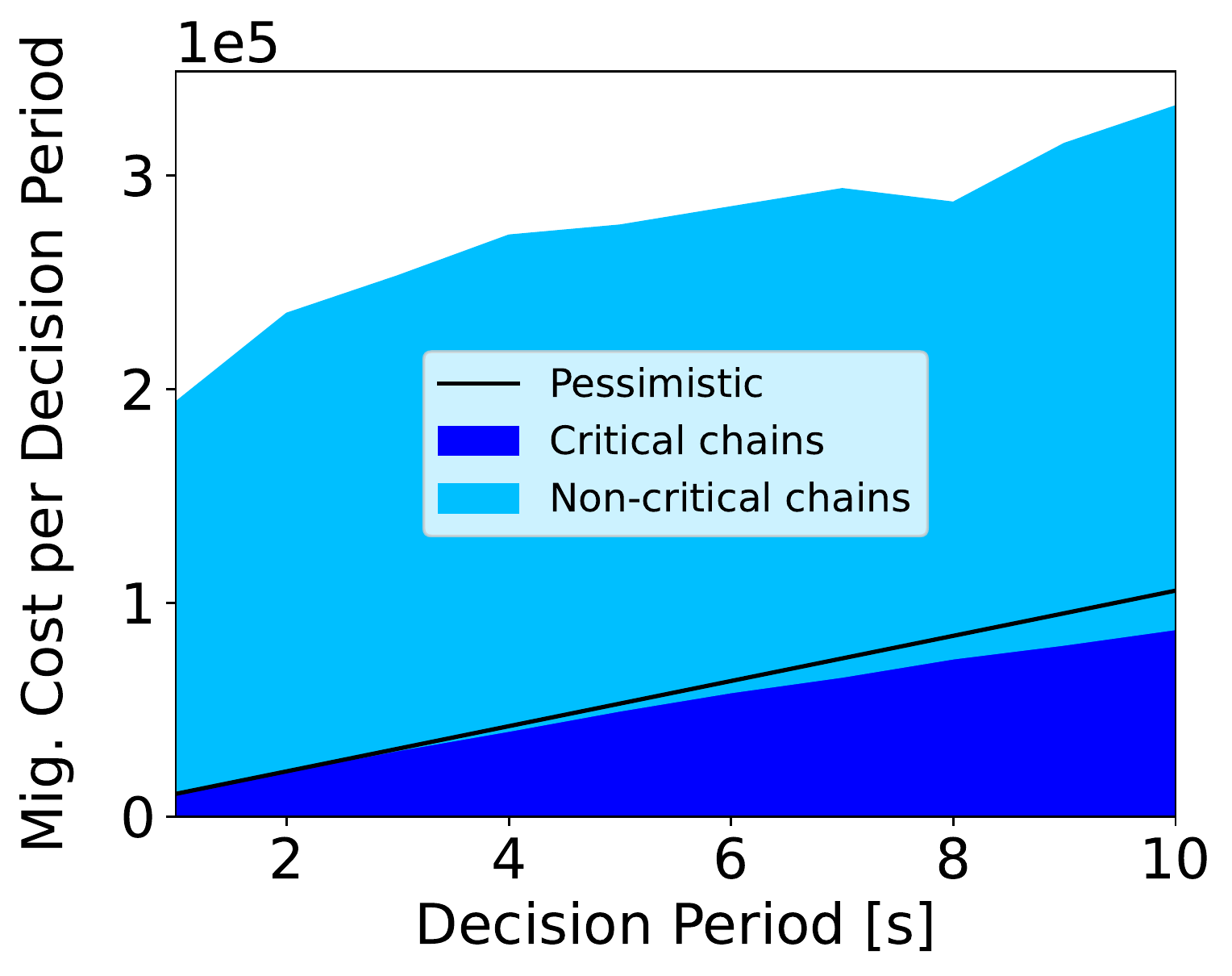}}
    }

    \caption{Impact of $T$ on migration cost.
    \label{fig:crit_n_mig_vs_T}}    
\end{figure}

\section{Related work}\label{sec:related}
Service migration has been extensively investigated in recent years.
The works~\cite{Dynamic_facility_location_Raz, Dynamic_facility_location_interval_graphs} address
    a dynamic VNF placement, where chains are migrated to serve mobile users. These studies disregard the computational delay, focusing on network delay solely. 
\cite{Dynamic_facility_location_Raz} designs efficient algorithms with strong performance guarantees -- e.g., $O(1)$ approximation ratio for the network cost, (i.e., the average user-datacenter distance), while using an $O(1)$ resource augmentation on the datacenter's capacities.
However,~\cite{Dynamic_facility_location_Raz} interprets the network delay as a part of the objective function, and not as a constraint. 
 Furthermore, \cite{Dynamic_facility_location_Raz} assumes that all current and future users' locations are known in advance. 
 
More realistic assumptions on the knowledge of users' locations are made in~\cite{FollowMe_J, Mig_modeling_containers_ML}, where the migration problem is modeled as a Markov decision process. 
Nonetheless, such an analysis relies on some assumptions about users' mobility (e.g., a stationary system with known, or predictable, transition probabilities), which do not conform with realistic scenarios.

Highly heterogeneous and unpredictable user mobility is considered in~\cite{CPVNF_proactive_place_in_CDN, MoveWithMe, mig_correlated_VMs, Companion_Fog}, which envision algorithms that consider various optimization criteria. 
CPVNF~\cite{CPVNF_proactive_place_in_CDN} is a greedy approach that we use as a benchmark and discuss in Sec.~\ref{Sec:sim_settings}.
\cite{MoveWithMe} decreases the migration overhead by clustering users, thus shrinking the amount of data migrated between datacenters. However, the model used in~\cite{MoveWithMe} assumes that the migration destination always has enough resources, and ignores computational delay. 
This model may conform with cloud computing where computational resources are abundant, and the primary source of delay is network delay, but not with edge computing, where the scarcity of resources and the tight delay constraint dictate considering the computational delay.~\cite{mig_correlated_VMs} considers multiple optimization criteria, i.e., the amount of resources consumed by migration, migration time, and service downtime. However,  the model used  in~\cite{mig_correlated_VMs} allows declining a migration request, or handling a request while breaking its target delay. 
In contrast, we address the problem of handling all the migration requests, while satisfying target delay constraints.
In a fog computing scenario, \cite{Companion_Fog} selects for each migration request a destination, based on the topological distance from the user, the availability of resources in the destination, and the data protection level in the destination.
However,~\cite{Companion_Fog} focuses on a selfish optimization of the migration destination for a single user, while we target finding a feasible global solution, while minimizing the overall system cost. 

Other works formulate the migration problem as an ILP
\cite{dynamic_sched_and_reconf_t, SFC_mig}, MILP~\cite{mig_or_reinstall}, or Mixed-Integer Quadratic Program ~\cite{Avatar}, but  
 none of them guarantees to find a feasible solution (possibly using more resources).
In particular, these solutions usually handle requests in parallel (the ILP/MIQP solvers), or in an arbitrary order, e.g., based on the request time~\cite{mig_correlated_VMs}. As a result, requests with relaxed delay constraints may be deployed on a scarce resource, possibly making the resource unavailable for tighter-delay applications, that cannot be deployed elsewhere. Our solution, on the contrary, orders requests based on their delay constraints, before handling them.
    
Some studies address predicting future users' mobility and traffic fluctuation~\cite{VNF_mig_w_pred, MultiScaler_Kassler, NFV_ego_learning}. However, these works differ from ours by both the objective function, and the tools used.
The works in~\cite{Dance_elephants_VM_mig, Linux_Containers_mig_18}  
decrease the migration overhead 
 (migration time, service downtime, and quantity of resources consumed by migration) by optimally deciding which data to migrate and in what order and pace. 
\cite{Multiple_mig_sched_Buyya} considers multiple simultaneous migration requests, and focuses on scheduling the  migrations and determining the bandwidth allocated to each migration process to minimize the total migration time and the service level objective violation.
\cite{dynamic_sched_and_reconf_t}
uses the optimal stopping theory to optimize the length of the decision period between subsequent runs of the migration algorithm. 

Some of the optimizations mentioned above are orthogonal to our work, and hence could be incorporated into our solution to boost performance.
Finally, implementation issues and performance overhead of VM live migration are discussed in~\cite{Mig_VNF_study}.

\section{Conclusions}\label{sec:conc}
We tackled mobile service provisioning in the edge-cloud continuum, envisioning algorithmic solutions with provable guarantees in terms of solution feasibility and resource usage. Besides addressing service chain deployment and resource allocation, our solution fulfills the service delay requirements and tackles service migration  by deciding which chains should be migrated and, if migrating, towards which datacenter. Our numerical results, derived in  large-scale, vehicular  scenarios, highlight interesting trade-offs and show that our approach may provide a feasible solution by using half the quantity of computing resources required by  state-of-the-art alternatives. We also show the robustness of the proposed approach, when varying the considered scenarios and the decision period of the algorithm.
\bibliographystyle{IEEEtran}
\bibliography{Refs}

\end{document}